\documentclass[11pt]{article}

\usepackage[margin=1in]{geometry}
\usepackage{graphicx} 
\usepackage{array} 
\usepackage{times}
\usepackage{cite}
\usepackage{mdframed}
\usepackage{amsmath,amsthm,amssymb,xspace,verbatim,multirow,bm,bbm,paralist,enumitem}
\usepackage{array,multicol,booktabs}
\usepackage{hyphenat,epsfig,subcaption}
\usepackage{nicefrac}
\usepackage[T1]{fontenc}
\usepackage[font=small, labelfont=bf]{caption}

\usepackage{thm-restate}

\usepackage[usenames,dvipsnames]{xcolor}
\usepackage[ruled]{algorithm2e}

\DeclareFontFamily{U}{mathx}{\hyphenchar\font45}
\DeclareFontShape{U}{mathx}{m}{n}{
      <5> <6> <7> <8> <9> <10>
      <10.95> <12> <14.4> <17.28> <20.74> <24.88>
      mathx10
      }{}
\DeclareSymbolFont{mathx}{U}{mathx}{m}{n}
\DeclareMathSymbol{\bigtimes}{1}{mathx}{"91}

\usepackage{tcolorbox}
\tcbuselibrary{skins,breakable}
\tcbset{enhanced jigsaw}

\usepackage[normalem]{ulem}
\usepackage[compact]{titlesec}

\definecolor{DarkRed}{rgb}{0.5,0.1,0.1}
\definecolor{DarkBlue}{rgb}{0.1,0.1,0.5}

\usepackage{nameref}
\definecolor{ForestGreen}{rgb}{0.1333,0.5451,0.1333}
\definecolor{Red}{rgb}{0.9,0,0}
\usepackage[linktocpage=true,
	pagebackref=true,colorlinks,
		urlcolor=black,
	linkcolor=DarkRed,citecolor=blue,
	bookmarks,bookmarksopen,bookmarksnumbered]
	{hyperref}
\usepackage[noabbrev,nameinlink]{cleveref}
\crefname{property}{property}{Property}
\creflabelformat{property}{(#1)#2#3}
\crefname{equation}{eq}{Eq}
\creflabelformat{equation}{(#1)#2#3}

\usepackage{bm}
\usepackage{url}
\usepackage{xspace}
\usepackage[mathscr]{euscript}

\usepackage{tikz}
\usetikzlibrary{arrows}
\usetikzlibrary{arrows.meta}
\usetikzlibrary{shapes}
\usetikzlibrary{backgrounds}
\usetikzlibrary{positioning}
\usetikzlibrary{decorations.markings}
\usetikzlibrary{patterns}
\usetikzlibrary{calc}
\usetikzlibrary{fit}
\usetikzlibrary{decorations.pathmorphing}

\usepackage[noend]{algpseudocode}
\makeatletter
\def\BState{\State\hskip-\ALG@thistlm}
\makeatother

\newcommand{\Ot}{\ensuremath{\widetilde{O}}}
\newcommand{\eps}{\ensuremath{\varepsilon}}

\newcommand{\bracket}[1]{\left[#1\right]}
\newcommand{\paren}[1]{\ensuremath{\left(#1\right)}\xspace}
\newcommand{\card}[1]{\left\vert{#1}\right\vert}

\newcommand{\tO}{\ensuremath{\widetilde{O}}}
\newcommand{\tTheta}{\tilde{\Theta}}
\newcommand{\tOmega}{\tilde{\Omega}}
\newcommand{\tomega}{\tilde{\omega}}

\newcommand{\prob}[1]{\Pr\left[#1\right]}
\newcommand{\expect}[1]{\Exp\bracket{#1}}

\newcommand{\set}[1]{\ensuremath{\left\{ #1 \right\}}}
\newcommand{\poly}{\rm poly}
\newcommand{\polylog}{\rm  polylog}
\newcommand{\ceil}[1]{\lceil{#1}\rceil}

\newcommand{\spidx}{\textsc{sp-idx}\xspace}
\newcommand{\eqidx}{\textsc{eq-idx}\xspace}
\newcommand{\idx}{\textsc{idx}\xspace}
\newcommand{\eq}{\textsc{eq}\xspace}
\newcommand{\xorconn}{\textsc{xor-conn}\xspace}
\newcommand{\xorbip}{\textsc{xor-bip}\xspace}
\newcommand{\distit}{\textsc{dist-item}\xspace}

\newcommand{\Rone}{R^{\to}}

\newcommand{\AS}{\textnormal{AS}}

\newcommand{\OMA}{\textnormal{MA}^{\to}}
\newcommand{\OMAnt}{\widehat{\textnormal{MA}}^{\to}}
\newcommand{\Rpub}{R^{\text{pub}}}
\newcommand{\Ronepub}{R^{\to \text{(pub)}}}
\newcommand{\hc}{\textsf{hcost}\xspace}
\newcommand{\vc}{\textsf{vcost}\xspace}
\newcommand{\tc}{\textsf{tcost}\xspace}
\newcommand{\cost}{\text{CC}\xspace}

\newcommand{\cS}{\mathcal{S}}
\newcommand{\Sp}{\ensuremath{\textnormal{Space}}\xspace}

\newcommand{\alg}{\ensuremath{\mathcal{A}}\xspace}
\newcommand{\out}{\ensuremath{\textsf{out}}}

\DeclareMathOperator*{\Exp}{\ensuremath{{\mathbb{E}}}}
\DeclareMathOperator*{\Prob}{\ensuremath{\textnormal{Pr}}}
\renewcommand{\Pr}{\Prob}

\newenvironment{tbox}{\begin{tcolorbox}[
		enlarge top by=5pt,
		enlarge bottom by=5pt,
		 breakable,
		 boxsep=0pt,
                  left=4pt,
                  right=4pt,
                  top=10pt,
                  arc=0pt,
                  boxrule=1pt,toprule=1pt,
                  colback=white
                  ]
	}
{\end{tcolorbox}}

\newcommand{\II}{\ensuremath{\mathbb{I}}}

\newcommand{\mireal}[1][]{
  \ifx\relax#1\relax%
    \II(\mione \,; \mitwo)%
  \else%
    \II(\mione \,; \mitwo\mid #1)%
  \fi
}

\newcommand{\cP}{\mathcal{P}}
\newcommand{\cH}{\mathcal{H}}

\newcommand{\cR}{\mathcal{R}}

\newcommand{\cX}{\mathcal{X}}
\newcommand{\cY}{\mathcal{Y}}

\newcommand{\freq}{\ensuremath{\textnormal{freq}}}

\newcommand{\msg}{\ensuremath{\textsf{msg}}}

\newcommand{\val}{\alpha}

\newcommand{\mypar}[1]{\medskip\noindent{\bfseries #1.}~}

\newtheorem{theorem}{Theorem}
\newtheorem{lemma}{Lemma}[section]

\newtheorem{proposition}[lemma]{Proposition}
\newtheorem{corollary}[lemma]{Corollary}
\newtheorem{claim}[lemma]{Claim}

\newtheorem{fact}[lemma]{Fact}

\newtheorem{problem}{Problem}

\newtheorem*{claim*}{Claim}
\newtheorem*{proposition*}{Proposition}
\newtheorem*{lemma*}{Lemma}
\newtheorem*{problem*}{Problem}

\crefname{lemma}{Lemma}{Lemmas}
\crefname{claim}{Claim}{Claims}

\newtheorem{mdresult}{Result}

\newtheorem*{mdresult*}{Main Result}

\newtheorem{definition}[lemma]{Definition}
\theoremstyle{definition}

\newtheorem{observation}[lemma]{Observation}

\newtheorem{mdinvariant}[lemma]{Lemma}

\theoremstyle{definition}
\newtheorem{mdalg}{Algorithm}
\newenvironment{Algorithm}{\begin{tbox}\begin{mdalg}}{\end{mdalg}\end{tbox}}

\allowdisplaybreaks

\renewcommand{\qed}{\nobreak \ifvmode \relax \else
      \ifdim\lastskip<1.5em \hskip-\lastskip
      \hskip1.5em plus0em minus0.5em \fi \nobreak
      \vrule height0.75em width0.5em depth0.25em\fi}

\title{New Lower Bounds in Merlin-Arthur Communication and Graph Streaming Verification}

\author{
Prantar Ghosh\footnote{(prantar.ghosh@gmail.com) Department of Computer Science, Georgetown University. Supported in part by NSF under award 1918989. Part of this work was done while the author was at DIMACS, Rutgers University, supported in part by a grant (820931) to DIMACS from the Simons Foundation.}
\and
Vihan Shah\footnote{(\href{mailto:vihan.shah@uwaterloo.ca}{\text{vihan.shah@uwaterloo.ca}}) 
Department of Computer Science, University of Waterloo. Part of this work was done while the author was at Rutgers University and was supported in part by an NSF CAREER Grant CCF-2047061. 
}
}

\date{}

\begin{document}
\maketitle

\pagenumbering{roman}

\begin{abstract}
We show new lower bounds in the \emph{Merlin-Arthur} (MA) communication model and the 
related \emph{annotated streaming} or stream verification model. The MA communication model 
is an enhancement of the classical communication model, where in addition to the usual 
players Alice and Bob, there is an all-powerful but untrusted player Merlin who knows their 
inputs and tries to convince them about the output. Most functions have MA protocols with total
communication significantly smaller than what 
would be needed without Merlin. We focus on the online MA (OMA) model, which is the MA 
analogue of one-way communication, and introduce the notion of \emph{non-trivial-OMA}
complexity of a function. This is the minimum total communication needed by any non-trivial 
OMA protocol computing that function, where a trivial OMA protocol is one where Alice sends 
Bob roughly as many bits as she would have sent without Merlin. We prove a lower bound on the non-trivial-OMA 
complexity of a natural function \emph{Equals-Index} (basically the well-known Index problem
on large domains) and identify it as a canonical problem for proving strong lower bounds on 
this complexity: reductions from it (i) reproduce and/or improve 
upon the lower bounds for all functions that were previously known to have large non-trivial-OMA complexity, (ii) exhibit the first explicit functions whose non-trivial-OMA complexity is superlinear, and even exponential, in their classical one-way 
complexity, and (iii) show functions on input size $n$ for which this complexity is as large as $n/\log n$. While exhibiting a function with $\omega(\sqrt{n})$ (standard) OMA complexity is a 
longstanding open problem, we did not even know of any function with $\omega(\sqrt{n})$ non-trivial-OMA complexity.  

Next, we turn to the annotated streaming model, the Prover-Verifier analogue of single-pass data streaming. We reduce from Equals-Index to establish strong lower bounds on the non-trivial complexity (for the analogous notion in this setting) of the fundamental streaming problem of counting distinct items, as well as of graph problems such as $k$-connectivity (both vertex and edge versions) in a certain edge update model that we call the \emph{support graph turnstile} (SGT) model. To set the benchmark space under which non-trivial annotated streaming schemes should solve these problems, we design classical streaming (sans Prover) algorithms for them under SGT streams by building \emph{strong} $\ell_0$-samplers that are robust to such streams and might be of independent interest. Finally, we exploit graph theoretic properties to design efficient schemes for $k$-connectivity on dynamic graph streams. This establishes a conceptual separation between classical streaming and annotated streaming: the former can handle certain turnstile and SGT streams almost as efficiently as dynamic streams, but the latter cannot.

\end{abstract}

\clearpage

%
\setcounter{tocdepth}{2}

\tableofcontents
\clearpage

\pagenumbering{arabic}
\setcounter{page}{1}

\section{Introduction}\label{sec:intro}

In the classical two-player communication model, the players Alice and Bob
each receive an input unknown to the other player and exchange as few bits
as possible to compute a function of the inputs. In their seminal paper on 
communication complexity classes, Babai, Frankl, and Simon~\cite{BabaiFS86} defined  the 
Merlin-Arthur (MA) communication model, where there is also a ``super-player'' Merlin, who 
knows the inputs of both Alice and Bob (collectively ``Arthur'') and 
hence, also knows the output. However, Merlin is not trusted, and so he 
provides Arthur a ``proof'' that should be thought of as a help message. After receiving this message, Alice and Bob communicate between themselves to verify its correctness.
It turns out that for most well-studied functions, the total communication needed between the players can be significantly reduced with the help of Merlin. This holds even in the most restrictive version of the MA model, where both Merlin and Alice send just a single message to Bob, from which he must verify the solution. This is known as the \emph{online Merlin-Arthur} (OMA) model. We focus on this setting and prove new lower bounds that also extend to a related stream-verification model called \emph{annotated streaming} \cite{ChakrabartiCM09, ChakrabartiCMT14}. We discuss the motivation and context of our results below.     

\subsection{Motivation and Context}
\mypar{Online Merlin-Arthur communication} The OMA complexity of a function is defined as the total communication between the players, i.e., the sum of the lengths of Merlin's and Alice's messages, in the optimal protocol that minimizes this sum. In a protocol, we denote Merlin's message length by \hc (shorthand for {\em help cost}) and Alice's message length by \vc (shorthand for {\em verification cost}), while the sum of their message lengths is termed \tc (short for {\em total cost}). Thus, the OMA complexity of a function is the minimum \textsf{tcost} over all possible OMA protocols computing the function. Naturally, a function with high OMA complexity, possibly close to the input size $N$, is deemed ``hard'' in the OMA model. Over the years, researchers have wondered what problems are hard in this model \cite{BabaiFS86, AaronsonW08, ChakrabartiCMT14, Gavinsky21}. Surprisingly, Aaronson and Wigderson \cite{AaronsonW08} showed that even the ``hardest'' functions in classical communication complexity, namely the disjointness and inner product functions, have OMA complexity only $O(\sqrt{N}\log N)$. Although the same authors showed (via a simple counting argument) the existence of functions with OMA complexity $\Omega(N)$, exhibiting an {\em explicit} function even with OMA complexity $\omega(\sqrt{N})$ remains a longstanding open problem. 

Another notion of ``hard functions'' in this setting might be ones that need large \hc in order to reduce the \vc from what Alice would have already needed to send without Merlin (i.e., in the classical one-way communication model). This implies high OMA complexity of these functions \emph{relative} to their classical one-way complexity (which might be much smaller than the input size, rendering them \emph{not} hard in the earlier sense).  Again, disjointness or inner product are not hard in this sense either: Aaronson and Wigderson's protocol shows that for these problems, \hc and \vc can be smoothly traded off while keeping the product $\tTheta(N)$. This means even a proof of size slightly larger than $\log N$ suffices to reduce Alice's message size to $o(N)$, whereas she would need to communicate $\Omega(N)$ in the classical setting. Further, prior work has shown that for many other hard functions in classical communication complexity, we can achieve the tradeoff $\hc\cdot \vc = O(N)$ \cite{ChakrabartiCMT14, ChakrabartiCGT14, ChakrabartiGT20}, ruling them out from being hard in the OMA model. Interestingly, there is a complementary general lower bound that says that an OMA protocol solving \emph{any} function $f$ must have $\hc\cdot \vc = \Omega(\Rone(f))$ \cite{BabaiFS86}, where $\Rone(f)$ is the one-way communication complexity of $f$. Observe that this lower bound implies that the OMA complexity $\OMA(f)$ of the function $f$ must be at least $\Omega(\sqrt{\Rone(f)})$. Stronger lower bounds are hardly known, with a couple of exceptions.\footnote{Indeed, these exceptions have $\Rone(f) = o(N)$ and they do not break the $\sqrt{N}$-barrier for OMA complexity.}     

Prior work has shown stronger lower bounds for the \emph{sparse} version of the fundamental index problem \cite{ChakrabartiCGT14}, and testing connectivity and bipartiteness of the \emph{XOR} of two graphs \cite{Thaler16}. All these functions have OMA complexity linear (or nearly linear) in their classical one-way complexity. This means the \emph{trivial} OMA protocol, where Merlin sends an empty string (or a junk message) while Alice and Bob (optimally) solve the problem on their own, is (nearly) as good as any other protocol for the problem. Further, since the lower bounds imply that an OMA protocol for any of these functions $f$ must have \tc $\tOmega(R(f))$, it follows that reducing the \vc to $o(\Rone(f))$ necessitates \hc to be $\tOmega(\Rone(f))$. Hence, these problems are hard in the second sense described above. Given this ray of hope on identifying hard OMA problems, we pursue this thread. In this work, we define the notion of \emph{non-trivial} OMA complexity of a function~$f$ as follows. Call an OMA protocol computing $f$ as {\em non-trivial} if its $\vc$ is $o(\Rone(f))$. Then, the non-trivial OMA complexity $\OMAnt(f)$ is the minimum \tc over all non-trivial protocols for $f$. 

Observe that for functions $f$ with standard OMA complexity $\OMA(f)=o(\Rone(f))$, we have $\OMAnt(f)=\OMA(f)$ (since the protocol minimizing \tc must be non-trivial); whereas if $\OMA(f)=\Omega(\Rone(f))$, then \hc dominates \vc in any non-trivial protocol for $f$, and $\OMAnt(f)$ essentially measures the minimum \hc required to achieve \vc = $o(\Rone(f))$. Thus, non-trivial OMA complexity formally captures the ``relative'' notion of hardness discussed above. Notably, in terms of the input size $N$, the best-known lower bound on $\OMAnt(f)$ is $\tOmega(\sqrt{N})$ (since a bound better than $\tOmega(\OMA(f))$ is unknown, and $\OMA(f)$ has a longstanding $\sqrt{N}$-barrier).

It is now natural to ask the following questions. 
\begin{itemize}
    \item {\it Can we find a single problem with high non-trivial OMA complexity that can be used to prove strong $\OMAnt$ lower bounds for multiple problems?} It would then serve as a \emph{canonical} hard function for OMA in this sense. 
    \item {\it Can we exhibit an explicit function $f$ with $\OMAnt(f)$ (strongly) superlinear in $\Rone(f)$?} This would mean that the trivial OMA protocol 
    is \emph{significantly} better than any non-trivial protocol in terms of total communication.
    \item {\it If the answer to the above question is ``yes'', how large can this gap be---can it be exponential?} 
    \item {\it Can we show an explicit function $f$ with $\OMAnt(f)$ (strongly) superlinear in $\sqrt{N}$ for input size $N$?} This would mean that while there is a strong $\sqrt{N}$-barrier for $\OMA$ complexity \cite{Gavinsky21}, the situation is not at all similar for $\OMAnt$.  
\end{itemize}

In this work, we answer all these questions in the affirmative. 
We discuss these results in detail in \Cref{sec:results}.

\mypar{Annotated Streaming} Next, we consider the analogous stream verification or \emph{annotated streaming} model \cite{ChakrabartiCM09, ChakrabartiCMT14}, where a space-restricted Verifier and an all-powerful Prover with unlimited space simultaneously receive a huge data stream. Following the input stream, the Prover (with knowledge of the entire stream) and Verifier (who could only store a summary) invoke a \emph{scheme} where the Prover tries to convince the Verifier about the answer to an underlying problem, similar to the online MA model. The total cost of a scheme is defined as the sum of the number of bits communicated by the Prover (\hc) and the number of bits of space used by the Verifier (\vc). The {\em non-trivial} annotated-streaming complexity of a function $f$ is analogously defined as the minimum total cost over all non-trivial schemes computing $f$, where a trivial scheme is one that uses as much space (up to polylogarithmic factors) as is needed in classical streaming (without Prover). 

Since all known lower bounds in annotated streaming are proven via reduction from problems in OMA communication, our knowledge of lower bounds in the two models are similar.
We use our OMA results to prove strong lower bounds on the non-trivial complexity in this model as well. We show that fundamental data streaming problems such as {\em counting distinct items} have high non-trivial annotated streaming complexity when frequencies can be huge. Further, we show that graph problems such as connectivity and more generally $k$-connectivity have high non-trivial complexity under certain graph streams that we call \emph{support graph turnstile} (SGT) streams. It might be intuitive that these problems might be hard in this model, even with a Prover. Perhaps surprisingly, we show that in the classical (sans Prover) model, we can solve these problems under SGT streams---featuring as large as exponential edge weights---almost as efficiently as under standard (unweighted) graph streams. We do this by building strong $\ell_0$-samplers that can handle large frequencies and might be of independent interest. These results set the stage for our lower bounds on non-trivial complexity: they provide the benchmark space for any non-trivial scheme solving these problems.  

    


Our final set of results give efficient annotated streaming schemes for $k$-connectivity on (standard) dynamic graph streams. We exploit graph theoretic properties on $k$-connected graphs to come up with short certificates for proving or disproving $k$-connectedness. This might be of independent interest in the graph-theoretic literature. Furthermore, these results establish a conceptual separation between classical streaming and annotated streaming: in the former, graph connectivity problems have roughly the same complexity in dynamic and SGT streams, whereas in the latter, they are much harder on SGT streams than on dynamic. We discuss our annotated streaming results in detail in \Cref{sec:results}.





%
\mypar{Basic Terminology} For the remainder of Section 1, it helps to define some basic terminology for ease of presentation. Later, in \Cref{sec:models}, we describe all notation and terminology in detail. An OMA protocol with \hc $O(h)$ (resp. $\Ot(h)$) and \vc $O(v)$ (resp. $\tO(v)$) is called an $(h,v)$-OMA-protocol (resp. $[h,v]$-OMA-protocol). Analogously, an annotated streaming protocol is called an $(h,v)$-scheme or an $[h,v]$-scheme. For computing a function $f$, a \emph{trivial OMA protocol} is one that has \vc $\Omega(\Rone(f))$, and a \emph{trivial scheme} is one that uses verification space $\tOmega(S(f))$, where $S(f)$ is the classical streaming complexity of $f$ (the asymptotically optimal space for computing $f$ in classical streaming). The ``Index'' communication problem and its variants come up frequently in our discussions. In the standard version, Alice has a string $x\in\{0,1\}^N$ and Bob has an index $j\in [N]$, where his goal is to output $x[j]$, the $j$th bit of $x$. 
\subsection{Our Results and Contributions}\label{sec:results}

First, we define the Equals-Index (henceforth, \eqidx) problem, which is the basis of our lower bounds. 

For arbitrary natural numbers $p$ and $q$, in the $\eqidx_{p,q}$ communication problem, Alice holds strings $x_1,\ldots,x_p$ where each $x_i\in \{0,1\}^q$. Bob holds a string $y\in\{0,1\}^q$ and an index $j\in [p]$. The goal is for Bob to output whether $x_j=y$.  This problem can be interpreted as (a boolean version of) the Index problem on large domains: Alice has a $p$-length string over the domain $\{0,\ldots,2^q -1\}$ (instead of just $\{0,1\}$). Bob needs to verify whether the $j$th index of Alice's string equals his value $y$. 

Our main result on the non-trivial-OMA complexity of $\eqidx$ is as follows. 

\begin{restatable}{theorem}{mainthmeqidx}\label{thm:eqidx-omant}
    For any $p,q$ with $p=\Omega(\log q)$, we have $\OMAnt(\eqidx_{p,q})=\omega(q)$
\end{restatable}

Observe that $q$ can be as large as $\exp(\Omega(p))$. We also show that $\Rone(\eqidx_{p,q})$ is only $\Theta(p+\log q)$ (\Cref{lem:onewayeqidx}). This immediately implies that the gap between $\OMAnt(f)$ and $\Rone(f)$ (or $\OMA(f))$ can be exponential.   

\begin{restatable}{corollary}{expgapcor}\label{cor:expgap}
    There is an explicit function $f$ with $$\OMAnt(f) = \exp(\Omega(\OMA(f))) \text{  and  }\OMAnt(f) = \exp(\Omega(\Rone(f))).$$
\end{restatable}

Conceptually, this means that the trivial protocol where Merlin sends nothing and Alice and Bob solve the problem on their own, is exponentially better (in terms of total communication) than any other protocol\footnote{Here, we are discounting clearly-suboptimal protocols where Alice sends $\omega(\Rone(f))$ bits or where Merlin sends a non-empty message despite Alice sending $\Omega(\Rone(f))$ bits. Hence, ``any other protocol'' essentially means any non-trivial protocol.}.

Recall that previously we did not know of any function $f$ on input size $N$ with $\OMAnt(f)=\tomega(\sqrt{N})$. While exhibiting a function $f$ with (standard) OMA complexity $\OMA(f)=\omega(\sqrt{N})$ remains a longstanding open problem, our results show that $\OMAnt$ complexity does not have such a barrier. In fact, \Cref{thm:eqidx-omant} implies that it can be as large as $N/\log N$. 

\begin{restatable}{corollary}{omegartncor}\label{cor:omegartn}
    For any $C\in (\sqrt{N},N/\log N)$, there is an explicit function $f$ on input size $N$ with $$\OMAnt(f) = \omega(C).$$
\end{restatable}

Next, we turn to the related \emph{Sparse Index} problem (henceforth \spidx). Chakrabarti et al.~\cite{ChakrabartiCGT14} defined the $\spidx_{m,N}$ problem as the version of Index where Alice's string is promised to have hamming weight (i.e., number of $1$'s) at most $m$. They identified $\spidx_{m,N}$ for $m=\log N$ as the first problem whose non-trivial OMA complexity is nearly as large as its one-way complexity.\footnote{It follows implicitly from their result establishing $\OMA(\spidx_{\log N,N})=\tOmega(\Rone(\spidx_{\log N,N}))$} Via reduction from $\eqidx$, we improve upon \cite{ChakrabartiCGT14}'s lower bound for $\spidx_{m,N}$ (for any $m$). In particular, our improved lower bound implies that $\spidx$ with sparsity $O(\log \log N)$ has non-trivial-OMA complexity \emph{exponential} in its classical one-way complexity. 

\begin{restatable}{corollary}{thmspidx}\label{cor:spidxomant}
For $m=\log \log N$, we have $\OMAnt(\spidx_{m, N})=\omega(\log N)$, whereas $\Rone(\spidx_{m, N})=\OMA(\spidx_{m, N})=\Theta(\log\log N)$.
\end{restatable}

We remark that en route to establishing the above result, we also improve upon \cite{ChakrabartiCGT14}'s upper bound on $\Rone(\spidx_{m,N})$ and settle its complexity in the classical one-way model (see \Cref{lem:onewayspidx}).   

The only other functions (to the best of our knowledge) that prior work has shown to have OMA complexity (nearly) linear in its one-way complexity are the $\xorconn_n$ and $\xorbip_n$ problems \cite{Thaler16}: in these problems, Alice and Bob have one graph each on the same vertex set $[n]$ (hence, the input size $N=\Theta(n^2)$), and they need to check connectivity and bipartiteness (respectively) of the graph obtained by XOR-ing their graphs, i.e., the graph induced by the symmetric difference of their edge sets. Thaler \cite{Thaler16} showed that each of these functions $f$ have $\OMA(f)=\Omega(n)=\tOmega(\Rone(f))$, which implies the same about its $\OMAnt$ complexity. We reduce from $\eqidx$ to reproduce this result, thus making a convincing case for \eqidx being a canonical problem for establishing high $\OMAnt$ complexity.  

\begin{restatable}{corollary}{xoromant}\label{cor:xoromant}
(\cite{Thaler16}, paraphrased)
    For $f\equiv \xorconn_n$ or $f\equiv \xorbip_n$, we have $$\OMAnt(f)=\Omega(\Rone(f))=\Omega(n).$$
\end{restatable}

We now turn to the annotated streaming model. Our OMA results can be used to prove lower bounds on the total cost of any non-trivial scheme for certain problems.  

First, consider the fundamental distinct items problem (henceforth, $\distit_{N,F}$) on turnstile streams, where frequencies of elements from universe $[N]$ get incremented and decremented, and we are promised that the absolute value of the max-frequency is bounded above by $F$. At the end of the stream, we need to output the number of elements with non-zero frequency. We show a separation between classical streaming and (non-trivial) annotated streaming complexities for this problem, given as follows.  

\begin{restatable}{theorem}{distitmain}\label{thm:distitgap}
    There is a setting of $F$ such that $\distit_{N, F}$ can be solved in $\tO(N)$ space in classical streaming, but any non-trivial annotated streaming scheme for the problem must have total cost $\Omega(N^{\polylog(N)})$.
\end{restatable}

Next, we show a similar separation for graph streaming problems. We define a support graph turnstile (SGT) stream to be one where an $n$-node graph is induced by the support of the edge-frequency vector. The parameter $\val$ of the SGT stream is the maximum possible absolute frequency.
For the (undirected) graph connectivity problem, where we need to check if all pairs of nodes are reachable from each other, and the $k$-vertex-connectivity (resp. $k$-edge-connectivity) problem, where we need to check whether removal of some $k-1$ vertices (resp. edges) disconnects the graph, we show an $\tO(n)$ vs $\Omega(n^{\polylog(n)})$ separation between classical streaming and non-trivial annotated streaming complexities of these problems under SGT streams. 

\begin{theorem}\label{thm:conn-sgt}
     Under certain support graph turnstile streams on $n$-node graphs, connectivity can be solved in $\tO(n)$ space by a classical streaming algorithm, whereas any non-trivial annotated streaming scheme for the problem must have total cost $\Omega(n^{\polylog(n)})$. 
\end{theorem}

\begin{theorem}\label{thm:kconn-sgt}
     Under certain support graph turnstile streams on $n$-node graphs, $k$-vertex-connectivity and $k$-edge-connectivity can be solved in $\tO(kn)$ space by a classical streaming algorithm, whereas any non-trivial annotated streaming scheme for the problem must have total cost $\Omega(n^{\polylog(n)})$. 
\end{theorem}

We remark that while the classical streaming space bounds of $\tO(n)$ for connectivity and $\tO(kn)$ for $k$-connectivity are known for standard dynamic graph streams \cite{AhnGM12, assadi2023tight} where edge multiplicities are $0$ or $1$ throughout the stream, it was not known whether the same can be achieved for the harder SGT streams. We establish these upper bounds by designing $\ell_0$-samplers that can handle frequency exponential in $n$ while incurring just polylogarithmic factors in space. These might be of independent interest (see \Cref{lem:L0}).

Finally, we design efficient schemes for $k$-vertex-connectivity and $k$-edge-connectivity under standard dynamic graphs streams. These schemes have total cost significantly smaller than the lower bound proven for schemes processing SGT streams. 

\begin{restatable}{theorem}{verconndyn}\label{thm:ver-conn-dyn}
    Under dynamic graph streams on $n$-node graphs, there exists a $[k \cdot (h+kn),v]$-scheme for $k$-vertex-connectivity for any $h,v$ such that $h\cdot v=n^2$. In particular, under such streams, the problem has non-trivial annotated streaming schemes with total cost $\tO(k^2n)$.  
\end{restatable}

\begin{restatable}{theorem}{econndyn}\label{thm:e-conn-dyn}
    Under dynamic graph streams on $n$-node graphs, there exists a $[k^2n+h,v]$-scheme for $k$-edge-connectivity for any $h,v$ such that $h\cdot v=n^2$. In particular, under such streams, the problem has non-trivial annotated streaming schemes with total cost $\tO(k^2n)$.  
\end{restatable}

\begin{restatable}{theorem}{edgeconndyn}\label{thm:edge-conn-dyn}
    Under dynamic graph streams on $n$-node graphs, there exists an $[n,n]$-scheme for $k$-edge-connectivity (for any $k$). 
    In particular, under such streams, the problem has non-trivial annotated streaming schemes with total cost $\tO(n)$.   
\end{restatable}

Contrast this with our results for annotated streaming under SGT streams, where the total cost for these problems can be as large as $\Omega(n^{\polylog(n)})$. Thus, we can see a conceptual separation between classical streaming and annotated streaming: the former can tolerate SGT streams incurring negligible factors in complexity over dynamic graph streams, whereas the latter incurs significantly large factors for SGT streams over dynamic. 



\subsection{Related Work}

In their seminal paper, Babai, Frankl, and Simon \cite{BabaiFS86} defined communication classes similar to classes in computational complexity; this included the Merlin-Arthur (MA) communication class, which essentially defined the Merlin-Arthur communication model. Klauck \cite{Klauck03} proved that disjointness and inner product have MA complexity $\Omega(\sqrt{N})$, which was surprisingly proven to be tight (up to logarithmic factors) by Aaronson and Wigderson \cite{AaronsonW08} who gave a protocol with total cost $O(\sqrt{N\log N})$. Chen\cite{Chen20} recently improved this bound to $O(\sqrt{N\log \log \log N})$. Chakrabarti et al.~\cite{ChakrabartiCMT14} were the first to consider the online version of the MA model. They used it to show lower bounds for annotated streaming schemes, as has been traditionally done by subsequent works. They proved that for any function $f$, an $[h,v]$-OMA-protocol that computes it must have $h\cdot v\geq \Rone(f)$. Further, for many problems including frequency moments and subset checks, they gave annotated streaming schemes (which imply OMA protocols with the same bounds), achieving this smooth tradeoff. \cite{ChakrabartiCGT14} were the first to exhibit a problem, namely sparse index, where such a tradeoff is not possible. In fact, they showed that for sparse index with sparsity logarithmic in the input size, the OMA complexity is as large as the one-way communication complexity. Later, Thaler \cite{Thaler16} exhibited two more problems with this property: XOR-connectivity and XOR-bipartiteness. The motivation was to show the existence of graph problems that provably need semi-streaming schemes (an $[n,n]$-scheme for $n$-node graphs) to solve and that they are equally hard in the annotated streaming model as in classical streaming. 

An ``augmented'' version of the Equals-Index was studied by Jayram and Woodruff \cite{JayramW13} in the classical communication model. They called it the ``Augmented Index problem on large domains'': here Bob also knows all the entries of Alice's vector before his input index. 

For the problem of computing distinct items, \cite{ChakrabartiCMT14} gave an $(n^{2/3}(\log n)^{4/3}, n^{2/3}(\log n)^{4/3})$-scheme, which was later simplified and improved to an $(n^{2/3}\log n, n^{2/3}\log n)$-scheme by \cite{Ghosh20}. Note that both of these works assume that the stream length $m$ is $O(N)$, where $N$ is the universe size. Our results are for streams that are exponentially longer. 

With the growing interest in graph streaming algorithms over the last couple of decades, much of the recent literature on stream verification has focused on graph problems \cite{CormodeMT13,ChakrabartiCGT14, AbdullahDRV16, Thaler16, ChakrabartiG19, ChakrabartiGT20}. Most of them design stream verification protocols for insert-only or insert-delete graph streams. For the graph connectivity problem on $n$-node graphs, \cite{ChakrabartiCMT14} gave an $[h,v]$-scheme for any $h,v$ with $h\geq n$ and $h\cdot v=n^2$. For sparse graphs with $m$ edges, \cite{ChakrabartiCGT14} designed an $[n+m/\sqrt{v}, v]$-scheme. As mentioned above, \cite{Thaler16} studied the problem in the XOR-edge-update model and proved that any $[h,v]$-scheme must have $(h+n)\cdot v\geq n^2$. A number of works \cite{CormodeMT13, ChakrabartiG19, ChakrabartiGT20} studied verification schemes for shortest-path and $s,t$-connectivity related problems. No prior work on stream verification, however, studied the $k$-connectivity problem.  

In the classical streaming model, \cite{AhnGM12} gave the first algorithm for $k$-edge-connectivity in dynamic streams using $\Ot(kn)$ space. 
\cite{guha2015vertex} gave the first algorithm for $k$-vertex-connectivity in dynamic graph streams, which was improved by a factor of $k$ and made nearly optimal by \cite{assadi2023tight} using $\Ot(kn)$ space.
\cite{sun2015tight} proved a lower bound of $\Omega(kn)$ bits for both problems, even for insertion-only streams (see also \cite{assadi2023tight} for extending the lower bound for $k$-vertex-connectivity to multiple passes).

 Other variants of stream verification include a {\em prescient} setting where Prover knows the entire stream upfront, i.e., before Verifies sees it, and can send help messages accordingly \cite{ChakrabartiCMT14, ChakrabartiCGT14}. Versions where Prover sends very large proofs have also been considered \cite{KlauckP13}. 
 Natural generalizations to allow multiple rounds of interaction between the Prover and Verifier have been investigated. These include {\em Arthur-Merlin streaming protocols} of Gur and Raz \cite{GurR13} and the {\em streaming interactive proofs} (SIP) of Cormode et al.~\cite{CormodeTY11}. The latter setting was further studied by multiple works \cite{AbdullahDRV16, ChakrabartiCMTV15, KlauckP14}. Very recently, the notion of \emph{streaming zero-knowledge proofs} has been explored \cite{CormodeDGH_streamzkp}. For a more detailed survey of this area, see \cite{Thaler-encyclopedia}. 

\subsection{Technical Overview}

\subsubsection{Communication Lower Bounds}

\mypar{The \eqidx lower bound} The basis of all our lower bounds is an
OMA lower bound on the $\eqidx$ problem. Recall the classical Index 
problem $\idx_N$ where Alice has a string $x$ of length $N$ and Bob has 
an index $j\in [N]$ such that he needs to know $x[j]$. Chakrabarti et 
al.~\cite{ChakrabartiCMT14} showed that for any $p,q$ with $p\cdot 
q=N$, we can get a $(q,p)$-OMA-protocol as follows. The string $x$ can 
be partitioned into $p$ chunks of length $q$ each. Merlin sends Bob the 
purported chunk where $j$ lies, thereby sending $q$ bits. Again, Alice 
sends Bob an $\Theta(1)$-size equality sketch for each chunk, thereby 
sending $\Theta(p)$ bits. Using the relevant sketch, Bob can figure out 
whether the chunk sent by Merlin is accurate and find the solution if 
it is. 

Note that from Alice's perspective, the problem actually boils down to 
the following subproblem: Alice has a string $x$ partitioned into 
chunks $x_1,\ldots,x_p$ and Bob has a string $y$; she needs to help him 
verify whether $y$ is identical to her $k^{th}$ chunk, where she 
doesn't know $k$. Our main observation is that in the above protocol, 
to solve this subproblem, she spends as many bits as she would have 
{\em without} Merlin: $\Theta(p)$ bits (we can show that this is tight 
for the problem under classical one-way communication). So now that 
Alice and Bob have Merlin, why not take his help and improve this 
communication to $o(p)$? If they could do this with at most $O(q)$ bits 
of help, then they would obtain an improved $(q,o(p))$-OMA
protocol for $\idx_{pq}$. But the known lower bound for $\idx_{pq}$ 
\cite{ChakrabartiCMT14} says that the product of \hc and \vc must be 
$\Omega(pq)$. Hence, Merlin cannot bring down the \vc to $o(p)$ without 
sending $\omega(q)$ bits. Since $\Theta(p)$, as noted, is the 
communication needed for the subproblem in the classical one-way model, we get that the non-trivial OMA complexity of the subproblem is $\omega(q)$. But $q$ can be  much larger than $p$---even 
exponential in $p$. Hence, we 
identify a problem whose non-trivial OMA complexity can be significantly larger 
than their one-way complexity. We essentially abstract out this 
subproblem as the $\eqidx_{p,q}$ problem and formalize the reduction.    

 \mypar{The sparse index lower bound} Although this lower bound follows by a reduction from $\eqidx$, we discuss it separately to point out the challenges in proving its non-trivial-OMA complexity. 
Recall that the $\spidx_{m,N}$ problem is the version of $\idx_N$ where Alice's string is promised to have at most $m$ 1's. We prove that $\OMAnt(\spidx_{\log\log N, N})$ is $\Omega(\log N)$, while $\Rone(\spidx_{\log\log N, N}) = \Theta(\log \log N)$, thus establishing yet another exponential gap. 

Let us discuss the classical one-way complexity first. Chakrabarti et al.~\cite{ChakrabartiCGT14} showed that for any~$m$, $\Rone(\spidx_{m,N})=O(m\log m+\log N)$, but it was not known to be tight. We improve it to a tight bound of $\Theta(m+\log\log N)$, and then setting $m=\log\log N$ gives the desired bound. This improvement is the most challenging part in this result. The protocol of \cite{ChakrabartiCGT14} works as follows. Alice picks a random function $h: [N]\to [m^3]$ from a pairwise independent hash family and sends Bob $h$ along with the set of at most $m$ values $h(i)$ for each index $i$ where her string has a 1. Bob checks if $h(j)$ is in this set, where $j$ is his input index, and if so, announces the answer to be $1$, and $0$ otherwise. The protocol can err only if Alice has a $0$ at index $j$ but $h(j)$ collides with one of the $m$ $h(i)$'s, which happens with low probability by pairwise independence and union bound over the $m$ elements. The set can be expressed using $O(m \log m)$ bits and the function $h$ takes $O(\log N)$ bits, giving their bound of $O(m\log m +\log N)$. 

It is not hard to see that the $\log m$ factor can be removed with a tighter analysis: we reduce the range of the hash family from $[m^3]$ to $[3m]$. By pairwise independence, $h(j)$ collides with a single $h(i)$ with probability $1/3m$. Since we need to take union bound over at most $m$ such $h(i)$'s and are happy with error probability at most $1/3$, this family suffices. Further, expressing the set of $m$ $h(i)$'s then takes only $O(m)$ bits since Alice can simply send the set's characteristic vector of size $3m$. It is, however, not clear that we can reduce the additive $O(\log N)$ any further as it appears because the smallest known size of a pairwise independent hash family with domain $[N]$ is $\Omega(N)$. We observe that $\spidx_{m,N}$ actually has input size $O(\binom{N}{m})$. Hence, if we draw the hash function using public randomness, thus obtaining an $O(m)$-cost public coin protocol, and then convert it to a private coin protocol using Newman's theorem, we incur an additive factor of $O(\log\log \binom{N}{m}) = O(\log m + \log\log N)$. Therefore, we prove that $\Rone(\spidx_{m,N})=O(m+\log\log N)$. The lower bounds of $\Omega(m)$ and $\Omega(\log\log N)$ easily follow from well-known lower bounds of the (standard) index and equality functions, establishing the tight bound of $\Theta(m+\log\log N)$.  

Using the above bound, we get that $\Rone(\spidx_{m,N})=\Theta(m)$ for any $m\geq \log\log N$, meaning that any non-trivial protocol for the problem must have $\vc = o(m)$.
To prove that the non-trivial complexity $\OMAnt(\spidx_{\log\log N, N}) = \Omega(\log N)$, we reduce $\spidx_{m,N}$ from $\eqidx_{p,q}$ with $p=m$ and $q =\log (N/m)$. Then, it follows that to obtain $\vc=o(m)$, we need $\hc$ of $\Omega(q)=\Omega(\log (N/m))$, which is $\Omega(\log N)$ for $m=\log\log N$. The reduction is very similar to the one in \cite{ChakrabartiCGT14}. They reduced from $\idx$, while we reduce from $\eqidx$: we essentially modify their construction to fit the structure of the $\eqidx$ problem. 

\mypar{Connectivity and Bipartite on XOR-graphs} Via reduction from $\eqidx$, we reproduce the lower bounds for the other two problems that we know have high non-trivial OMA complexity. These are the $\xorconn_n$ and $\xorbip_n$ problems that ask whether the XOR of the graphs that Alice and Bob hold on the vertex set $[n]$ is connected and bipartite respectively. Our reduction is very similar to those by Thaler \cite{Thaler16}. Although he reduced from $\idx_{n^2}$ and we reduce from $\eqidx_{n,n}$, our construction can be seen as adapting his construction to the structure of the $\eqidx$ problem. The reduction to $\xorconn$ works as follows. Given the input strings $x_1,\ldots,x_n$ where each $x_i\in \{0,1\}^n$, Alice constructs a bipartite graph with $n$ nodes on each partite set: for each $i\in [n]$, she sees $x_i$ as the characteristic vector of the neighborhood of $\ell_i$, the $i$th vertex on the left. To ensure that her graph is connected, she joins all vertices to a dummy vertex, say $v$. Bob treats his input $y$ as the characteristic vector of a set of edges incident to $\ell_j$. If $y$ is indeed identical to $x_j$, then XOR of these edges and the edge $\{(\ell_j,v)\}$ with the original graph must isolate $\ell_j$. Otherwise, $\ell_j$ must be adjacent to some vertex, say $w$ on the right, and $w$, in turn, is adjacent to $v$. Since all other vertices are adjacent to $v$, this ensures that the graph is connected in this case. Thus we can reduce $\xorconn_{n}$ from $\eqidx_{n,n}$ and obtain the desired result on its non-trivial OMA complexity. 

For $\xorbip$, after constructing the same base bipartite graph using the $x_i$'s, Alice joins only the vertices on the right to the dummy vertex $v$. Bob constructs the same edge set using $y$. The XOR of these edges and ${\ell_j,v}$ with Alice's graph creates a triangle if $y\neq x_j$: as before, $\ell_j$ must be adjacent to some $w$ on the right, and both of them are adjacent to $v$. Hence, in this case, the graph is not bipartite. If $y$ is indeed identical to $x_j$, then adding $v$ and $\ell_j$ respectively to the left and the right partite sets of the original graph gives a bipartition of the XOR graph. Thus, $\xorbip$ can be used to solve $\eqidx_{n,n}$, giving us the desired result.

\subsubsection{Annotated Streaming Lower Bounds}

We reduce from $\eqidx$ to prove lower bounds on the \emph{non-trivial
annotated streaming complexity} of several problems. It is defined as 
the minimum $\tc$ over all annotated streaming schemes that solve the 
problem with space (i.e., $\vc$) sublinear in its classical streaming 
complexity (ignoring polylogarithmic factors).

\mypar{Distinct Items} First, we consider the fundamental {\em distinct items} problem
$\distit_{N,F}$, where we need to count the number of items with 
non-zero frequency as elements from the universe $[N]$ get inserted
and deleted in a stream, and the absolute value of each frequency is
bounded above by $F$. Let us first see how we can solve it in classical 
streaming. For each element, we keep a {\em non-zero detector}, which, 
at the end of the stream, reports whether or not the element has 
non-zero frequency. In this work, we build such detectors, each of which takes $\tO(\log \log F)$ space.
 Thus, the problem can be solved using $\tO(N\log\log F)$ space in classical streaming. Therefore, even
 if we set $F=\exp(N^{\polylog(N)})$, we get that the classical streaming complexity is $\tO(N)$. 
A lower bound of $\Omega(N)$ is easy to show. Hence, any non-trivial annotated streaming scheme for the 
problem must have $\vc=o(N)$. Then we reduce $\distit_{N,F}$ from $\eqidx_{p,q}$ for $p=N$ and $q=\log F$, 
implying that any non-trivial scheme must have $\hc = \Omega(q) = \Omega(\log F)$. For $F=\exp(N^{\polylog(N)})$, 
this gives a bound of $\Omega(N^{\polylog(N)})$ on the $\hc$, and hence the $\tc$. Thus, we see that for a certain
 setting of $F$, the non-trivial annotated streaming complexity of $\distit_{N,F}$ can be as large as $\Omega(N^{\polylog(N)})$, 
while the classical streaming complexity is only $\tO(N)$. 

The reduction is straightforward: given the inputs $x_1,\ldots,x_N\in \{0,1\}^{\log F}$ in the $\eqidx_{N,\log F}$ problem, Alice treats the $x_i$s 
as binary representations0 of numbers in $[0,F-1]$ and creates the stream-prefix with $x_i+1$ insertions of the element $i$, for each $i\in [N]$. 
Bob appends to the stream $y+1$ deletions of the element $j$. It follows that this stream has $N$ distinct items iff $x_j\neq y$.    

\mypar{Connectivity-related problems on SGT streams} We introduce the support graph turnstile model where the input $n$-node graph is defined by the 
support of the $\binom{n}{2}$-length frequency vector whose entries are incremented or decremented by the stream tokens. We reduce global connectivity 
and $k$-connecitivity (both vertex and edge versions) under such streams from $\eqidx$ to establish high non-trivial annotated streaming complexity of those
 problems. To set the stage, we first prove that in the classical streaming model, the known upper bounds for these problems under insert-only or dynamic streams 
can be matched under SGT streams at the cost of just polylogarithmic factors. This means global connectivity under SGT streams can be solved using $\tO(n)$ space
and both $k$-vertex-connectivity and $k$-edge-connectivity can be solved in $\tO(kn)$ space. The crucial tool we use here are the strong $\ell_0$ samplers that 
we build in this paper. We give an overview of this in the next subsection. Here, we discuss our reductions proving the non-trivial complexity of these problems.

For connectivity, we reduce from $\eqidx_{n, n \log F}$. Given her inputs $x_1,\ldots, x_{n} \in \{0,1\}^{n\log F}$, Alice constructs a bipartite graph with $n$ nodes 
on each partite set. She treats each $x_i$ as $n$ blocks of size $\log F$ each, with each block being the binary representation of a number in $[0,F-1]$. 
Denote the $b$th block $x_i$ by $x_{i,b}$. Then for each $i,b$, she makes $x_{i,b}+1$ insertions of the edge between $i$th vertex on the left and the $b$th vertex on the right. 
Thus, we get a complete bipartite graph with each edge having a weight in $[F]$. Bob then breaks $y$ into $n$ blocks $y_1,\ldots,y_n$ of $\log F$ bits each. For each $b\in [n]$, 
she deletes the edge between the $j$th vertex on the left and $b$th vertex on the right $y_b+1$ times. Therefore, if $y=x_j$, then the $j$th vertex on the left gets isolated and the 
graph is not connected. Otherwise, it must be connected to some node on the right, and the rest of the graph is a complete bipartite graph. Hence, the graph is connected. 
It follows that any non-trivial scheme, i.e., a scheme with $\vc=o(n)$ necessiates $\hc = \Omega(n\log F)$. As before setting $F=\exp(n^{\polylog{n}})$, we get that 
the non-trivial streaming complexity of connectivity can be as large as $\Omega(n^{\polylog{n}})$ while the classical streaming complexity is $\tO(n)$.

For $k$-connectivity, we reduce from $\eqidx_{kn^2, \log F}$. Given her inputs $x_1,\ldots, x_{kn^2} \in \{0,1\}^{\log F}$, Alice constructs a bipartite graph with $kn$ nodes on the left partite set 
and $n$ nodes on the right partite set. She treats each $x_i\in {0,1}^{\log F}$ as the binary representation of a number in $[0,F-1]$. Following some canonical ordering of the 
$kn^2$ possible edges, she inserts the $i$th edge $x_i+1$ times for each $i\in[kn^2]$. Bob then deletes the $j$th edge $y+1$ times. Thus, the resultant support graph is a complete bipartite 
if $x_j\neq y$. Otherwise it's a complete graph minus only the $j$th edge. It is easy to prove that the graph in the first case is $k$-vertex-connected as well as $k$-edge-connected, whereas 
the graph in the second case is not.

\subsubsection{Classical Streaming Algorithms}

Some of our algorithms for SGT streams follow by replacing the $\ell_0$-samplers in existing algorithms by the new $\ell_0$-samplers that we design in this work. However, an algorithm for $k$-edge-connectivity does not immediately follow. So we design a new algorithm here, and this is one of our significant technical contributions. Another key technical ingredient is our ``layering lemma'' that we use to design efficient schemes for vertex connectivity on dynamic graphs. We discuss these tools and techniques in detail below. 

\mypar{Strong $\ell_0$-samplers} We design new $\ell_0$-samplers that can handle very large frequencies (hence called strong) and still take $\polylog(n)$ space. This is helpful because SGT streams can have very large frequencies. Standard $\ell_0$-samplers usually assume that the frequencies are $\poly(n)$ where $n$ is the universe size.

The main tool we use to design our sampler is our \emph{decision counter}. These counters, given a stream of insertions and deletions, can detect if the number of insertions is exactly equal to the number of deletions or not, using very little space. This even works when the difference between the number of insertions and deletions is very large which is challenging to do in small space deterministically.
The idea is to maintain a standard counter modulo a large random prime.

We then use ideas from the sparse recovery and $\ell_0$-sampling literature to build the strong $\ell_0$-sampler.
We first solve the problem when the non-zero support is $1$ using non-adaptive binary search.

The problem here is that we are given a stream of insertion and deletions over a universe of $n$ elements and promised that at the end of the stream, there is exactly one element with a non-zero frequency. The goal is to find that element.
The usual idea is to maintain one counter that keeps a track of the number of elements (i.e. number of insertions minus deletions).
Another counter is maintained which is used to find the value of the non-zero frequency element. When an insertion/deletion for element $i$ arrives it is scaled by $i$ and then added/subtracted from the counter. At the end of the stream, this counter contains $i$ times frequency of element $i$ and the other counter has just the frequency of element $i$ thus recovering the non-zero frequency element $i$.
This does not work when the frequencies are large so we have to use non-adaptive binary search to solve the problem which uses $O(\log n)$ counters to find $i$.

In the case with no promise, we want to reduce to the support $1$ case. We do this by guessing the support size $s$ in powers of $2$ and sampling the elements with probability $1/s$.
This means that for the correct guess of $s$ we reduce to the support $1$ case with constant probability. We can repeat $O(\log n)$ times for high probability.

Finally, we need to distinguish the support $1$ case from the cases where the support is not $1$. 
This is because we have many problems and only a few of them are the support $1$ case and in others the support can be $0$ or larger than $1$.
We design a sketch for this by randomly partitioning the universe into many parts, summing those parts up and checking how many parts are non-zero.
The hope is that if there are at least $2$ elements then multiple parts will become non-zero.indicating that the support is more than $1$. If the non-zero support is $1$ then exactly one part will be non-zero and if the non-zero support is $0$ then all parts will be $0$.
This happens with constant probability so we repeat $O(\log n)$ times for high probability.
Putting everything together gives us the strong $\ell_0$-sampler sketch.

\mypar{Edge Connectivity} We give a randomized algorithm for getting a certificate (a spanning subgraph with the same answer to $k$-edge-connectivity) of $k$-edge-connectivity. The certificate is of size $\Ot(kn)$, which is optimal up to polylog factors. The certificate needs to be of size $\Omega(kn)$ because every vertex needs to have degree at least $k$.
This algorithm can be easily converted into a dynamic streaming algorithm using the dynamic streaming spanning forest implementation of \cite{AhnGM12}.
We also show that using our strong $\ell_0$-samplers, the spanning forest streaming implementation of \cite{AhnGM12} can be extended to SGT streams in $\Ot(n)$ space. Using our certificate along with the spanning forest algorithm in SGT streams, we can solve $k$-edge-connectivity in SGT streams in $\Ot(kn)$ space.

\cite{AhnGM12} also give an algorithm for $k$-edge-connectivity in dynamic streams using $\Ot(kn)$ space. However, this algorithm cannot be extended to SGT streams by simply replacing their $\ell_0$-samplers with strong $\ell_0$-samplers.
Their certificate is $k$ edge-disjoint spanning forests. On a high level, this does not work in SGT streams because the spanning forests depend on each other (since they are edge-disjoint).
In their algorithm, after recovering one spanning forest $T_1$ of $G$ we essentially need to delete the edges of $T_1$ from $G$ and then recover a new spanning forest $T_2$ (to get a disjoint spanning forest). This is easy to do in dynamic streams because the frequencies of the edges of $T_1$ are exactly $1$. So we can generate a new stream which is the old stream appended with deletions of the edges of $T_1$. Then $T_2$ can be recovered from this new stream.
However, in SGT streams, the frequency of the edges of $T_1$ could be arbitrary, so we cannot easily generate another stream that represents the graph $G - T_1$ (to do this we need to know the exact frequencies of all edges in $T_1$).
This dependency between the spanning forests makes extending this algorithm difficult. Our algorithm, on the other hand, gets rid of this dependency and thus is extendable to SGT streams.

The randomized algorithm for getting a certificate of $k$-edge-connectivity is heavily inspired by the randomized algorithm for getting a certificate of $k$-vertex-connectivity \cite{assadi2023tight}.
In the randomized algorithm, we independently sample every edge with probability $1/k$ and find a spanning forest of the sampled subgraph.
The certificate then is a union of the spanning forests in $\Ot(k)$ such independent iterations.
The analysis then shows that edges whose endpoints are not very well connected in the original graph exist in the certificate, and pairs of vertices that are very well connected in the original graph are well connected in the certificate. This is enough to prove that the certificate preserves the answer to $k$-edge connectivity.

We also give an $[n,n]$-scheme for $k$-edge-connectivity in the annotated dynamic streaming model. We achieve this by simulating the two-pass streaming algorithm for minimum-cut implied by \cite{rsw18}. During the stream, the verifier computes a cut sparsifier of the graph. After the stream, the verifier compresses the vertices into supernodes (by compressing all large cuts) such that only the small cuts remain. The prover then sends all the edges of this supernode graph (\cite{rsw18} showed that there are only $O(n)$ such edges). The verifier then can compute the exact mincut of the graph, thus solving $k$-edge-connectivity for all values of $k$.

\mypar{Vertex connectivity}
We also get an algorithm for $k$-vertex-connectivity in SGT streams in $\Ot(kn)$ space using the techniques used for edge connectivity. We use strong $\ell_0$-samplers in the algorithm of \cite{assadi2023tight}.

We also give an $[k^2 n,n/k]$-scheme for $k$-vertex-connectivity in the annotated dynamic streaming model.
The heart of the algorithm is what we call the layering lemma, which says that there exists a short proof (size $\Ot(kn)$) to show that any arbitrary fixed vertex has $k$ vertex-disjoint paths to all other vertices. This proof can also be verified in $\Ot(n/k)$ space. Also, when proving this for $r$ vertices, we can reuse space and thus get a $[r \cdot k n,n/k]$-scheme instead of a $[r k n,r n/k]$-scheme. If we show the layering lemma for $k$ vertices, then using properties of $k$-vertex-connectivity, we can show that the graph is $k$-vertex-connected, giving the desired bound. We also extend these ideas to $k$-edge-connectivity getting a $[k^2 n,n/k^2]$-scheme.

We show that a short proof exists for the layering lemma using the probabilistic method.
The idea is to cleverly partition the vertices into $\log n$ layers and show that vertices in the first layer have $k$ vertex-disjoint paths to the special vertex. Then we inductively show that vertices in a layer have $k$ vertex-disjoint path to the previous layer. Using the properties of $k$-vertex-connectivity, this is enough to show that the special vertex has $k$ vertex-disjoint paths to all other vertices. The proof between any two layers is of size $\Ot(kn)$, giving us the desired bound.
The verification is more involved, but the verifier essentially checks just two things using some tools. He first checks whether the edges sent by the prover belong to the input graph (using a subset check). He then checks if what the prover sent are indeed vertex-disjoint paths (using a duplicate-detection scheme).

We also give an $[n^2 \log \val + k^2 n,1]$-scheme for $k$-vertex-connectivity and $k$-edge-connectivity in the annotated SGT streaming model with parameter $\val$.
The proof idea is the same as in dynamic streams, but the verification is more complicated because the frequencies for the edges could be large or even negative. The auxiliary information used for verification adds an overhead of $\Ot(n^2 \log \val)$ bits in the proof. The verification follows the same steps as in the dynamic streaming case, except the subset check is not easy to do. So we came up with a different way to do the subset check, which needs a large amount of auxiliary information. One of the main ideas for this is to separate out the elements with positive and negative frequencies and do a subset check for them separately.

\subsection{Paper Organization}
We start by defining the models and notation in \Cref{sec:models}. Then in \Cref{sec:prelim}, we list all the facts we use throughout the paper. In \Cref{sec:eqinx-LB}, we introduce our central problem called Equals-Index and use it to prove lower bounds for different problems in MA communication complexity and annotated streaming. \Cref{sec:SGT-streaming-algs} is our toolkit section where we start with our algorithm for the strong $\ell_0$-samplers which can be of independent interest. We then provide an algorithm for $k$-edge-connectivity in the classical setting as well as in dynamic and SGT streaming which can be of independent interest. The techniques in \Cref{sec:SGT-streaming-algs} can be used to get SGT streaming algorithms for connectivity and $k$-veretex-connectivity.
In \Cref{sec:Annotated-dynamic}, we show annotated dynamic streaming schemes for $k$-vertex-connectivity and $k$-edge-connectivity. At the heart of the proof for $k$-vertex-connectivity is the layering lemma, which can be of independent interest. Finally, in \Cref{sec:Annotated-SGT}, we show an annotated SGT streaming scheme for $k$-vertex-connectivity.
\section{Models, Notation, and Terminology}\label{sec:models}

\subsection{Models}

We formally describe the various communication and streaming models that we focus on in this paper and define the notation and terminology in each model.

\subsubsection{Communication Models}

\mypar{Classical Two-Party Communication} In the (randomized) two-party communication model introduced by Yao \cite{Yao79}, two players Alice and Bob possess inputs $x\in \cX$ and $y\in \cY$ respectively. Their goal is to compute $f(x,y)$, where $f:\cX\times \cY\to \{0,1\}$. To this end, they use a communication protocol $\Pi$ where they toss some random coins and accordingly send messages to each other back and forth in rounds. The protocol $\Pi$ terminates when one of the players announces an output. In this paper, we only focus on one-way communication, i.e., when $\Pi$ has just a single round: Alice sends a single message to Bob, from which he must declare the output. We now describe this variant of the model in more detail.  

\noindent
{\it One-way randomized communication.} In a \emph{private coin} protocol, Alice's coin tosses are private, i.e., she draws a random string $\cR$ that is unknown to Bob, and sends him a message $\msg$ as a function of $(x,\cR)$. In a {\em public coin} protocol, the string $\cR$ is also known to Bob (without any communication). After receiving Alice's message, Bob outputs a bit $\out(\Pi)$ as a function of $(y,\msg)$ (resp. $(y,\msg, \cR)$) for private (resp. public) coin protocols. We say that a protocol $\Pi$ for a function $f$ has error $\delta$ if $\Pr_{\cR} [\out(\Pi) \neq f(x,y)] \leq \delta$. We say that a protocol $\Pi$ \emph{solves} a function $f$ if it has error $\delta\leq 1/3$. 

The communication cost $\cost(\Pi)$ of protocol $\Pi$ is the maximum length of Alice's message over all possible $(x,\cR)$. The randomized one-way communication complexity of a function is defined as 

$$\Rone(f) := \min\{\cost(\Pi) : \Pi \text{ is a private coin protocol that solves } f\} $$

When public randomness is allowed, we denote this complexity as

$$\Ronepub(f) := \min\{\cost(\Pi) : \Pi \text{ is a public coin protocol that solves } f\} $$

\mypar{Merlin-Arthur Communication} In the Merlin-Arthur (MA) communication model \cite{BabaiFS86}, we have the usual two players Alice and Bob (collectively called ``Arthur'') with their inputs $x\in \cX$ and $y\in \cY$ respectively. They want to compute $f(x,y)$, where $f:\cX\times \cY\to \{0,1\}$. In addition, there is an all-powerful player Merlin who knows the inputs $x$ and $y$. Merlin is untrusted; he sends Alice and Bob a ``proof'' of the solution, following which these two players interact between themselves to verify the proof. In this paper, we focus on the \emph{online Merlin-Arthur} (OMA) model which is an MA analog of the one-way randomized communication model defined above. We describe this model in more detail.  

\noindent
{\it Online Merlin-Arthur communication.} An OMA protocol $\Pi$ works as follows. Merlin sends Bob a help message $\cH$. Then Alice generates a random string $\cR$, based on which she sends Bob a message $\msg$. Bob then outputs $\out(\Pi)\in \{0,1,\bot\}$ as a function of $(y,\cH, \msg,\cR)$. An OMA protocol $\Pi$ is said to have completeness error $\delta_c$ and soundness error $\delta_s$ if the following conditions are satisfied.

\begin{itemize}
    \item (Completeness) If $f(x,y) = 1$, then there exists a function $\cH$ such that $\Pr_{\cR}[\out(\Pi) \neq 1] \leq \delta_c$
    \item (Soundness) If $f(x,y) = 0$, then $\forall \cH': \Pr_{\cR}[\out(\Pi) = 1] \leq \delta_s$
\end{itemize}

We say that a protocol $\Pi$ \emph{solves} a function $f$ if $\Pi$ has $\delta_s, \delta_c\leq 1/3$ for computing $f$.

For an OMA protocol $\Pi$, the
\emph{help cost} $\hc(\Pi)$ is the maximum length of $\cH$ over all possible $(x,y)$, and the \emph{verification cost} $\vc(\Pi)$ is the
maximum number of bits sent by Alice over all possible $(x,\cR)$. The \emph{total cost} $\tc(\Pi)$ is then defined as $\hc(\Pi) +\vc(\Pi)$. 
The OMA-complexity of $f$ is defined as

$$\OMA(f) := \min\{\tc(\Pi) : \text{$\Pi$ \text{solves} $f$}\}.$$ 

\noindent
{\it Trivial OMA protocol.} We say that an OMA protocol $\Pi$ that solves a function $f$ is \emph{trivial} if $\vc(\Pi)=\Omega(\Rone(f))$ and is \emph{non-trivial} if $\vc(\Pi)=o(\Rone(f))$.

\begin{definition}[Non-trivial OMA complexity]\label{def:ntoma}
The \emph{non-trivial} OMA complexity of a function $f$ is 

$$\OMAnt(f) := \min\{\tc(\Pi) : \text{$\Pi$ \text{is a non-trivial protocol that solves} $f$}\}.$$    
\end{definition}

Note that, by definition, $\OMAnt(f)\geq \OMA(f)$. We also make the following observation. 

\begin{observation}\label{obs:ntoma-exceeds-oma}
If $\OMAnt(f)= \omega(\OMA(f))$, then it must be that $\OMA(f)=\Omega(\Rone(f))$.
\end{observation}
 
This is because if $\OMAnt(f)$ is larger than $\OMA(f)$, then the ``optimal'' protocol $\Pi$ for which $\OMA(f)=\tc(\Pi)$ must be a trivial protocol. Hence, $\tc(\Pi)=\Omega(\vc(\Pi))=\Omega(\Rone(f))$ by definition.   

\subsubsection{Streaming Models} 

First, we define the different types of data streams that we consider in this paper. 

\mypar{Data Streams} A data stream of length $m$ is defined as a sequence $\sigma:= \langle (a_i, \Delta_i): i\in [m]\rangle$, where $a_i\in [N]$ and $\Delta_i$ is an integer. The data stream defines a $\emph{frequency vector}$ $\textbf{f}(\sigma):=\langle\freq(1),\ldots,\freq(N)\rangle$, where $$\freq(j):=\sum_{i\in [m]:~a_i = j}\Delta_i$$
is the frequency of element $j\in [N]$. In the streaming model, we need to compute some function of $\textbf{f}(\sigma)$ for an input stream $\sigma$ using space sublinear in the input size, i.e., $o(m)$ space.

\noindent
{\it Turnstile streams.} If a stream is as defined above, with no restriction on $\Delta_i$ and $\textbf{f}(\sigma)$, then we call it a turnstile stream. Conceptually, turnstile streams allow frequencies of elements to be incremented and decremented arbitrarily with each update, and the frequency vector can have negative entries.

\noindent
{\it Insert-only streams.} A data stream is called insert-only if $\Delta_i>0$ for all $i\in [m]$. A well-studied special case is the \emph{unit update} model, where $\Delta_i=1$ for all $i\in [m]$.

\noindent
{\it Insert-delete streams.} We call a data stream an insert-delete stream if $\Delta_i\in \{-1,1\}$ for all $i\in [m]$ and $\freq(j)\in \{0,1\}$ for all $j\in [N]$ at all points during the stream. Conceptually, elements are either inserted or deleted, frequency can be at most $1$, and an element can be deleted only if its frequency is $1$ during the update.   

\mypar{Graph Streams} A graph stream is a data stream whose frequency vector defines a graph. Here, $N=\binom{n}{2}$ where $[n]$ is the set of nodes of the graph. The frequency vector is indexed by pairs $(u,v)\in \binom{[n]}{2}$ (using a canonical bijection from $[N]$ to $\binom{[n]}{2}$), and typically, $\freq(u,v)$ represents the weight of the edge $\{u,v\}$.  

\noindent
{\it Dynamic graph streams.} We call an insert-delete graph stream as a dynamic graph stream. Conceptually, the input graph is unweighted, and an edge can be deleted only if it is present in the graph during the update.

We now introduce support graph turnstile (SGT) streams. Conceptually, the graph is induced by the support of the frequency vector of the input stream, which may be a turnstile stream. 

\begin{definition}[Support graph turnstile streams]
A turnstile stream $\sigma$ is called a support graph turnstile stream if it is of the form $\langle (u,v)_i, \Delta_i): i\in [m]\rangle$, where $(u,v)\in \binom{[n]}{2}$ and the graph that it defines is given by $G=([n], \{(u,v) : \freq(u,v)\neq 0\})$.  
\end{definition}

Now we describe the streaming models that we study in this paper. 

\mypar{Classical Streaming} In the classical streaming model, a (randomized) data streaming algorithm $\alg$ reads an input stream $\sigma$ and needs to compute some function $f(\sigma)$ of the stream. To this end, it maintains a summary $\alg_{\cR}(\sigma)$ based on a random string $\cR$ and at the end of the stream, outputs $\out(\alg)$. The goal is to optimize the space usage, i.e., the size of $\alg_{\cR}(\sigma)$. In this paper, we only consider algorithms that make a single pass over the input stream. We say that a streaming algorithm $\alg$ has error $\delta$ if $\Pr_{\cR}[\out(\alg)\neq f(\sigma)]\leq \delta$. Further, we say that an algorithm $\alg$ \emph{solves} a function $f$ if it has error $\leq 1/3$. The maximum number of bits stored by $\alg$ over all possible inputs $\sigma$ is denoted by $\Sp(\alg)$. 

We define the \emph{classical streaming complexity} of a function $f$ as 
$$\cS(f):= \min\{\Sp(\alg): \alg \text{ solves } f\}$$

\mypar{Annotated Streaming} In the annotated streaming model of \cite{ChakrabartiCM09}, we have a space-bounded Verifier and an all-powerful Prover with unlimited space. 
  Given an input stream $\sigma$, a {\em scheme} for computing a function $f(\sigma)$ is a triple $\cP=(\cH, \alg, \out)$, where $\cH$ is a function that Prover uses to generate the help message or proof-stream $\cH(\sigma)$ that she sends Verifier \emph{after} the input stream, $\alg$ is a data streaming algorithm that Verifier runs on $\sigma$ using a random string $\cR$ to produce a summary $\alg_{\cR}(\sigma)$, and $\out$ is an algorithm that Verifier uses to process the proof $\cH(\sigma)$ and generate an output $\out(\cH(\sigma), \alg_{\cR}(\sigma), \cR)\in \text{range}(f)\cup \{\bot\}$, where $\bot$ denotes rejection of the proof. Note that if the proof length $|\cH(\sigma)|$ is larger than the memory of the Verifier, then $\out$ is a streaming algorithm that processes $\cH(\sigma)$ as a stream and stores a summary subject to its memory. 

A scheme $\cP=(\cH, \alg, \out)$ has completeness error $\delta_c$ and soundness error $\delta_s$ if it satisfies 
\begin{itemize}[topsep=4pt,itemsep=0pt]
  \item (completeness) $\forall \sigma: \Pr_{\cR}[\out(\cH(\sigma),\alg_R(\sigma),\cR) \neq f(\sigma)] \le \delta_c$; 
  \item (soundness) $\forall \sigma, \cH': \Pr_{\cR}[\out(\cH',\alg_{\cR}(\sigma),\cR) \notin \{f(\sigma), \bot\}] \le \delta_s$.
\end{itemize}

We say that a scheme $\cP$ \emph{solves} a function $f$ if $\cP$ has $\delta_s, \delta_c\leq 1/3$ for computing $f$.

The {\em hcost} (short for ``help cost'') of a scheme $\cP=(\cH, \alg, \out)$ is defined as $\max_{\sigma} |\cH(\sigma)|$, i.e., the maximum number of bits required to express a proof over all possible inputs $\sigma$. The {\em vcost} (short for ``verification cost'') is the maximum bits of space used by the algorithms $\alg$ and $\out$ respectively, where the maximum is taken over all possible $(\sigma,\cR)$. The total cost $\tc(\cP)$ is defined as the sum $\hc(\cP)+\vc(\cP)$.  

The \emph{annotated streaming complexity} of a function $f$ is defined as

$$\AS(f) := \min\{\tc(\cP) : \text{$\cP$ \text{solves} $f$}\}.$$ 

\noindent
{\it Trivial Scheme.} We say that a scheme $\cP$ that solves a function $f$ is \emph{trivial} if $\vc(\cP)=\tOmega(\cS(f))$.

\begin{definition}[Non-trivial Annotated Streaming complexity]\label{def:ntas}
The \emph{non-trivial} annotated streaming complexity of a function $f$ is 

$$\widehat{\AS}(f) := \min\{\tc(\cP) : \text{$\cP$ \text{is a non-trivial scheme that solves} $f$}\}.$$    
\end{definition}

A scheme $\cP$ with $\hc(\cP)= O(h)$ and $\vc(\cP)=O(v)$ is called an $(h,v)$-scheme. Again, if $
\hc(\cP)= \tO(h)$ and $\vc(\cP)=\tO(v)$, we call $\cP$ an $[h,v]$-scheme.

\subsection{Notation and Terminology}

\mypar{Basic Notation} We define some notation and terminology that we use throughout the paper. All logarithms are base $2$. The $\tO(.), \tOmega(.), \tomega(.)$ notation hides factors polylogarithmic in the input size. The notation $[k]$ for a natural number $k$ denotes the set $\{1,\ldots,k\}$. We use $[a,b]$ for integers $a<b$ to denote the set $\{a,\ldots,b\}$. For a string $z\in [\alpha]^d$, we use $z[k]$ to denote the element at the $k$th index of $z$. We use the term ``with high probability" to mean with probability at least $1-1/\poly(N)$, where $N$ is the input size. Although the standard notation to denote the input size in a communication problem is $n$, we use $N$ instead to avoid confusion with the number of nodes (for which we use $n$). 

\mypar{Graph Notation}
All graphs in this paper are simple and undirected. Given a graph $G=(V,E)$, we use $n$ for its number of nodes  $|V|$ and $m$ for its number of edges $|E|$ unless specified otherwise. The degree of a vertex $v \in V$ is denoted by $\deg(v)$, and $N(v)$ denotes its neighborhood. For a subset $F$ of edges in $E$, we use $V(F)$ 
to denote 
the vertices incident on $F$; similarly, for a set $U$ of vertices, $E(U)$ denotes the edges 
incident on $U$. We further use $G[U]$ for any set $U$ of vertices to denote the induced 
subgraph of $G$ on $U$. 
For any two vertices $s,t \in V$, we say that a collection of $s$-$t$ paths are 
vertex-disjoint if they do not share any vertices other than $s$ and $t$.

\section{Preliminaries}\label{sec:prelim}


We use the following standard forms of Chernoff bounds.

\begin{fact}[Chernoff bound; 
c.f.~\cite{dubhashi2009concentration}]\label{prop:chernoff}
	Suppose $X_1,\ldots,X_m$ are $m$ independent random variables with range $[0,b]$ 
	each for some $b \geq 1$. Let 
	$X := \sum_{i=1}^m X_i$ and $\mu_L \leq \expect{X} \leq \mu_H$. Then, for any $\eps > 
	0$, 
	\[
	\Pr\paren{X >  (1+\eps) \cdot \mu_H} \leq \exp\paren{-\frac{\eps^2 \cdot 
	\mu_H}{(3+\eps) \cdot b}} \quad 
	\textnormal{and} \quad \Pr\paren{X <  (1-\eps) \cdot \mu_L} \leq \exp\paren{-\frac{\eps^2 
	\cdot 
			\mu_L}{(2+\eps) \cdot b}}.
	\]
\end{fact}

In certain places, we also use limited independence hash functions in our algorithms to reduce their space complexity. 
 
\begin{definition}[Limited-independence hash functions]\label{def:pihash}
 For integers $n,m,k \geq 1$, a family $\mathcal{H}$ of hash functions from $[n]$ to $[m]$ is called a $k$-wise independent hash function iff for any two $k$-subsets $a_1,\ldots,a_k \subseteq [n]$ and $b_1,\ldots,b_k \subseteq [m]$,  
 \[
 \Pr_{h \sim \mathcal{H}}\paren{h(a_1)=b_1 \wedge \cdots \wedge h(a_k) = b_k} = \frac1{m^k}. 
 \]
\end{definition}
 
Roughly speaking, a $k$-wise independent hash function behaves like a totally random function when considering at most $k$ elements. When $k=2$, we call it a pairwise independent hash family.
We use the following standard result for $k$-wise independent hash functions.
\begin{fact}[\cite{MotwaniR95}]\label{fact:hashfam}
 For integers $n,m,k \geq 2$, there is a $k$-wise independent hash function $\mathcal{H} = \set{h: [n] \rightarrow [m]}$ so that sampling and storing a function $h \in \mathcal{H}$ takes $O(k \cdot (\log n + \log m))$ bits of space.
\end{fact}

We use the following basic communication complexity facts.

\begin{fact}[\cite{KushilevitzNisan-book}]\label{fact:eqsk}
There is an $O(1)$-cost public random protocol for $\eq_N$ for any $N$ with error probability at most $1/3$. 
\end{fact}

\begin{fact}[Equality\cite{KushilevitzNisan-book}]\label{fact:eqlb}
    The one-way randomized communication complexity $\Rone(\eq_N) = \Omega(\log N)$.
\end{fact}

\begin{fact}[Index \cite{Ablayev96}]\label{fact:idxlb}
     The one-way randomized communication complexity $\Rone(\idx_N) = \Omega(N)$.
\end{fact}

\begin{fact}[Newman's Theorem\cite{KushilevitzNisan-book}]\label{fact:Newman}
For any function $f:\{0,1\}^N\times \{0,1\}^M\to \{0,1\}$, a public coin protocol with communication cost $C$ and error $\eps$ can be simulated by a private coin protocol with communication cost at most $O(C+\log (N+M) + \log (1/\delta))$ and error at most $\eps+\delta$.
\end{fact}

We use a pseudorandom generator to reduce the space used by our random bits in certain algorithms. The fact that we exploit is given below. 

\begin{fact}[PRG \cite{Nisan90}]\label{fact:NisanPRG}
    Any randomized algorithm that runs in $S$ space and uses one-way access to $R$ random bits may be converted to an algorithm that uses $O(S \log R)$ random bits and runs in $O(S \log R)$ space using a pseudorandom generator. For dynamic streaming algorithms, the algorithm should be able to access the random bits in a read-once manner for some permutation of the stream.
\end{fact}

We often refer to the following generic lower bound in the OMA model as well as annotated streaming. 

\begin{fact}[General OMA lower bound,\cite{ChakrabartiCMT14}]\label{fact:genlb}
For any function $f$, an $(h,v)$-OMA-procotol or an $(h,v)$-scheme $\Pi$ solving it must have $$h\cdot v\geq \Omega(C)$$
where $C= \Rone(f)$ in case $\Pi$ is an OMA protocol, and $C=\cS(f)$ (the streaming complexity of $f$) in case it is an annotated streaming scheme. 
\end{fact}

The standard schemes below are often used as subroutines in our protocols.

\begin{fact}[Subset-check and Intersection-count Scheme, \cite{ChakrabartiCMT14}, \cite{ChakrabartiGT20}]\label{fact:subsetdisjchk}
    Given an insert-delete stream of elements from sets $X, Y\subseteq [N]$ (interleaved arbitrarily), for any $h,v$ with $h\cdot v=N$, there are $[h,v]$-schemes for checking whether $X\subseteq Y$ and for counting $|X\cap Y|$ .
\end{fact}


\begin{fact}[Duplicate-detection Scheme]\label{fact:dupdet}
    Given an insert-only stream of elements from the universe $[N]$, for any $h,v$ with $h\cdot v=N$, there is an $[h,v]$-scheme for checking whether all elements in the stream are distinct. 
\end{fact}

We now mention an important proposition for edge and vertex connectivity.
The first is Menger's theorem which gives an equivalent definition of 
$k$-vertex-connectivity via vertex-disjoint paths.
\begin{fact}[Menger's Theorem; c.f.~{\cite[Theorem 17]{west2001introduction} 
	}]\label{fact:menger}
    Let $G$ be an undirected graph and $s$ and $t$ be two vertices. 
    Then the size of the minimum edge cut for $s$ and $t$ is equal to the maximum number of pairwise edge-independent paths from $s$ to $t$. 
    If $s,t$ are non-adjacent then the size of the minimum vertex cut for $s$ and $t$ is equal to the maximum number of vertex-disjoint paths between $s$ and $t$. \\
    Moreover, a graph is $k$-edge-connected (resp $k$-vertex-connected) iff every pair of vertices has $k$ edge-disjoint (resp vertex-disjoint) paths in between them.
\end{fact}

We often use the following well-known sketch from the streaming literature.

\begin{fact}[\cite{AhnGM12}] \label{fact:spanning-forest}
There is an algorithm that given any $n$-vertex graph $G$ in a dynamic stream computes a spanning forest $T$ of $G$ with probability at least $1-n^{-10}$ in $O(n \log^3 (n))$ bits of space.
\end{fact}

We give a new way to detect if two multi-sets are identical using $\ell_0$-samplers.
The idea is to insert elements of the first multi-set into an $\ell_0$-sampler and delete elements of the second multi-set from the sampler, and at the end, check if the sampler is empty.
\begin{proposition}[Equality-Detection Scheme]\label{fact:equaldet}
    Given an insert-only stream of elements of two multi-sets $A, B$ in any order (interleaved arbitrarily) from the universe $[N]$, we can check if the two sets are identical with probability $1-1/\poly(N)$ in space $\polylog(N)$. 
\end{proposition}


\section{OMA Lower Bound for Equals-Index and its Implications}\label{sec:eqinx-LB}

\subsection{The Equals-Index Problem}

We formally define the $\eqidx_{p,q}$ communication game between two players Alice and Bob as follows. 

Let $p$ and $q$ be arbitrary integers. Alice gets $p$ strings $x_1,\ldots, x_p$ such that $x_i\in \{0,1\}^q$ for each $i\in [p]$. Bob gets a string $y\in \{0,1\}^q$ and an index $j\in [p]$. The output is $1$ if $y=x_j$, and $0$ otherwise. 

First, let us show tight bounds on its one-way (Alice $\to$ Bob) randomized communication complexity.
\begin{lemma}\label{lem:onewayeqidx}
    $\Rone(\eqidx_{p,q}) = \Theta(p+\log q)$
\end{lemma}

\begin{proof}
    First consider the following protocol using public randomness. For each $i\in [p]$, Alice sends Bob an $O(1)$-size equality sketch (\Cref{fact:eqsk}) for $x_i$. Bob uses only the $j$th sketch to check if $x_j=y$. Thus, $\Rpub(\eqidx_{p,q})=O(p)$. By Newman's theorem (\Cref{fact:Newman}),   $\Rone(\eqidx_{p,q})=O(p+\log (pq))=O(p+\log q)$. 
    
    The lower bound  $\Rone(\eqidx_{p,q}) = \Omega(p+\log q)$ easily follows from the facts that $\eqidx_{p,1}\equiv\idx_p$ and $\eqidx_{1,q}\equiv\eq_q$. Therefore, we have
    
    $\Rone(\eqidx_{p,q})\geq \Rone(\eqidx_{p,1}) = \Rone(\idx_{p}) = \Omega(p)$ (\Cref{fact:idxlb}). 
    
    Again, $\Rone(\eqidx_{p,q})\geq \Rone(\eqidx_{1,q}) \geq \Omega(\log q)$ (\Cref{fact:eqlb}).  
\end{proof}


Now we prove an OMA lower bound for \eqidx.

\begin{lemma}\label{lem:eqidxmain}
    Any (h,v)-OMA-protocol solving $\eqidx_{p,q}$ must have $$(h+q)\cdot v \geq \Omega(pq)$$  
\end{lemma}
    
\begin{proof}
    We show that given any $(h,v)$-OMA-protocol $\Pi$ for $\eqidx_{p,q}$, we can design an $(h+q,v)$-OMA-protocol $\Pi'$ for $\idx_{pq}$. The lower bound of $(h+q)\cdot v\geq \Omega(pq)$ then immediately follows from the OMA lower bound for $\idx_{pq}$ (\Cref{fact:idxlb}).

    Suppose Alice and Bob have inputs $x\in\{0,1\}^{pq}$ and $k\in [pq]$ in the $\idx_{pq}$ problem. Then the protocol $\Pi'$ is as follows. Alice partitions $x$ into $p$ chunks $x_1,\ldots,x_p$, each of size $q$. Bob sets $j=\ceil{k/p}$. Merlin sends $y\in \{0,1\}^q$ to Bob and claims that it equals $x_j$. The players can now interpret $(x_1,\ldots,x_p)$ and $(y,j)$ as inputs to the $\eqidx_{p,q}$ problem, and run the protocol $\Pi$ to verify that $y$ is indeed equal to $x_j$. If the check passes, then Bob knows $x[k]$ since it lies in $x_j=y$. If not, he outputs $\bot$, i.e., rejects the proof. 
    
    We analyze the completeness and soundness errors of $\Pi'$. If Merlin is honest, then $y$ is indeed $x_j$, and the protocol fails only if Bob rejects because the check in $\Pi$ doesn't go through. The probability of this is exactly the completeness error of $\Pi$. Hence, $\Pi'$ has the same completeness error as $\Pi$. Again, if Merlin is dishonest, then the protocol fails only when Bob outputs the incorrect bit-value for $x[k]$. This happens only when $y\neq x_j$, but the check passes in $\Pi$. This has the same probability as the soundness error of $\Pi$. Therefore, $\Pi'$ also has the same soundness error as $\Pi$. Thus, by definition, since $\Pi$ solves $\eqidx_{p,q}$, the protocol $\Pi'$ solves $\idx_p$.
    
    Finally, we analyze the cost of $\Pi'$. Merlin sends Bob $y$ as well as his message due to $\Pi$. Thus, hcost$(\Pi') = O(h+q)$. Alice sends Bob only her message due to $\Pi$, implying that vcost$(\Pi')=v$. Thus, $\Pi'$ is an $(h+q,v)$-OMA-protocol as claimed. This completes the proof.  
\end{proof}

Our main result on the non-trivial OMA complexity $\OMAnt(\eqidx_{p,q})$ follows from the above two lemmas.

\mainthmeqidx*

\begin{proof}
    By \Cref{lem:onewayeqidx}, we have $\Rone(\eqidx)=\Theta(p)$ for this setting of $p$ and $q$. Therefore, any non-trivial OMA protocol $\Pi$ for the problem must have $\vc(\Pi)=o(p)$. But by \Cref{lem:eqidxmain}, if $\hc(\Pi)= O(q)$, then $\vc(\Pi)=\Omega(p)$. Therefore, $\hc(\Pi)$ must be $\omega(q)$, which means $\tc(\Pi)=\omega(q)$. Then, by definition, $\OMAnt(\eqidx) = \omega(q)$.   
\end{proof}

\expgapcor*

\begin{proof}
    Set $f=\eqidx_{p, q}$ with $p=\log n$ and $q=n/\log n$. Indeed, $p\geq \log q$. Hence, by \Cref{thm:eqidx-omant}, we have $\OMAnt(f)=\omega(q) = \exp(\Omega(p))$. 
    
    By \Cref{lem:onewayeqidx}, we have $\Rone(f)=\Theta(p)$. Also, by definition $\OMA(f)\leq \Rone(f)=\Theta(p)$. Therefore, the claimed result holds.
\end{proof}

\omegartncor*

\begin{proof}
    Set $f=\eqidx_{p, q}$ with $p=n/C$ and $q=C$. Since $q=C\leq n/\log n$, we have $$p=n/C \geq \log n > \log q$$  
    By \Cref{thm:eqidx-omant}, $\OMAnt(f)=\omega(C)$.  
\end{proof}

We also give a tight bound on the (standard) OMA complexity of $\eqidx_{p,q}$.

\begin{theorem}\label{thm:eqidx-oma}
     For any $p,q$ with $p\geq \log q$, we have $\OMA(\eqidx_{p,q}) = \Theta(\min\{\sqrt{pq}, p\})$. 
\end{theorem}

\begin{proof}
    First, we prove the lower bound. Given an $(h,v)$-OMA-protocol for $\eqidx_{p,q}$, first consider the case that $h>q$. Then, by
    \Cref{lem:eqidxmain}, we have $h\cdot v = \Omega(pq)$, which 
    means $h+v=\Omega(\sqrt{pq})$. Otherwise, i.e., if $h\leq q$, then by the same lemma,     we have $q\cdot v=\Omega(pq)$, which means $v=\Omega(p)$, 
    and hence, $h+v = \Omega(p)$. Therefore, we conclude 
    $$\OMA(\eqidx_{p,q}) = \Omega(\min\{\sqrt{pq}, p\}).$$

    For the upper bound, we first show that if $q\leq p$, then we can design a $(\sqrt{pq}, \sqrt{pq})$-OMA-protocol for $\eqidx_{p,q}$. Alice combines $x_1,\ldots,x_p$ into a single $pq$-bit string $x$, and then again splits $x$ into $\sqrt{pq}$ chunks $z_1,\ldots,z_{\sqrt{pq}}$, each of which has $\sqrt{pq}$ bits. Observe that since $q\leq p$, we have $q\leq \sqrt{pq}$, and hence, $x_j$ lies entirely in (the concatenation of) at most two consecutive chunks $z_k \circ z_{k+1}$. Merlin sends Bob $z_k$ and $z_{k+1}$ that he claims are identical to $z_k'$ and $z_{k+1}'$. This takes $O(\sqrt{pq})$ bits. If they are indeed as claimed, Bob can determine whether $y=x_j$. The problem now reduces to checking whether $z_k'=z_k$ and $z_{k+1}'=z_{k+1}$ without Merlin. To this end, it suffices to show that there exists a one-way randomized private-coin protocol of cost $O(\sqrt{pq})$. The protocol is very similar to the one mentioned in the proof of \Cref{lem:onewayeqidx}. An $O(\sqrt{pq})$ public-coin protocol can be obtained by having Alice send an $O(1)$-size equality sketch (\Cref{fact:eqsk}) for each $z_i$, while Bob uses only $z_k$ and $z_{k+1}$ for the check. Newman's theorem (\Cref{fact:Newman}) now gives a private-coin protocol of cost $O(\sqrt{pq} +\log (pq)) = O(\sqrt{pq})$. Therefore, we obtain an OMA protocol of total cost $O(\sqrt{pq})$ for $q\leq p$. 
    
    Finally, note that $\OMA(\eqidx_{p,q})$ is always trivially upper bounded by $\Rone(\eqidx_{p,q})$, which, by \Cref{lem:onewayeqidx}, is $O(p)$ for this setting of $p\geq \log q$. Hence, we can conclude that whenever $p\geq \log q$, 

    $$\OMA(\eqidx_{p,q}) \leq O(\min\{\sqrt{pq},p\}). \qedhere$$
\end{proof}

\subsection{Implications of the Equals-Index Lower Bounds }

\subsubsection{Sparse Index}

Here, we consider the $\spidx_{m,n}$ problem. This is the $\idx_n$ problem with the promise that Alice's string has hamming weight $m$. First, we prove a tight bound on its randomized one-way communication complexity. This tells us the benchmark for $\vc$ of a non-trivial OMA protocol for the problem. It also improves upon the previously best-known upper bound of $O(m\log m)$ proven by \cite{ChakrabartiCMT14}.  

\begin{lemma}\label{lem:onewayspidx}
$\Rone(\spidx_{m,n})=\Theta(m+\log\log n)$
\end{lemma}

\begin{proof}
First, we describe a public-coin protocol for $\spidx_{m,n}$ with cost $O(m)$. Using public randomness, Alice and Bob sample a function $h: [n] \to [3m]$ from a pairwise independent hash family $\cH$. Alice sends Bob the $3m$-size characteristic vector of the set $S=\{h(i) : x[i] =1, i\in [n]\}$. Bob checks whether $h(j)\in S$. If yes, he announces $x[j]=1$, and otherwise he announces $x[j]=0$. Observe that if $x[j]$ is indeed $1$, then $h(j)\in S$, and Bob is always correct. Otherwise, if $x[j]$ is actually $0$, Bob errs only when $h(j) = h(i)$ for some $i$ with $x[i]=1$. Fix such an $i$. By pairwise independence of $\cH$, we have
$$\Pr_{h\in \cH}[h(j)=h(i)] = \frac{1}{3m}.$$

By union bound over the $m$ values $i\in Y$, we get
$$\Pr_{h\in \cH}[\exists i\in Y: h(j)=h(i)] = \frac13.$$

Thus, the error probability of the protocol is $1/3$. Therefore, we get $\Ronepub(\spidx)=O(m)$.

To get a private coin protocol, we appeal to Newman's theorem (\Cref{fact:Newman}) once again. To obtain a tight bound, we need to have an appropriate interpretation of the problem. Observe that since the number of possible inputs to Alice is $\binom{n}{m}$, we can interpret the domain of $\spidx_{m,n}$ as $\{0,1\}^{\log \binom{n}{m}} \times \{0,1\}^{\log n}$. Thus, by \Cref{fact:Newman}, we get  

$$\Rone(\spidx_{m,n}) = O\left(m + \log \log \binom{n}{m}\right) = O(m + \log m + \log\log n) = O(m + \log\log n)$$

For the lower bound part, an  $\Omega(m)$ lower bound follows trivially by reduction from $\idx_m$. Again, an $\Omega(\log \log n)$ lower bound follows since $\Rone(\spidx_{1,n}) = \Omega(\log \log n)$: this is because $\spidx_{1,n} \equiv \eq_{\log n}$. Hence, we also get 

$$\Rone(\spidx_{m,n}) = \Omega(\max \{m, \log \log n\}) = \Omega(m+ \log \log n)$$ 

This completes the proof.    
\end{proof}

Our next lemma proves a lower bound for any OMA protocol that solves $\spidx$. We remark that the proof is similar to the one by \cite{ChakrabartiCMT14}, but they reduce from $\idx$, while we reduce from $\eqidx$. This lets us improve upon the previous lower bound.  

\begin{lemma}\label{lem:spidxmain}
    Any $(h,v)$-OMA-protocol for $\spidx_{m,n}$ must have $$(h+\log (n/m))\cdot v \geq \Omega(m\log (n/m))$$  
\end{lemma}

\begin{proof} 
We reduce from $\eqidx_{p,q}$. Given any instance of this problem, we construct an instance of $\spidx_{m,n}$ where $m=p$ and $n=2^qp$ (hence, $q=\log(n/m)$) such that an $(h,v)$-OMA-protocol that solves the $\spidx_{m,n}$ instance also solves the $\eqidx_{p,q}$ instance (without any additional hcost or vcost). The lower bound then follows from \Cref{lem:eqidxmain}. 

Recall that Alice's input to the $\eqidx_{p,q}$ problem is $(x_1,\ldots, x_p)$, where each $x_i\in \{0,1\}^q$.  Alice interprets each $x_i$ as the binary representation of an integer in $[0,2^q-1]$. She then constructs a string $z$ of length $n=2^qp$ indexed by pairs $(i,j)$ where $i\in [p]$ and $j\in [0,2^q-1]$. For each $i\in [p]$, she sets $z[i,x_i]=1$. She sets all the other bits of $z$ to $0$. Since $z$ has exactly $p$ 1's, its hamming weight is $m=p$. Again, Bob, who has inputs $y\in\{0,1\}^q$ and $j\in [p]$ to $\eqidx_{p,q}$, sets $k:= (j,y)$ interpreting $y$ as an integer in $[0,2^q-1]$. Alice and Bob now solve $\spidx_{m,n}$ with the inputs $z$ and $k$ respectively. Indeed, by construction, $z[k]=1$ if and only if $y=x_j$. Therefore, an $(h,v)$-OMA-protocol that solves $\spidx_{m,n}$ can be used to solve any instance of $\eqidx_{p,q}$ with $p=m$ and $q=\log (n/m)$.           
\end{proof}

Note that our bound is strictly better than Theorem 3.9 of \cite{ChakrabartiCMT14}, which (essentially) says that an $(h,v)$-OMA-protocol for $\spidx_{m,n}$ must have $(h+\log n)\cdot v \geq \Omega(m\log (n/m))$; this inequality is trivially implied by \Cref{lem:spidxmain}. Further, our strengthening of the additive factor from $\log n$ to $\log (n/m)$, combined with our tight classical one-way communication bound of $\spidx$, enables us to prove strong lower bounds on its non-trivial OMA complexity as we show below. 

Consider $\spidx_{m,n}$ for $m\geq \log \log n$. By \Cref{lem:onewayspidx}, its one-way randomized communication complexity is $\Theta(m)$. Therefore, any non-trivial OMA protocol for the problem must have vcost $o(m)$. But by \Cref{lem:spidxmain}, for any $(h,v)$-OMA-protocol for $\spidx$, $$h =  O(\log (n/m)) \Rightarrow v = \Omega(m)$$

Hence, any non-trivial OMA-protocol $\Pi$ must have $\hc(\Pi)=\omega(\log (n/m))$, which means $\OMAnt(\spidx_{m,n}) = \Omega(\log (n/m))$ whenever $m\geq \log \log n$. This gives us the following theorem.  

\begin{theorem} $\OMAnt(\spidx_{\log \log n, n})=\Omega(\log n)$. 
\end{theorem}

\begin{proof}
    Follows from above along with the fact that for $m=\log \log n$, we have $\Omega(\log (n/m))=\Omega(\log n)$.
\end{proof}

Compare this with the fact that $\OMA(\spidx_{\log \log n, n}) = \Rone(\spidx_{\log \log n, n})=  \Theta(\log \log n)$. Hence, we show an exponential separation between $\OMA$ and $\OMAnt$ complexity of $\spidx$ with sparsity $\log \log n$.

\subsubsection{Connectivity and Bipartite Testing}
Now we show that the lower bounds for the only two other functions, which were known to have $\OMAnt$ and $\OMA$ complexity (almost) as large as their $\Rone$ complexity, can be derived by reduction from $\eqidx$ and applying \Cref{lem:eqidxmain}. These are XOR-Connectivity and XOR-Bipartiteness, as studied in \cite{Thaler16}. We formally define the problems below. 

\noindent
$\xorconn_n$: Alice and Bob have graphs $G_A = (V, E_A)$ and $G_B = (V, E_B)$ on the same set of nodes $V$ with $|V|=n$. The graph $G$ is defined as $(V, E_A\oplus E_B)$. The goal is to determine whether $G$ is connected: Bob needs to output $1$ if it is and $0$ otherwise.

\noindent
$\xorbip_n$: Alice and Bob have graphs $G_A = (V, E_A)$ and $G_B = (V, E_B)$ on the same set of nodes $V$ with $|V|=n$. The goal is to determine whether $G:=(V, E_A\oplus E_B)$ is bipartite: Bob needs to output $1$ if it is, and $0$ otherwise.

As noted by \cite{Thaler16}, one can use a modified version of the connectivity and bipartiteness-sketches by \cite{AhnGM12} to obtain classical streaming $\tO(n)$-space algorithms for the streaming versions of \xorconn and \xorbip respectively. These algorithms can, in turn, be used to obtain $\tO(n)$-cost one-way randomized protocols for these problems in the classical communication model. This means $\Rone(\xorconn_n)$ and $\Rone(\xorbip_n)$ are both $\tO(n)$. Well-known communication lower bounds for connectivity and bipartiteness also imply that $\Rone(\xorconn_n)=\tTheta(n)$ and $\Rone(\xorbip_n)=\tTheta(n)$

\cite{Thaler16} showed that $\OMA(\xorconn_n)$ and $\OMA(\xorbip_n)$ are $\Omega(n)$. This follows from their more general result:

\begin{theorem}[\cite{Thaler16}, Theorem 4.1, restated]
    Any $(h,v)$-OMA-protocol solving $\xorconn_n$ or $\xorbip_n$ must have
    $$(h+n)\cdot v \geq \Omega(n^2).$$
\end{theorem}

We prove that this theorem also follows by reduction from $\eqidx_{n,n}$, thus adding to the usefulness of \eqidx in proving strong $\OMA$ lower bounds. 

\begin{proof}
    First, we prove the result for $\xorconn_{2n}$. The construction for $\xorbip_{2n}$ is very similar, and so we simply mention the modifications for that case.
    
    Given an instance of $\eqidx_{n,n}$, Alice first interprets her input $(x_1,\ldots, x_{n})$ where each $x_i\in \{0,1\}^{n}$ as the biadjacency matrix of a bipartite graph with $n$ nodes on each side. To be precise, Alice constructs the following bipartite graph. Let $V_L=\{u_1,\ldots,u_{n}\}$ and $V_R=\{v_1,\ldots,v_{n}\}$. For each $\ell\in [n]$, Alice joins the edges $\{\{u_\ell, v_r\}: x_{\ell}[r] = 1, r\in [n]\}$. This completes the description of the graph $G_A=(V, E_A)$ where $V= V_L\cup V_R\cup \{s\}$ and $E_A$ is the set of all edges added by Alice. Again, given his inputs $y\in \{0,1\}^{n}$ and $j\in [n]$, Bob joins the edges $\{\{u_j, v_r\}: y[r] = 1, r\in [n]\}$. He then joins vertex $s$ to all nodes in $V_L \cup V_R$ except $u_j$. 
    This completes the description of the graph $G_B=(V, E_B)$ where $E_B$ is the set of all edges added by Bob.

    Observe that if $y$ is indeed equal to $x_j$, then the neighborhoods of the vertex $u_j$ in $G_A$ and $G_B$ are identical by construction. Then, $u_j$ is an isolated vertex in $G=(V, E_A\oplus E_B)$, which means $G$ is not connected. Again, if $y\neq x_j$, then there must exist $r\in [n]$ such that $y[r]\neq x_j[r]$. By construction, exactly one of the sets $E_A$ and $E_B$ has the edge $e:=\{u_j, v_r\}$, and hence $e\in E_A\oplus E_B$. Therefore, in $G$, the node $u_j$ is connected to $v_r$, which in turn is connected to $s$. Observe that every other node $w \in V$ with $w\neq u_j$ is connected to $s$ in $G$: they are connected in $E_B$, while $E_A$ does not contain any such edge $\{s,w\}$. Thus, $s$ is connected to every other node in $G$, implying that $G$ is connected. Hence, $G$ is connected iff $y \neq x_j$. Therefore, given an $(h,v)$-OMA-protocol for $\xorconn_{2n}$, the players can run it for the constructed graph $G$, and Bob can determine the output to $\eqidx_{n,n}$. \Cref{lem:eqidxmain} now gives the desired result for $\xorconn$.

    For $\xorbip$, the construction of the graphs $G_A$ and $G_B$ are the same, \emph{except} Bob joins $s$ to $u_j$ and all vertices in $V_R$ only (instead of $V_L \cup V_R$). If $x_j=y$, then the neighborhoods of $u_j$ in $G_A$ and $G_B$ are the same, except $u_j$ is also adjacent to $s$ in $G_B$. Hence, in $G$, the only edge incident on $u_j$ is $\{u_j,s\}$. Therefore, we can bipartition $G$ into $(\{s\}\cup V_L\setminus \{u_j\}, V_R\cup \{u_j\})$. Observe that this is a valid bipartition since all edges in $G$ are between the left and right sides. Otherwise, if $x_j\neq y$, then by the same logic as in the case of \xorconn, there must exist $r\in [n]$ such that $\{u_j, v_r\} \in G$. Further, $\{s, v_r\}, \{u_j, s\} \in G$ since these edges are in $G_B$ but not $G_A$. Hence, $G$ has a triangle in this case and is not bipartite. Thus, we conclude that $G$ is bipartite iff the answer to $x_j = y$, which means an $(h,v)$-OMA-protocol for $\xorbip_{2n}$ can be used to determine the solution to any instance of $\eqidx_{n,n}$. The result again follows from \Cref{lem:eqidxmain}.

\end{proof}
\begin{figure}[!ht]
    \begin{subfigure}[t]{0.33\textwidth}
	\centering
	{  \vspace{-4.2cm} \resizebox{80pt}{90pt}{ \tikzset{every picture/.style={line width=0.75pt}} 

\begin{tikzpicture}[x=0.75pt,y=0.75pt,yscale=-1,xscale=1]

\draw   (153,65) .. controls (153,61.69) and (155.69,59) .. (159,59) .. controls (162.31,59) and (165,61.69) .. (165,65) .. controls (165,68.31) and (162.31,71) .. (159,71) .. controls (155.69,71) and (153,68.31) .. (153,65) -- cycle ;
\draw   (153,95) .. controls (153,91.69) and (155.69,89) .. (159,89) .. controls (162.31,89) and (165,91.69) .. (165,95) .. controls (165,98.31) and (162.31,101) .. (159,101) .. controls (155.69,101) and (153,98.31) .. (153,95) -- cycle ;
\draw   (153,125) .. controls (153,121.69) and (155.69,119) .. (159,119) .. controls (162.31,119) and (165,121.69) .. (165,125) .. controls (165,128.31) and (162.31,131) .. (159,131) .. controls (155.69,131) and (153,128.31) .. (153,125) -- cycle ;
\draw   (153,166) .. controls (153,162.69) and (155.69,160) .. (159,160) .. controls (162.31,160) and (165,162.69) .. (165,166) .. controls (165,169.31) and (162.31,172) .. (159,172) .. controls (155.69,172) and (153,169.31) .. (153,166) -- cycle ;
\draw   (203,65) .. controls (203,61.69) and (205.69,59) .. (209,59) .. controls (212.31,59) and (215,61.69) .. (215,65) .. controls (215,68.31) and (212.31,71) .. (209,71) .. controls (205.69,71) and (203,68.31) .. (203,65) -- cycle ;
\draw   (203,95) .. controls (203,91.69) and (205.69,89) .. (209,89) .. controls (212.31,89) and (215,91.69) .. (215,95) .. controls (215,98.31) and (212.31,101) .. (209,101) .. controls (205.69,101) and (203,98.31) .. (203,95) -- cycle ;
\draw   (203,125) .. controls (203,121.69) and (205.69,119) .. (209,119) .. controls (212.31,119) and (215,121.69) .. (215,125) .. controls (215,128.31) and (212.31,131) .. (209,131) .. controls (205.69,131) and (203,128.31) .. (203,125) -- cycle ;
\draw   (203,165) .. controls (203,161.69) and (205.69,159) .. (209,159) .. controls (212.31,159) and (215,161.69) .. (215,165) .. controls (215,168.31) and (212.31,171) .. (209,171) .. controls (205.69,171) and (203,168.31) .. (203,165) -- cycle ;
\draw    (165,65) -- (203,65) ;
\draw    (165,65) -- (203,95) ;
\draw    (165,65) -- (203,166) ;
\draw    (165,95) -- (203,95) ;
\draw    (165,125) -- (203,125) ;
\draw  [dashed]  (165,95) -- (203,125) ;
\draw    (165,165) -- (203,125) ;
\draw    (165,165) -- (203,65) ;
\draw    (165,125) -- (203,95) ;

\draw (145,42) node [anchor=north west][inner sep=0.75pt]   [align=left] {$n$};
\draw (210,42) node [anchor=north west][inner sep=0.75pt]   [align=left] {$n$};
\draw (135,90) node [anchor=north west][inner sep=0.75pt]   [align=left] {$u_j$};
\draw (220,120) node [anchor=north west][inner sep=0.75pt]   [align=left] {$v_r$};

\draw (206,130) node [anchor=north west][inner sep=0.75pt]   [align=left] {\small{\vdots}};
\draw (155,130) node [anchor=north west][inner sep=0.75pt]   [align=left] {\small{\vdots}};

\end{tikzpicture}}  }
	\vspace{1cm}
 \caption{Alice uses $x_i$ to construct the neighborhood of vertex $i$.
		\label{fig:connectivity1}}
    \end{subfigure}
    \begin{subfigure}[t]{0.33\textwidth}
	\centering
        {  \vspace{-4.2cm} \resizebox{80pt}{90pt}{ \tikzset{every picture/.style={line width=0.75pt}} 

\begin{tikzpicture}[x=0.75pt,y=0.75pt,yscale=-1,xscale=1]

\draw   (153,65) .. controls (153,61.69) and (155.69,59) .. (159,59) .. controls (162.31,59) and (165,61.69) .. (165,65) .. controls (165,68.31) and (162.31,71) .. (159,71) .. controls (155.69,71) and (153,68.31) .. (153,65) -- cycle ;
\draw   (153,95) .. controls (153,91.69) and (155.69,89) .. (159,89) .. controls (162.31,89) and (165,91.69) .. (165,95) .. controls (165,98.31) and (162.31,101) .. (159,101) .. controls (155.69,101) and (153,98.31) .. (153,95) -- cycle ;
\draw   (153,125) .. controls (153,121.69) and (155.69,119) .. (159,119) .. controls (162.31,119) and (165,121.69) .. (165,125) .. controls (165,128.31) and (162.31,131) .. (159,131) .. controls (155.69,131) and (153,128.31) .. (153,125) -- cycle ;
\draw   (153,166) .. controls (153,162.69) and (155.69,160) .. (159,160) .. controls (162.31,160) and (165,162.69) .. (165,166) .. controls (165,169.31) and (162.31,172) .. (159,172) .. controls (155.69,172) and (153,169.31) .. (153,166) -- cycle ;
\draw   (203,65) .. controls (203,61.69) and (205.69,59) .. (209,59) .. controls (212.31,59) and (215,61.69) .. (215,65) .. controls (215,68.31) and (212.31,71) .. (209,71) .. controls (205.69,71) and (203,68.31) .. (203,65) -- cycle ;
\draw   (203,95) .. controls (203,91.69) and (205.69,89) .. (209,89) .. controls (212.31,89) and (215,91.69) .. (215,95) .. controls (215,98.31) and (212.31,101) .. (209,101) .. controls (205.69,101) and (203,98.31) .. (203,95) -- cycle ;
\draw   (203,125) .. controls (203,121.69) and (205.69,119) .. (209,119) .. controls (212.31,119) and (215,121.69) .. (215,125) .. controls (215,128.31) and (212.31,131) .. (209,131) .. controls (205.69,131) and (203,128.31) .. (203,125) -- cycle ;
\draw   (203,165) .. controls (203,161.69) and (205.69,159) .. (209,159) .. controls (212.31,159) and (215,161.69) .. (215,165) .. controls (215,168.31) and (212.31,171) .. (209,171) .. controls (205.69,171) and (203,168.31) .. (203,165) -- cycle ;
\draw  [draw = red]  (165,95) -- (203,95) ;
\draw   [dashed, draw = red] (165,95) -- (203,125) ;

\draw (145,42) node [anchor=north west][inner sep=0.75pt]   [align=left] {$n$};
\draw (210,42) node [anchor=north west][inner sep=0.75pt]   [align=left] {$n$};
\draw (135,90) node [anchor=north west][inner sep=0.75pt]   [align=left] {$u_j$};
\draw (220,120) node [anchor=north west][inner sep=0.75pt]   [align=left] {$v_r$};

\draw (168,74) node [anchor=north west][inner sep=0.75pt]   [align=left] {$N(u_j)$};

\draw (206,130) node [anchor=north west][inner sep=0.75pt]   [align=left] {\small{\vdots}};
\draw (155,130) node [anchor=north west][inner sep=0.75pt]   [align=left] {\small{\vdots}};

\end{tikzpicture}}  }
\vspace{1cm}	
 \caption{Bob constructs $N(u_j)$ based on $y$.
		\label{fig:connectivity2}}
  \end{subfigure}
    \begin{subfigure}[t]{0.33\textwidth}
	\centering
	{   \resizebox{130pt}{120pt}{ \tikzset{every picture/.style={line width=0.75pt}} 

\begin{tikzpicture}[x=0.75pt,y=0.75pt,yscale=-1,xscale=1]

\draw   (153,65) .. controls (153,61.69) and (155.69,59) .. (159,59) .. controls (162.31,59) and (165,61.69) .. (165,65) .. controls (165,68.31) and (162.31,71) .. (159,71) .. controls (155.69,71) and (153,68.31) .. (153,65) -- cycle ;
\draw   (153,95) .. controls (153,91.69) and (155.69,89) .. (159,89) .. controls (162.31,89) and (165,91.69) .. (165,95) .. controls (165,98.31) and (162.31,101) .. (159,101) .. controls (155.69,101) and (153,98.31) .. (153,95) -- cycle ;
\draw   (153,125) .. controls (153,121.69) and (155.69,119) .. (159,119) .. controls (162.31,119) and (165,121.69) .. (165,125) .. controls (165,128.31) and (162.31,131) .. (159,131) .. controls (155.69,131) and (153,128.31) .. (153,125) -- cycle ;
\draw   (153,166) .. controls (153,162.69) and (155.69,160) .. (159,160) .. controls (162.31,160) and (165,162.69) .. (165,166) .. controls (165,169.31) and (162.31,172) .. (159,172) .. controls (155.69,172) and (153,169.31) .. (153,166) -- cycle ;
\draw   (203,65) .. controls (203,61.69) and (205.69,59) .. (209,59) .. controls (212.31,59) and (215,61.69) .. (215,65) .. controls (215,68.31) and (212.31,71) .. (209,71) .. controls (205.69,71) and (203,68.31) .. (203,65) -- cycle ;
\draw   (203,95) .. controls (203,91.69) and (205.69,89) .. (209,89) .. controls (212.31,89) and (215,91.69) .. (215,95) .. controls (215,98.31) and (212.31,101) .. (209,101) .. controls (205.69,101) and (203,98.31) .. (203,95) -- cycle ;
\draw   (203,125) .. controls (203,121.69) and (205.69,119) .. (209,119) .. controls (212.31,119) and (215,121.69) .. (215,125) .. controls (215,128.31) and (212.31,131) .. (209,131) .. controls (205.69,131) and (203,128.31) .. (203,125) -- cycle ;
\draw   (203,165) .. controls (203,161.69) and (205.69,159) .. (209,159) .. controls (212.31,159) and (215,161.69) .. (215,165) .. controls (215,168.31) and (212.31,171) .. (209,171) .. controls (205.69,171) and (203,168.31) .. (203,165) -- cycle ;
\draw  [dashed]  (165,95) -- (203,125) ;
\draw  [draw = blue] (178,195) .. controls (178,191.69) and (180.69,189) .. (184,189) .. controls (187.31,189) and (190,191.69) .. (190,195) .. controls (190,198.31) and (187.31,201) .. (184,201) .. controls (180.69,201) and (178,198.31) .. (178,195) -- cycle ;
\draw [draw = blue]   (190,195) .. controls (246.38,183.27) and (250.38,124.27) .. (215,95) ;
\draw  [draw = blue]  (159,172) -- (178,195) ;
\draw  [draw = blue]  (165,125) -- (184,189) ;
\draw  [draw = blue]  (190,195) -- (203,165) ;
\draw   [draw = blue] (190,195) .. controls (255.38,221.27) and (268.38,66.27) .. (215,65) ;
\draw  [draw = blue]  (178,195) .. controls (101.38,170.27) and (113,95) .. (153,65) ;
\draw  [draw = blue]  (203,125) -- (184,189) ;

\draw (145,42) node [anchor=north west][inner sep=0.75pt]   [align=left] {$n$};
\draw (210,42) node [anchor=north west][inner sep=0.75pt]   [align=left] {$n$};
\draw (135,90) node [anchor=north west][inner sep=0.75pt]   [align=left] {$u_j$};
\draw (220,120) node [anchor=north west][inner sep=0.75pt]   [align=left] {$v_r$};

\draw (206,130) node [anchor=north west][inner sep=0.75pt]   [align=left] {\small{\vdots}};
\draw (155,130) node [anchor=north west][inner sep=0.75pt]   [align=left] {\small{\vdots}};

\end{tikzpicture}}  }
	\caption{Bob created an additional vertex and connects all vertices except $u_j$ to it.
		\label{fig:connectivity3}}
  \end{subfigure}
\end{figure}

\subsubsection{Distinct Items}

Here, we discuss lower bounds for the well-studied distinct items problem in the turnstile model. We formally define the problem as follows.

\noindent
$\distit_{N,F}$: Given a turnstile stream of elements in $[N]$ such that $\max\{|\freq(j)|:j\in [N]\} \leq F$, output the number of items with non-zero frequency, i.e., $|\{j\in [N]: \freq(j) \neq 0 \}|$

\begin{lemma}\label{lem:distit-main}
    Any $(h,v)$-scheme for $\distit_{N,F}$ must have 
    $$(h+\log F)\cdot v\geq \Omega(N \log F)$$
\end{lemma}

\begin{proof}
   Given an $(h,v)$-scheme $\alg$ for $\distit_{N,F}$, we design an $(h,v)$-OMA-protocol for $\eqidx_{N, \log F}$. In an instance of this problem, Alice has $x_1,\ldots,x_N$, where each $x_i\in \{0,1\}^{\log F}$. She interprets each $x_i$ as (the binary representation of) an integer in $[0,F-1]$ and feeds $x_i+1$ insertions of the element $i$ to $\alg$. Observe that $1\leq \freq(i)\leq F$ for each $i$ in the stream. Alice then sends Bob the memory state of $\alg$. Bob has inputs $y\in \{0,1\}^{\log F}$ and $j\in [N]$. Again, he interprets $y$ as an integer in $[0,F-1]$ and continues the run of $\alg$ by appending the stream with $y+1$ deletions of the item $j$. Merlin sends Bob the proof that the Prover would send in $\alg$. Bob can now run the verification part of $\alg$ and get an output. If the answer is $n-1$, then he announces that $x_j=y$. Otherwise, he declares $x_j\neq y$. For the correctness, observe that Alice added the item $j$ to the stream $x_j+1$ times and Bob deleted it $y+1$ times. Therefore, $\freq(j)=0$ iff $x_j=y$. All other stream elements have frequency at least $1$. Hence, the number of items with non-zero frequency in the stream is $N-1$ if $x_j=y$, and $N$ otherwise.  
\end{proof}

Again, we need to prove a tight bound on the classical streaming complexity of $\distit$ so that we know what it means to be a non-trivial scheme for the problem. Using the non-zero detectors that we build in \Cref{subsec:special-counters}, we can design an efficient classical streaming algorithm for $\distit_{N,F}$ for large $F$.

\begin{lemma}\label{lem:distit-classtr}
    There exists a classical streaming algorithm using $O(N \log \log F)$ space for $\distit_{N,F}$. Again, any such algorithm needs $\Omega(N)$ space.  
\end{lemma}

\begin{proof}
    We can keep one non-zero detector given by \Cref{cor:nonzdetect} for each element in $[N]$, and can figure out the answer. by \Cref{cor:nonzdetect}, we need $\tO(N\log \log F)$ space and this succeeds with probability $1-1/\polylog(F)$ by union bound over the $N$ detectors. The lower bound follows by a simple reduction from $\idx_N$, even for insert-only streams: Alice inserts $\{i: x_i=1\}$ and sends Bob the hamming weight of $x$. Bob inserts $j$ and checks whether or not the number of distinct elements in the stream exceeds the hamming weight of $x$. If yes, $x_j=0$, and if not $x_j$ must be $1$.  
\end{proof}

Given the above lemma, we can prove our lower bound on the non-trivial total cost of $\distit_{n,F}$.

\begin{lemma}\label{lem:disit-nttc}
    Any non-trivial scheme for $\distit_{N,F}$ with $F\leq \exp(N^{\polylog(N)})$ must have total cost $\omega(\log F)$. 
\end{lemma}

\begin{proof}
By \Cref{lem:distit-classtr}, the classical streaming complexity of $\distit_{N,F}$ is $\tTheta(N)$ for $F\leq \exp(N^{\tTheta(1)})$. Hence, any non-trivial scheme $\Pi$ for the problem must have $\vc(\Pi) = o(N)$. By \Cref{lem:distit-main}, $\hc(\Pi)=O(\log F)$ implies $\vc(\Pi) = \Omega(N)$. Hence, $\hc(\Pi)=\omega(\log F)$, which means $\tc(\Pi)$ must be $\omega(\log F)$. 
\end{proof}

This establishes a separation between the classical streaming complexity and the (non-trivial) annotated streaming complexity of $\distit$ and proves \Cref{thm:distitgap}. 

\distitmain*




\subsubsection{Connectivity Problems in Support Graph Turnstile Streams}

We first make the following claim about classical streaming complexities of connectivity and $k$-connectivity.

\begin{claim}\label{clm:connclasstream}
    In the classical streaming model under SGT streams, graph connectivity can be solved in $\tO(n)$ space, and each of the $k$-vertex-connectivity and $k$-edge-connectivity problems can be solved in $\tO(kn)$ space. These bounds match (up to polylogarithmic factors) the space-bound known for each problem under dynamic graph streams. 
\end{claim}

\begin{proof}
The upper bounds of connectivity and $k$-vertex-connectivity follow by using the \cite{AhnGM12} and \cite{assadi2023tight} algorithms respectively, and replacing their $\ell_0$ samplers with our strong $\ell_0$ samplers given by \Cref{lem:L0}. For the $k$-edge connectivity problem, it is not clear that the algorithm by \cite{AhnGM12} can be implemented under SGT streams. So we design a new $\tO(kn)$-space $k$-edge-connectivity algorithm under such streams, presented in \Cref{sec:edgeconnsgt}. The claim then follows. 
\end{proof}

We show a lower bound for connectivity in this model. First we prove the following lemma.
\begin{lemma}
    Any $(h,v)$-scheme for connectivity satisfies $(h+ q n) \cdot v \geq n^2 \cdot q$. In particular, when $q= n^{\polylog(n)}$, for $v=o(n)$ we need $h\geq n^{\polylog(n)}$.
\end{lemma}
The proof of this lemma is along the same lines as the proof of XOR-Connectivity except that we use $\eqidx_{n,q n}$.
Alice and Bob construct the same graph, but now for each pair of vertices there could be multiple edge insertions (up to $2^q$).

We now prove lower bounds for $k$-vertex-connectivity and $k$-edge-connectivity in this model. 
\begin{lemma}
    Any $(h,v)$-scheme for $k$-vertex-connectivity or $k$-edge-connectivity satisfies $(h+q) \cdot v \geq kn \cdot q$. In particular, when $q= 2^{\polylog(n)}$, for $v=o(kn)$ we need $h\geq 2^{\polylog(n)}$.
\end{lemma}

Say that we are given an $(h,v)$-scheme for $k$-vertex-connectivity or $k$-edge-connectivity in the Support Graph Turnstile streaming model.
Given an instance of $\eqidx_{kn,q}$, Alice constructs the following graph stream. Let $V_L=\{u_1,\ldots,u_{n}\}$ and $V_R=\{v_1,\ldots,v_{k}\}$. For each $\ell\in [kn]$, Alice inserts an edge between the $\ell^{th}$ pair of vertices $x_\ell$ times. 
Given his inputs $y\in \{0,1\}^{q}$ and $j\in [kn]$, Bob creates the following graph stream. Bob deletes an edge between the $j^{th}$ pair of vertices $y$ times. Bob also adds an edge between every pair of vertices except the $j^{th}$ pair.
Alice runs the $(h,v)$-scheme on her graph stream and sends the memory content to Bob, who then continues running the scheme on his graph stream. Iff the graph is $k$-connected, then output $0$ for the $\eqidx_{kn,q}$ instance and output $1$ otherwise. 
Note that the complete bipartite graph is $k$-connected, but even if you remove one edge it is not $k$-connected.
The correctness follows from the following claim.
\begin{claim}\label{clm:k-conn-LB-corr}
    The graph formed by the stream is $k$-connected in the support graph turnstile model iff the $\eqidx_{kn,q}$ instance has value $0$.
\end{claim}
\begin{proof}
    Consider the case when the $\eqidx_{kn,q}$ instance has value $1$.
    We have $x_j=y$, so the $j^{th}$ pair of vertices has $y$ edge insertions and $y$ edge deletions. This means that there is no edge between the $j^{th}$ pair of vertices, implying that the support graph is not $k$-connected.

    Now consider the case when the $\eqidx_{kn,q}$ instance has value $0$.
    We have $x_j \neq y$, so the $j^{th}$ pair of vertices has a different number of edge insertions and edge deletions. This implies that the $j^{th}$ pair of vertices has a non-zero support implying that the $j^{th}$ edge exists.
    Every other pair of vertices has only edge insertions, including exactly $1$ edge insertion by Bob. This implies that every other pair of vertices has a non-zero support implying that the edge exists.
    Thus, the support graph is a complete bipartite graph implying that the support graph is $k$-connected.
\end{proof}

\Cref{clm:k-conn-LB-corr} shows that this protocol solves $\eqidx_{kn,q}$.
We know by \Cref{lem:eqidxmain} that the following holds for any protocol that solves $\eqidx_{kn,q}$:
\begin{equation*}
    (h+q) \cdot v \geq kn \cdot q.
\end{equation*}
The only help from Merlin is $h$ bits for the $(h,v)$-scheme. The only communication from Alice is sending the memory content that takes $v$ bits of space.
This bound implies that even if $h=O(q)$, $v=\Omega(kn)$. So for $v=o(kn)$ we need $h=\omega(q)$. Setting $q= 2^{\polylog(n)}$ implies that for $v=o(kn)$ we need $h\geq 2^{\polylog(n)}$. However, using the strong $\ell_0$-Samplers gives us $h=0$ and $v=\Ot(kn)$.

\begin{figure}[ht!]
    \begin{subfigure}[t]{0.45\textwidth}
	\centering
	{  \resizebox{80pt}{90pt}{ \tikzset{every picture/.style={line width=0.75pt}} 

\begin{tikzpicture}[x=0.75pt,y=0.75pt,yscale=-1,xscale=1]

\draw   (153,65) .. controls (153,61.69) and (155.69,59) .. (159,59) .. controls (162.31,59) and (165,61.69) .. (165,65) .. controls (165,68.31) and (162.31,71) .. (159,71) .. controls (155.69,71) and (153,68.31) .. (153,65) -- cycle ;
\draw   (153,95) .. controls (153,91.69) and (155.69,89) .. (159,89) .. controls (162.31,89) and (165,91.69) .. (165,95) .. controls (165,98.31) and (162.31,101) .. (159,101) .. controls (155.69,101) and (153,98.31) .. (153,95) -- cycle ;
\draw   (153,125) .. controls (153,121.69) and (155.69,119) .. (159,119) .. controls (162.31,119) and (165,121.69) .. (165,125) .. controls (165,128.31) and (162.31,131) .. (159,131) .. controls (155.69,131) and (153,128.31) .. (153,125) -- cycle ;
\draw   (153,166) .. controls (153,162.69) and (155.69,160) .. (159,160) .. controls (162.31,160) and (165,162.69) .. (165,166) .. controls (165,169.31) and (162.31,172) .. (159,172) .. controls (155.69,172) and (153,169.31) .. (153,166) -- cycle ;
\draw   (203,65) .. controls (203,61.69) and (205.69,59) .. (209,59) .. controls (212.31,59) and (215,61.69) .. (215,65) .. controls (215,68.31) and (212.31,71) .. (209,71) .. controls (205.69,71) and (203,68.31) .. (203,65) -- cycle ;
\draw   (203,95) .. controls (203,91.69) and (205.69,89) .. (209,89) .. controls (212.31,89) and (215,91.69) .. (215,95) .. controls (215,98.31) and (212.31,101) .. (209,101) .. controls (205.69,101) and (203,98.31) .. (203,95) -- cycle ;
\draw   (203,125) .. controls (203,121.69) and (205.69,119) .. (209,119) .. controls (212.31,119) and (215,121.69) .. (215,125) .. controls (215,128.31) and (212.31,131) .. (209,131) .. controls (205.69,131) and (203,128.31) .. (203,125) -- cycle ;
\draw   (203,165) .. controls (203,161.69) and (205.69,159) .. (209,159) .. controls (212.31,159) and (215,161.69) .. (215,165) .. controls (215,168.31) and (212.31,171) .. (209,171) .. controls (205.69,171) and (203,168.31) .. (203,165) -- cycle ;
\draw    (165,65) -- (203,65) ;
\draw    (165,65) -- (203,95) ;
\draw    (165,65) -- (203,166) ;
\draw    (165,95) -- (203,95) ;
\draw    (165,125) -- (203,125) ;
\draw  [dashed]  (165,95) -- (203,125) ;
\draw    (165,165) -- (203,125) ;
\draw    (165,165) -- (203,65) ;
\draw    (165,125) -- (203,95) ;

\draw (145,42) node [anchor=north west][inner sep=0.75pt]   [align=left] {$n$};
\draw (210,42) node [anchor=north west][inner sep=0.75pt]   [align=left] {$k$};
\draw (135,90) node [anchor=north west][inner sep=0.75pt]   [align=left] {$u$};
\draw (220,120) node [anchor=north west][inner sep=0.75pt]   [align=left] {$v$};

\draw (206,130) node [anchor=north west][inner sep=0.75pt]   [align=left] {\small{\vdots}};
\draw (155,130) node [anchor=north west][inner sep=0.75pt]   [align=left] {\small{\vdots}};

\end{tikzpicture}}  }
 \caption{Alice adds $x_i$ edges for pair $i$.
		\label{fig:kconn1}}
    \end{subfigure}
    \begin{subfigure}[t]{0.45\textwidth}
	\centering
	{   \resizebox{80pt}{90pt}{ \tikzset{every picture/.style={line width=0.75pt}} 

\begin{tikzpicture}[x=0.75pt,y=0.75pt,yscale=-1,xscale=1]

\draw   (153,65) .. controls (153,61.69) and (155.69,59) .. (159,59) .. controls (162.31,59) and (165,61.69) .. (165,65) .. controls (165,68.31) and (162.31,71) .. (159,71) .. controls (155.69,71) and (153,68.31) .. (153,65) -- cycle ;
\draw   (153,95) .. controls (153,91.69) and (155.69,89) .. (159,89) .. controls (162.31,89) and (165,91.69) .. (165,95) .. controls (165,98.31) and (162.31,101) .. (159,101) .. controls (155.69,101) and (153,98.31) .. (153,95) -- cycle ;
\draw   (153,125) .. controls (153,121.69) and (155.69,119) .. (159,119) .. controls (162.31,119) and (165,121.69) .. (165,125) .. controls (165,128.31) and (162.31,131) .. (159,131) .. controls (155.69,131) and (153,128.31) .. (153,125) -- cycle ;
\draw   (153,166) .. controls (153,162.69) and (155.69,160) .. (159,160) .. controls (162.31,160) and (165,162.69) .. (165,166) .. controls (165,169.31) and (162.31,172) .. (159,172) .. controls (155.69,172) and (153,169.31) .. (153,166) -- cycle ;
\draw   (203,65) .. controls (203,61.69) and (205.69,59) .. (209,59) .. controls (212.31,59) and (215,61.69) .. (215,65) .. controls (215,68.31) and (212.31,71) .. (209,71) .. controls (205.69,71) and (203,68.31) .. (203,65) -- cycle ;
\draw   (203,95) .. controls (203,91.69) and (205.69,89) .. (209,89) .. controls (212.31,89) and (215,91.69) .. (215,95) .. controls (215,98.31) and (212.31,101) .. (209,101) .. controls (205.69,101) and (203,98.31) .. (203,95) -- cycle ;
\draw   (203,125) .. controls (203,121.69) and (205.69,119) .. (209,119) .. controls (212.31,119) and (215,121.69) .. (215,125) .. controls (215,128.31) and (212.31,131) .. (209,131) .. controls (205.69,131) and (203,128.31) .. (203,125) -- cycle ;
\draw   (203,165) .. controls (203,161.69) and (205.69,159) .. (209,159) .. controls (212.31,159) and (215,161.69) .. (215,165) .. controls (215,168.31) and (212.31,171) .. (209,171) .. controls (205.69,171) and (203,168.31) .. (203,165) -- cycle ;
\draw  [line width= 0.01mm, draw=blue!50]  (165,65) -- (203,65) ;
\draw  [line width= 0.01mm, draw=blue!50]  (165,65) -- (203,95) ;
\draw  [line width= 0.01mm, draw=blue!50]  (165,65) -- (203,166) ;
\draw  [line width= 0.01mm, draw=blue!50]  (165,95) -- (203,95) ;
\draw  [line width= 0.01mm, draw=blue!50]  (165,125) -- (203,125) ;
\draw  [line width= 2pt, draw=red]  (165,95) -- (203,125) ;
\draw   [line width= 0.01mm, draw=blue!50] (165,165) -- (203,125) ;
\draw   [line width= 0.01mm, draw=blue!50] (165,165) -- (203,65) ;
\draw  [line width= 0.01mm, draw=blue!50]  (165,125) -- (203,95) ;

\draw  [line width= 0.01mm, draw=blue!50]  (165,65) -- (203,125) ;

\draw  [line width= 0.01mm, draw=blue!50]  (165,95) -- (203,65) ;
\draw   [line width= 0.01mm, draw=blue!50] (165,95) -- (203,165) ;

\draw [line width= 0.01mm, draw=blue!50]   (165,125) -- (203,65) ;
\draw   [line width= 0.01mm, draw=blue!50] (165,125) -- (203,165) ;

\draw  [line width= 0.01mm, draw=blue!50]  (165,165) -- (203,95) ;
\draw [line width= 0.01mm, draw=blue!50]   (165,165) -- (203,165) ;

\draw (145,42) node [anchor=north west][inner sep=0.75pt]   [align=left] {$n$};
\draw (210,42) node [anchor=north west][inner sep=0.75pt]   [align=left] {$k$};
\draw (135,90) node [anchor=north west][inner sep=0.75pt]   [align=left] {$u$};
\draw (220,120) node [anchor=north west][inner sep=0.75pt]   [align=left] {$v$};

\draw (206,130) node [anchor=north west][inner sep=0.75pt]   [align=left] {\small{\vdots}};
\draw (155,130) node [anchor=north west][inner sep=0.75pt]   [align=left] {\small{\vdots}};

\end{tikzpicture}}  }
	\caption{Bob deletes $(u,v)$ $y$ times and adds every other pair once.
		\label{fig:kconn3}}
  \end{subfigure}
\end{figure}

\section{Classical Streaming Algorithms under SGT streams}\label{sec:SGT-streaming-algs}


\subsection{Basic Tools: Non-Zero Detection and \texorpdfstring{$\ell_0$}{L0}-Sampling for Large Frequencies}\label{subsec:special-counters}

In this section, we introduce an $\ell_0$-sampler that works when the entries in the vector are exponential 
in size.

\subsubsection{Non-zero detection}

Consider the following simple problem.

\begin{problem}\label{prob:special-counter}
Given a stream containing $+$'s and $-$'s output whether the number of $+$'s is exactly equal to 
the number of $-$'s or not.
The promise is that the difference between the number of $+$'s and $-$'s is less than $\val$ at the 
end of 
the stream.	
\end{problem}

A simple solution to \Cref{prob:special-counter} is to store a deterministic counter mod $\val$ which 
takes 
$O(\log 
\val)$ bits of space.
It is easy to show that we cannot do better deterministically.
\begin{claim}\label{clm:det-counter-LB}
	Any deterministic algorithm for \Cref{prob:special-counter} needs $\log \val$ bits of space.
\end{claim}
\begin{proof}
	Assume towards a contradiction that there is a deterministic algorithm that solves 
	\Cref{prob:special-counter} in 
	less than $\log \val$ bits of space.
	Consider the different states of memory when we have exactly $i$ $+$'s in the input stream.
	Since the memory used is less than $\log \val$ bits, there exist $i$ and $j$ (wlog $i<j$) such that 
	the state 
	of memory after inserting $i$ times is the same as the state of memory after inserting $j$ times.
	Now consider $i$ $-$'s appear in the stream after the $+$'s.
	In the case where there were $i$ $+$'s, the answer is $0$, but the answer in the other case is $1$.
	However, the algorithm cannot differentiate between the two cases because the memory state was the same 
	after seeing the $+$'s, giving us a contradiction.
\end{proof}

Now the question is whether randomized algorithms can do better. 
\Cref{clm:det-counter-LB} implies a lower bound of $\Omega(\log \log \val)$ bits for the private 
randomness 
communication complexity of this problem.
There exists a randomized algorithm for this problem that uses $O(\log \log \val)$ bits 
of 
space.

\begin{lemma}
	There exists a randomized algorithm that solves \Cref{prob:special-counter} with probability 
	$1-1/\polylog(\val)$ 
	and uses $O(\log \log \val)$ bits of space.
\end{lemma}
\begin{proof}
	Let $p$ be a random prime within the first $\polylog(\val)$ primes.
	Storing $p$ takes $O(\log \log \val)$ bits of space by the prime density theorem.
	We store a counter mod $p$ to solve \Cref{prob:special-counter}.
	
	Let the counter value be $\val'<\val$.
	We make a mistake if our counter is $0$, which happens when $p$ is a factor of $\val'$.
	Observe that $\val'$ has at most $\log \val$ prime factors because the product of all prime factors 
	is 
	at most 
	$\val'<\val$ and all prime factors are at least $2$.
	
	So the probability that $p$ divides $\val'$ is at most $\log \val/\polylog(\val)$ implying that the 
	probability of 
	success is at least $1-1/\polylog(\val)$.
\end{proof}

If we just slightly change the proof and let $p$ be a random prime within the first $\polylog(\val) \cdot \poly(n)$ primes, then we get the following:

\begin{corollary}\label{cor:nonzdetect}
	There exists a randomized algorithm that, given a parameter $n$, solves 
	\Cref{prob:special-counter} with probability $1-1/\polylog(\val) \cdot \poly(n)$ 
	and uses $O(\log \log \val + \log n)$ bits of space.
\end{corollary}

\subsubsection{Strong \texorpdfstring{$\ell_0$}{L0}-samplers}\label{subsec:Strong-L0}


In this subsection, we design an $\ell_0$-sampler that can handle large values.
Let the universe of elements be $[n]$.
Consider the following problem:
\begin{problem}\label{prob:special-L0sampling}
	Given a stream containing insertions and deletions of elements in $[n]$, output an element whose 
	frequency is non-zero. 
	The promise is that the value of each coordinate is between $-\val$ and $\val$ at the end of the 
	stream.
\end{problem}

The input can be thought of as a vector, so we will go back and forth between thinking of the input as a vector and a stream of elements.
We know how to solve this problem using $\polylog(n)$ space when $\alpha=\poly(n)$.
We now use the decision counters from \Cref{subsec:special-counters} to solve 
\Cref{prob:special-L0sampling}.

\begin{lemma}\label{lem:L0}
	There is a randomized algorithm that solves 
	\Cref{prob:special-L0sampling} with probability $1-1/\polylog(\val) \cdot \poly(n)$ 
	and uses $\poly(\log \log \val + \log n)$ bits of space.
\end{lemma}

We first solve the problem when the support size is $1$.
The idea is to compute inner products with vectors that reveal the bits of the position of the 
non-zero element. This idea is called non-adaptive binary search.

Consider the set $A$ of $\log n$ vectors of size $n$ where the $i^{th}$ vector $a_i$ ($i \in 
\set{0,1,\ldots, \log n-1}$) is defined as the vector which has a $1$ at an index iff the $i^{th}$ bit of the index is $0$.
We can also describe it as follows:
The vector $a_i$ is made of $n/2^i$ blocks of size $2^i$ each. Alternately fill the blocks with all $1$'s and then all $0$'s.
So $a_0$ will be the vector $10$ repeating, $a_1$ will be the vector $1100$ repeating, and $a_{\log 
n -1}$ will be the vector with $1$'s in the first half and $0$'s in the second half.
We use the set $A$ to get the following lemma:
\begin{lemma}\label{lem:L0-supp1}
	There is a randomized algorithm that solves 
	\Cref{prob:special-L0sampling} with the promise that the support size is $1$ with probability $1-1/\polylog(\val) \cdot \poly(n)$ 
	and uses $O(\log n \cdot \log \log \val + \log^2 n)$ bits of space.
 \end{lemma}

Consider the following algorithm for \Cref{lem:L0-supp1}:

\begin{Algorithm}\label{alg:L0-supp1}
	A dynamic streaming algorithm for \Cref{prob:special-L0sampling} with support size $1$.
	
	\medskip
	
	\textbf{Input:} A vector $v$ specified in a dynamic stream.
	
	\medskip
	
	\textbf{Output:} The element with non-zero frequency.
	
	\medskip
		
	\textbf{During the stream:}

	\begin{enumerate}
		\item Store inner products of the frequency vector with all vectors in set $A$ using the decision 
		counters.
	\end{enumerate}
	
	\textbf{Post-Processing:}	
	\begin{enumerate}
		\item Let $s$ be the $\log n$ bit solution.
		\item If the inner product with $a_i$ is $0$ then set $s_i=1$, otherwise set $s_i=0$.
		\item Output $s$ as the non-zero index (where $i=0$ is the least significant bit).
	\end{enumerate}
\end{Algorithm}

We first show that this algorithm can be implemented in small space.
\begin{claim}\label{clm:L0-supp1-space}
	The space taken by \Cref{alg:L0-supp1} is $O(\log n \cdot \log \log \val + \log^2 n)$ bits.
\end{claim}
\begin{proof}
	We store $\log n$ decision counters, each taking $O(\log n + \log \log \alpha)$ bits.
	We do not store the vectors $a_i$ explicitly.
	When we get element $e$ in the stream, we know which $a_i$'s have a $1$ in position $e$, so we 
	can update those counters. 
\end{proof}

We now show that we can recover the correct non-zero index with high probability.
\begin{claim}\label{clm:L0-supp1-corr}
	\Cref{alg:L0-supp1} outputs the correct answer with probability $1-1/\polylog(\val) \cdot \poly(n)$.
\end{claim}
\begin{proof}
    We first condition on none of the decision counters failing.
    $a_i$ is a vector which has a $1$ at some index iff the $i^{th}$ bit of the index is $0$.
    If the inner product of the frequency vector and $a_i$ is $1$ then $a_i$ is $1$ on the non-zero index. 
    This implies that the $i^{th}$ bit of the non-zero index is $0$.
    Thus, the protocol is deterministically correct after conditioning on none of the decision counters failing.
    A union bound over the failure probabilities of the counters concludes the proof.
\end{proof}

\Cref{clm:L0-supp1-space} and \Cref{clm:L0-supp1-corr} together prove \Cref{lem:L0-supp1}.

We also need an algorithm that can tell us whether we are in the support $1$ case or not.
The idea is to partition the coordinates into $4$ parts randomly and check the sums of each part.
If there are many non-zero coordinates then the hope is that they are spread across multiple 
partitions, making the sums in multiple partitions non-zero.
We use this idea to get the following lemma:
\begin{lemma}\label{lem:L0-detect-supp1}
	There is a randomized algorithm that can detect whether or not the support size is $1$ with probability $1-1/\polylog(\val) \cdot \poly(n)$ 
	and uses $O((\log \log \val + \log n)^2)$ bits of space.
 \end{lemma}
 Consider the following algorithm for \Cref{lem:L0-detect-supp1}:

\begin{Algorithm}\label{alg:L0-detect-supp1}
	A dynamic streaming algorithm for detecting a vector with support size $1$.
	
	\medskip
	
	\textbf{Input:} A vector $v$ specified in a dynamic stream.
	
	\medskip
	
	\textbf{Output:} Whether the support of non-zero elements in $v$ is exactly $1$ or not.

	\medskip

	\textbf{Pre-processing:}
	\begin{enumerate}
		\item Let $r= O(\log n + \log \log \val)$.
		\item Create $r$ independent random partitions $V^i_1,V^i_2,V^i_3,V^i_4$ ($i \in [r]$) where 
		every element independently belongs to $V^i_j$ with probability $1/4$.
	\end{enumerate}
	
	\medskip
	
	\textbf{During the stream:}
	
	\begin{enumerate}
		\item Store inner products $z^i_j$ of the frequency vector with the characteristic vectors of $V^i_j$ 
		using the decision counters for $i \in r$ and $j \in [4]$.
	\end{enumerate}
	
	\textbf{Post-Processing:}	
	\begin{enumerate}
		\item If there exists $i$ such that more than one $z^i_j$ is non-zero or all $z^i_j$'s are zero 
		then output non-zero support is not $1$.
		Otherwise, output non-zero support is $1$.
	\end{enumerate}
\end{Algorithm}

\begin{claim}\label{clm:L0-detect-supp1-corr}
	\Cref{alg:L0-detect-supp1} can identify whether the support of the 
	non-zero elements is $1$ or not with probability $1-1/\poly(n) \cdot \polylog(\val)$. 
\end{claim}
\begin{proof}
	We first condition on the event that none of the counters fail.
	In the case the support is $1$, exactly one $z^i_j$ will be non-zero for all $i \in [r]$.
	In the case the support is $0$, all $z^i_j$'s will be zero for all $i \in [r]$.
	
	Now consider the case where there are $t>1$ non-zero elements.
	Let $e$ and $e'$ be two different non-zero elements.
	Consider the state of the inner products $\tilde{z}^i_j$ without elements $e$ and $e'$.
	\begin{itemize}
		\item Case 1: If at least $2$ counters are zero then with probability at least $1/16$, $e$ and $e'$ are 
		assigned to the counters that are zero, making them non-zero. Thus, there are at least two non-zero 
		counters.
		\item Case 2: If at least $3$ counters are non-zero then with probability at least $1/16$, $e$ and 
		$e'$ are assigned to the same non-zero counter. Thus, there are at least two non-zero counters 
		since two out of the three remain unchanged.
	\end{itemize}
	In both cases, we are successful with probability at least $1/16$.
	Thus, the probability we fail over all iteration is $(1-1/16)^r \leq \exp(-r/16) \leq 1/ \poly(n) \cdot \polylog(\val)$.
     A union bound over the failure probabilities of the counters concludes the proof.
\end{proof}

\begin{claim}\label{clm:L0-detect-supp1-space}
	\Cref{alg:L0-detect-supp1} uses $O((\log n + \log \log \val)^2)$ bits of space. 
\end{claim}
\begin{proof}
	There are $4 r= O(\log n + \log \log \val)$ decision counters each taking $O(\log n + \log \log 
	\val)$ bits of space.
\end{proof}

\Cref{clm:L0-detect-supp1-corr} and \Cref{clm:L0-detect-supp1-space} together prove \Cref{lem:L0-detect-supp1}.

We now combine the ideas above to give the final algorithm to solve \Cref{prob:special-L0sampling} (for an arbitrary support size).
The idea is to reduce to the support size $1$ case.
Let $k$ be the number of non-zero elements in the frequency vector. 
If we sample $1/k$ fraction of the coordinates then with constant probability only one of the 
non-zero elements lies in the support. 
By repeating this many times we can reduce to the support $1$ case with high probability.
We do not know $k$ but we can guess it to within a factor of $2$ by trying all powers of $2$.
Also, we can detect when the support size is exactly $1$.
Consider the final algorithm:

\begin{Algorithm}\label{alg:L0-sampling}
	A dynamic streaming algorithm for \Cref{prob:special-L0sampling}.
	
	\medskip
	
	\textbf{Input:} A vector $v$ specified in a dynamic stream.
	
	\medskip
	
	\textbf{Output:} An element with non-zero frequency.
	
	\medskip
	
	\textbf{During the stream:}

        \smallskip
 
	For $i=0$ to $\log n$:

        \smallskip
	
	Repeat $r=O(\log n + \log \log \val)$ times:	
	\begin{enumerate}
		\item Sample every element independently with probability $1/2^i$.
		\item Run \Cref{alg:L0-supp1} and \Cref{alg:L0-detect-supp1} on the sampled subset.
	\end{enumerate}
	
	\textbf{Post-Processing:}	
	\begin{enumerate}
		\item Find an iteration where \Cref{alg:L0-detect-supp1} returns non-zero support $1$. If no such 
		iteration exists output ``FAIL''.
		\item Use the corresponding copy of \Cref{alg:L0-supp1} to output the non-zero index.
	\end{enumerate}
\end{Algorithm}

The space used by \Cref{alg:L0-sampling} is $\poly(\log \log \val + \log n)$.
\begin{claim}\label{clm:L0-sampling-space}
    \Cref{alg:L0-sampling} uses $O( \log n \cdot (\log \log \val + \log n)^3)$ bits of space.
\end{claim}
\begin{proof}
    We use $r \log n$ copies of \Cref{alg:L0-supp1} and \Cref{alg:L0-detect-supp1}.
    A copy of \Cref{alg:L0-supp1} takes space $O(\log n \cdot \log \log \val + \log^2 n)$ bits of space.
    A copy of \Cref{alg:L0-detect-supp1} takes space $O((\log \log \val + \log n)^2)$ bits of space.
    Thus the total space taken is $O( \log n \cdot (\log \log \val + \log n)^3)$ bits.
\end{proof}

We now show the correctness of \Cref{alg:L0-sampling}.
We start by conditioning on the high probability events of all copies of \Cref{alg:L0-supp1} and \Cref{alg:L0-detect-supp1}.
Let $k$ be the size of the non-zero support of the input vector $v$.
Consider the iterations when $i=i^*$ where $k \leq 1/p:=2^{i^*} \leq 2k$.
We need to show that with high probability, in some iteration, exactly one element is sampled.
\begin{claim}
    There is an iteration where exactly one element from the non-zero support is sampled with probability $1-1/\polylog(\val) \cdot \poly(n)$.
\end{claim}
\begin{proof}
    The claim is trivial when $k=1$, so we consider $k\geq 2$.
    The probability that exactly one element is sampled is:
    \begin{align*}
    \prob{\text{exactly one element is sampled}} &=
    k \cdot p \cdot (1-p)^{k-1} \\
    &\geq  k \cdot \frac{1}{2k} \cdot (1-1/k)^{k} \\
    &\geq 1/16. \tag{$k \geq 2 \text{ implies } (1-1/k)^{k} \geq 1/8$}  
    \end{align*}

    We fail when this does not happen in any iteration. The probability of failure is: 
    \begin{align*}
        \prob{\text{failure}} \leq (1-1/16)^r \leq 1/\polylog(\val) \cdot \poly(n).
    \end{align*}
Thus, there is an iteration where exactly one element from the non-zero support is sampled with probability $1-1/\polylog(\val) \cdot \poly(n)$.
\end{proof}

This means that with high probability, there is an iteration where exactly one element from the non-zero support is sampled. For this iteration with high probability, the corresponding copy of \Cref{alg:L0-detect-supp1} will output non-zero support $1$, and the corresponding copy of \Cref{alg:L0-supp1} will output the non-zero index.
Also, for all other iterations where the support size is not $1$, the corresponding copy of \Cref{alg:L0-detect-supp1} will output non-zero support is not $1$ with high probability.
A union bound over the failure probabilities of all copies of \Cref{alg:L0-detect-supp1} and \Cref{alg:L0-supp1} proves the correctness of \Cref{alg:L0-sampling} in \Cref{lem:L0}.
\Cref{clm:L0-sampling-space} proves the space bound of \Cref{alg:L0-sampling} in \Cref{lem:L0}.
Therefore, we have proved \Cref{lem:L0}.

Note that the number of random bits used is more than the space bound. But we can fix this by using Nisan's pseudorandom generator (\Cref{fact:NisanPRG}). Using this increases the space by a factor of $O(\log n)$ and makes the true randomness needed fit in the space bound.


\newcommand{\reps}{r}

\subsection{Edge Connectivity}\label{sec:edgeconnsgt}

\paragraph{A Certificate of Edge Connectivity}\label{sec:Certificate}

We present a certificate of $k$-edge-connectivity in this section. 
Our algorithm and analysis are very similar to the $k$-vertex connectivity algorithm of \cite{assadi2023tight}.

\begin{Algorithm}\label{alg:cert}
	An algorithm for computing a certificate of $k$-edge-connectivity.
	
	\medskip
	
	\textbf{Input:} A graph $G = (V,E)$ and an integer $k$.
	
	\medskip
	
	\textbf{Output:} A certificate $H$ for $k$-edge-connectivity of $G$.

	\smallskip
		
	\begin{enumerate}
		\item For $i = 1, 2, \ldots, \reps := \paren{200 k \ln n}$ do the following:
		\begin{itemize}
			\item Let $E_i$ be a subset of $E$ where each edge is sampled independently with 
			probability $1/k$.
			\item Let $G_i=(V,E_i)$ be the subgraph of $G$ containing edges of $E_i$.
			\item Compute a spanning forest $T_i$ of $G_i$. 
		\end{itemize}
	\item Output $H:= T_1 \cup T_2 \cup \ldots \cup T_{\reps}$ as a certificate for $k$-edge-connectivity 
	of $G$.
\end{enumerate}
	
\end{Algorithm}

The algorithm of \cite{assadi2023tight} samples vertices with probability $1/k$ and stores spanning forests over the sampled induced subgraph. The number of iteration in their algorithm is also larger by a factor of $k$ because in each iteration they are working with a smaller graph on roughly $n/k$ vertices.

The following theorem proves the main guarantee of this algorithm. 

\begin{theorem}\label{thm:k-con-cert}
	Given any graph $G=(V,E)$ and any integer $k \geq 1$, \Cref{alg:cert} outputs a 
	certificate $H$ of $k$-edge-connectivity of $G$ with $O(k n \cdot \log n)$ edges with high 
	probability.
\end{theorem}

The analysis in the proof of~\Cref{thm:k-con-cert} is twofold.
We first show that pairs of vertices that are at least $2k$-edge-connected in $G$ stay 
$k$-edge-connected in 
$H$. 
Secondly, we show that edges whose endpoints are not $2k$-edge-connected in $G$ will be 
preserved in $H$.
Putting these together, we then show that $H$ is a certificate for $k$-edge-connectivity of $G$ 
and 
has at most $\Ot(kn)$ edges.

We start by bounding the number of edges of the certificate $H$.

\begin{lemma}\label{lem:cert-space}
	The certificate $H$ in \Cref{alg:cert} has $O(k n \cdot \log n)$ edges.
\end{lemma}
\begin{proof}
	Each spanning forest $T_i$ has at most $n-1$ edges.
	Thus, the total number of edges in $H$ is at most $200 k n \ln n$.
\end{proof}

We now prove the correctness of this algorithm in the following lemma. 
\begin{lemma}\label{lem:cert-corr}
	Subgraph $H$ of~\Cref{alg:cert} is a certificate of $k$-edge-connectivity for $G$ with high 
	probability.
\end{lemma}

\Cref{lem:cert-corr} will be proven in two steps. We first show that every pair of vertices that have at least 
$2k$ edge-disjoint paths 
between them in $G$ have at least $k$ edge-disjoint paths in $H$ with high 
probability.
\begin{lemma}\label{clm:connectivity-in-H}
	Every pair of vertices $s,t$ in $G$ that have at least $2k$ edge-disjoint paths between them in $G$
	have at least $k$ edge-disjoint paths in $H$ with high probability.
\end{lemma}

We then show that every edge whose endpoints have less than $2k$ edge-disjoint 
paths between them in $G$ will belong to $H$ as well. 

\begin{lemma}\label{clm:edges-in-H}
	Every edge $(s,t) \in G$ that has less than $2k$ edge-disjoint paths between its endpoints in 
	$G$ belongs to $H$ also with high probability.
\end{lemma}
The proofs of these lemmas appear in the next two subsections. We first use these lemmas to prove 
\Cref{lem:cert-corr} and conclude the proof of~\Cref{thm:k-con-cert}.

\begin{proof}[Proof of \Cref{lem:cert-corr}]
	We first condition on the events in \Cref{clm:connectivity-in-H} and 
	\Cref{clm:edges-in-H} both of which happen with high probability.
	We need to show that $H$ is $k$-edge-connected iff $G$ is $k$-edge-connected.
	If $H$ is $k$-edge-connected then $G$ is also $k$-edge-connected simply because $H$ is a 
	subgraph of $G$. 
	
	We now assume towards a contradiction that $G$ is $k$-edge-connected, but $H$ is not.
	This means that there is a cut $S$ of size at most $k-1$, i.e.\ removing the set of edges $X$ going out 
	of $S$ disconnects $H$.
	Since $G$ is $k$-edge-connected, the cut $S$ has at least $k$ edges going out of it so deleting the 
	set of edges $X$ cannot disconnect $G$ and thus $G$ has an edge 
	$e=(s,t)$ between $S$ and $T:=V-S$ (see~\Cref{fig:partition}). 
	
	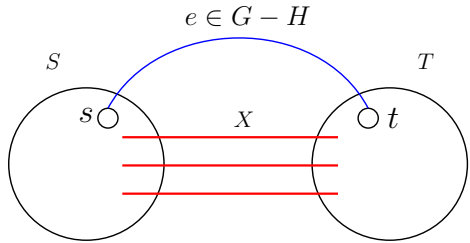
\begin{figure}[ht!]
	\centering
	{   \resizebox{180pt}{90pt}{ \tikzset{every picture/.style={line width=0.75pt}} 

\begin{tikzpicture}[x=0.75pt,y=0.75pt,yscale=-1,xscale=1]
	
	\draw   (121,179) .. controls (121,149.18) and (145.18,125) .. (175,125) .. controls (204.82,125) and 
	(229,149.18) .. (229,179) .. controls (229,208.82) and (204.82,233) .. (175,233) .. controls 
	(145.18,233) and (121,208.82) .. (121,179) -- cycle ;
	\draw   (332,179) .. controls (332,149.18) and (356.18,125) .. (386,125) .. controls (415.82,125) 
	and (440,149.18) .. (440,179) .. controls (440,208.82) and (415.82,233) .. (386,233) .. controls 
	(356.18,233) and (332,208.82) .. (332,179) -- cycle ;
	\draw   (183,146.5) .. controls (183,142.63) and (186.13,139.5) .. (190,139.5) .. controls 
	(193.87,139.5) and (197,142.63) .. (197,146.5) .. controls (197,150.37) and (193.87,153.5) .. 
	(190,153.5) .. controls (186.13,153.5) and (183,150.37) .. (183,146.5) -- cycle ;
	\draw   (364,146.5) .. controls (364,142.63) and (367.13,139.5) .. (371,139.5) .. controls 
	(374.87,139.5) and (378,142.63) .. (378,146.5) .. controls (378,150.37) and (374.87,153.5) .. 
	(371,153.5) .. controls (367.13,153.5) and (364,150.37) .. (364,146.5) -- cycle ;
	\draw [draw=blue]   (190,139.5) .. controls (225,74) and (336,72) .. (371,139.5) ;
	
	\draw [draw=red] [line width=1.2pt]  
	(200,160) --  (350,160);

	\draw [draw=red] [line width=1.2pt]  
	(200,180) --  (350,180);

	\draw [draw=red] [line width=1.2pt]  
	(200,200) --  (350,200);
		
	\draw (145,99) node [anchor=north west][inner sep=0.75pt]   [align=left] {\large{$S$}};
	\draw (404,99) node [anchor=north west][inner sep=0.75pt]   [align=left] {\large{$T$}};
	\draw (383,138) node [anchor=north west][inner sep=0.75pt]   [align=left] {\LARGE{$t$}};
	\draw (168,138) node [anchor=north west][inner sep=0.75pt]   [align=left] {\LARGE{$s$}};
	\draw (276,140) node [anchor=north west][inner sep=0.75pt]   [align=left] {\large{$X$}};
	\draw (242,65.5) node [anchor=north west][inner sep=0.75pt]   [align=left] {\Large{$e \in G-H$}};

\end{tikzpicture}}  }
	\caption{An illustration of the cut $S$ in $G$ and $H$. After $X$ is deleted, there are no edges  between 
	$S$ and $T$ in $H$, while $G$ has at least one edge $e=(s,t)$ between $S$ and $T$, to ensure its 
	$k$-edge-connectivity as $\card{X} < k$. 
		\label{fig:partition}}
\end{figure}
	\noindent
	We now consider two cases.
	\begin{itemize}
		\item Case 1: $s$ and $t$ have at least $2k$ edge-disjoint paths between them in $G$. \\
		We have conditioned on the event in \Cref{clm:connectivity-in-H}, so we can say that $s$ and $t$ 
		have at least $k$ edge-disjoint paths in $H$.
		Deleting $X$ can destroy at most $\card{X} \leq k-1$ of these paths in $H$.
		This implies that there is still an $s$-$t$ path in $H-X$ and thus there is an edge between $S$ 
		and $T$ in $H-X$, a contradiction. 
		
		\item Case 2: $s$ and $t$ have less than $2k$ edge-disjoint paths between them in $G$. \\
		Since there are fewer than $2k$ edge-disjoint paths between $s$ and $t$ in 
		$G$, by conditioning on the event of \Cref{clm:edges-in-H}, $e$ would be preserved in $H$, a 
		contradiction with $H$ having no edge between $S$ and $T$. 
	\end{itemize}	
	In conclusion, we get that $H$ is a certificate of $k$-edge-connectivity for $G$ with high probability.
\end{proof}

\Cref{thm:k-con-cert} now follows immediately from~\Cref{lem:cert-space} and~\Cref{lem:cert-corr}.

	Before moving on from this section, we present the following corollary 
	of~\Cref{alg:cert} that allows for using this algorithm
	for some other related problems in dynamic streams as well. 
	
	\begin{corollary}\label{cor:k-con-extension}
		The subgraph $H$ output by~\Cref{alg:cert} with high probability satisfies the following guarantees: 
		\begin{enumerate}[label=$(\roman*)$]
			\item For any pair of vertices $s,t$ in $G$, there are at least $k$ edge-disjoint $s$-$t$ paths in 
			$G$ iff there at least $k$ edge-disjoint $s$-$t$ paths in $H$ (this holds even if $G$ is not 
			$k$-edge-connected). 
			\item Every cut of $H$ with size less than $k$ is a cut in $G$ with size less than $k$ and vice 
			versa (this means all cuts of $G$ are preserved in $H$ and no new ones are created as long as 
			their size is less than $k$). 
		\end{enumerate}
	\end{corollary}
\noindent
The proof of this corollary is identical to that of~\Cref{lem:cert-corr} and is thus omitted.

\begin{proof}[Proof of \texorpdfstring{\Cref{clm:connectivity-in-H}}{Lemma}]
We prove \Cref{clm:connectivity-in-H} in this part following the same approach as in \cite{assadi2023tight}. 
For this proof, without loss of generality, we can assume that $k > 1$: for 
$k=1$, each graph $G_i$ is the same as $G$ and thus the spanning forest computes an $s$-$t$ path 
which will be added to $H$,  trivially implying the proof.

Fix any pair of vertices $s,t$ with at least $2k$ edge-disjoint paths between them.
We choose an arbitrary set $X$ of $k-1$ edges, and the goal is to show that $s$ and $t$ remain 
connected in the graph $H-X$ with very high probability. 
We do so by showing that out of the at least $k$ edge-disjoint paths between $s$ and $t$ in $G-X$, 
with probability $1-n^{-\Theta(k)}$, at least one of them is 
entirely sampled as part of the \emph{subset} of $G_i$'s for $i \in [r]$ that do not contain any edges 
from $X$.  
This will be sufficient to prove the existence of an $s$-$t$ path in $H-X$.
A union bound over the $\binom{n^2}{k-1}$ choices of $X$ and $\binom{n}{2}$ pairs $s,t$ 
concludes the proof.

Fix $X$ as a set of $k-1$ edges. Define:
\begin{align}
I(X):=\set{i \in [\reps]: E_i \cap X =\emptyset}; \label{eq:index-set}
\end{align}
that is, the indices of sampled graphs in $G_1,\ldots,G_r$ that contain no edges from $X$.
We first argue that $\card{I(X)}$ is large with very high probability. 

\begin{claim}\label{clm:I(X)}
	$\Pr\paren{\card{I(X)} \leq \reps/8} \leq n^{-5k}$.
\end{claim}
\begin{proof}
	Fix any index $i \in [r]$ and an edge $e \in X$. The probability that $e$ is not sampled in $E_i$ is 
	$(1-1/k)$ by definition and thus, 
	\[
		\Pr\paren{E_i \cap X = \emptyset} = (1-1/k)^{k-1} \geq 1/4,
	\]
	given that $k > 1$ (as argued earlier) and the sampling of edges being independent in $E_i$. 
	Therefore, we have, 
	\[
	\Exp\card{I(X)} = \reps \cdot (1-1/k)^{k-1} \geq \reps/4.
	\]
	By an application of the Chernoff bound (\Cref{prop:chernoff}) with $\mu_L = \reps/4$ and $\eps=1/2$, we have, 
	\[
		\prob{\card{I(X)} \leq \reps/8} \leq \exp(-\reps/4 \cdot 1/10) < n^{-5k}. \qedhere
	\]	
\end{proof}

	In the rest of the proof we condition on the event that $\card{I(X)} \geq \reps/8$. To continue, we need some definitions. 
	There are more than $k$ edge-disjoint paths between $s$ and $t$ in $G-X$ since there 
	were $2k$ of them in $G$ and only $k-1$ edges (set $X$) are deleted. 
	Choose $k$ of them arbitrarily denoted by $P_1(X), \ldots , P_k(X)$ (see \Cref{fig:aBc} for an 
	illustration).

	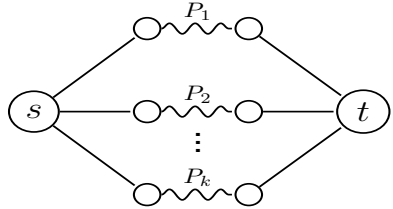
\begin{figure}[ht!]
	\centering
	{   \resizebox{150pt}{80pt}{ 
%
%
%
%

\begin{tikzpicture}[
	shorten >=1pt, auto, thick,
	node distance=1cm,
	state/.style={circle,draw,fill=white,font=\sffamily\Large\bfseries},
	labelstate/.style={circle,draw,fill=white,  draw,white, text=black, font=\sffamily\Large\bfseries}
	]
\node[state] (S) {$s$};
\node[state] (V1) [right =of S] {};
\node[state] (V3) [right =of V1] {};
\node[state] (T) [right=of V3] {$t$};

\node[state] (U1) [above =of V1] {};
\node[state] (U3) [above =of V3] {};

\node[state] (W1) [below =of V1] {};
\node[state] (W3) [below =of V3] {};

\path[every node/.style={font=\sffamily\small}]
(S) edge [bend right = 0] node {} (V1)
(V1) edge [bend  right = 0, decorate,decoration=snake] node {$P_2$} node [below] 
{\textbf{$\vdots$}} (V3)
(V3) edge [bend  right = 0] node {} (T)

(S) edge [bend right = 0] node {} (U1)
(U1) edge [bend  right = 0, decorate,decoration=snake] node {$P_1$} (U3)
(U3) edge [bend  right = 0] node {} (T)

(S) edge [bend right = 0] node {} (W1)
(W1) edge [bend  right = 0, decorate,decoration=snake] node {$P_k$} (W3)
(W3) edge [bend  right = 0] node {} (T)
;
\end{tikzpicture}}  }
	\caption{An illustration of the $k$ edge-disjoint $s$-$t$ paths $P_1(X),P_2(X),\ldots,P_k(X)$.
		Note that the paths can share vertices but are edge-disjoint.
		\label{fig:aBc}}
	\end{figure}

	We define the notion of ``preserving'' a path. 

		\begin{definition}
		Let $G_X := \cup_{i \in I(X)} G_i$ be the union of graphs indexed in~\Cref{eq:index-set}. We say that 
		a path $P$ in $G-X$ is \textnormal{\textbf{preserved}} in $G_X$ iff for every edge $e \in P$, 
		there exists at least one $i \in I(X)$ such that $e \in G_i$; in other words, the entire path $P$ belongs to $G_X$. 
	\end{definition}
		
	We are going to show that with very high probability, at least one path $P_j(X)$ for $j \in [k]$ is 
	preserved by $G_X$. Before that, we have the following claim that allows 
	us to use this property to conclude the proof. 
	\begin{claim}\label{clm:st-path}
		If any $s$-$t$ path $P_j(X)$ for $j \in [k]$ is preserved in $G_X$ then $s$ and $t$ are connected in 
		$H-X$. 
	\end{claim}
	\begin{proof}
		Given that $P := P_j(X)$ is preserved, we have that for any edge $e=(u,v) \in P$, there is some graph $G_i$ for $i \in I(X)$ that contains $e$. 
		This means that $u,v$ are connected in $G_i$ which in turn implies that the spanning forest $T_i$ of 
		$G_i$ contains a path between $u$ and $v$. 
		Moreover, since $i \in I(X)$, we know that $G_i$ and hence $T_i$ contain no edges of $X$ and thus 
		$u$ and $v$ are connected in $T_i - X$ as well. 
		Stitching together these $u$-$v$ paths for every edge $(u,v) \in P$ then gives us a walk between $s$ and $t$ in $H-X$,  implying that $s$ and $t$ are connected in $H-X$. 
	\end{proof}

	We will now prove that some path $P_j(X)$ for $j \in [k]$ is preserved with very high probability.
	\begin{claim}\label{clm:path-prob}
		Conditioned on $\card{I(X)} \geq r/8$, we have the following:
  \[
		\prob{P_j 
		\not\subseteq G_X \text{ for at least $k/2$ values of } j \in [k]} \leq n^{-11k}.
\]
 \end{claim}
	\begin{proof}	 
		To start the proof, note that even conditioned on a choice of $I(X)$, the edges in each path $P_j(X)$ 
		appear independently in each graph $G_i$ for $i \in I(X)$. This is because these paths do not 
		intersect with $X$ and by the independence in sampling of each graph $G_i$ for $i \in [r]$.  
		Moreover, given that these paths are edge-disjoint, the choices of their edges across each graph 
		$G_i$ for $i \in [r]$, are independent. We crucially use these properties in this proof. 
	
		Each edge in $P_j$ is not present in $G_i$ with probability 
		$(1-1/k)$ and hence is not present 
		in $G_X$ with probability $(1-1/k)^{\card{I(X)}}$. 
		Hence, by the union bound, 
		\[
		\prob{P_j \not\subseteq  G_X} \leq \card{P_j} \cdot \left(1-1/k \right)^{\card{I(X)}} \leq n \cdot 
		\left(1-1/k\right)^{r/8} \leq n \cdot \exp\paren{-200k\ln{n}/8k} = n^{-24}.
		\]
		Finally, note that since the paths $P_j$ for $j \in [k]$ are edge-disjoint, the probability of the above 
		event is independent for each one. Thus, 
		\[
			\prob{P_j \not\subseteq G_X \text{ for at least $k/2$ values of } j \in [k]} \leq 
			\binom{k}{k/2} \paren{n^{-24}}^{k/2} \leq2^{k} \cdot  n^{-12k} \leq n^{-11k}.
		\]
		As such, the entire path $P_j$ will lie inside $G_X$ for at least $k/2$ 
		values of $j \in [k]$ with very high probability.
	\end{proof}
	
	By union bound over the events of~\Cref{clm:I(X),clm:path-prob}, we have that there exists an index $j \in [k]$ such that the path $P_j(X)$ is preserved in $G_X$. 
	Thus, by~\Cref{clm:st-path}, for a fixed choice of $X$, and $s,t$, the probability that $s$ and $t$ are 
	not connected in $H-X$
	is at most $n^{-11k}$. A union bound over the $n^{2k}$ choices of $X$ and $n^2$ choices of $s,t$, 
	then implies that the probability that even one such choice of $X$ and $s,t$ exists is at most 
	$n^{-7k}$.
	This completes the 
	proof of~\Cref{clm:connectivity-in-H}. 
\end{proof}

\begin{proof}[Proof of \texorpdfstring{\Cref{clm:edges-in-H}}{Lemma}]
	We now prove \Cref{clm:edges-in-H}.
	For this proof also, without loss of generality, we can assume that $k > 1$: for 
	$k=1$, each graph $G_i$ is the same as $G$ and thus the spanning forest uses the only $s$-$t$ path, 
	namely, the edge $(s,t)$ (as $s$ and $t$ can only be $1$-connected through the edge $(s,t)$)
	which will be added to $H$, thus trivially implying the proof.  We now consider the main case. 
	
	Fix any pair of vertices $s,t \in G$ which have less than $2k$ edge-disjoint paths between them. 
	We know that deleting some set $X$ of at most $2k-1$ edges should disconnect $s$ and $t$.
	For any $i \in [r]$, we call the graph $G_i$ \textbf{good} if it samples the edge $(s,t)$ and does not 
	sample any edge from $X$. See~\Cref{fig:partition_LowCon} for an illustration.

	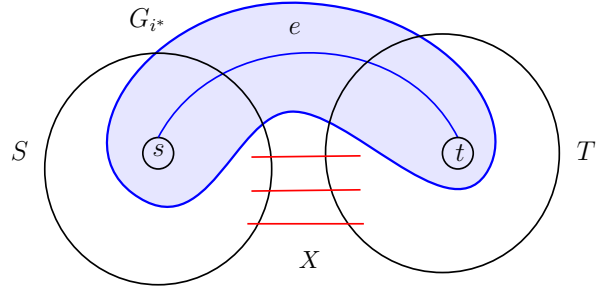
\begin{figure}[ht!]
	\centering
	{   \resizebox{230pt}{110pt}{ \tikzset{every picture/.style={line width=0.75pt}} 

\begin{tikzpicture}[x=0.75pt,y=0.75pt,yscale=-1,xscale=1]
	
	\draw   (180.75,148.75) .. controls (180.75,143.64) and (184.89,139.5) .. (190,139.5) .. controls 
	(195.11,139.5) and (199.25,143.64) .. (199.25,148.75) .. controls (199.25,153.86) and (195.11,158) .. 
	(190,158) .. controls (184.89,158) and (180.75,153.86) .. (180.75,148.75) -- cycle ;
	\draw  [draw=blue]   (190,139.5) .. controls (225,74) and (335,72) .. (370,139.5) ;

	\begin{scope}[on background layer]
	\draw  [draw=blue, fill=blue!10, line width=1pt] (276.55,60.27) .. controls (363.55,64.27) and (410,128.5) .. 
	(386.55,161.27) .. controls 
	(363.1,194.05) and (303.55,122.27) .. (269.55,124.27) .. controls (235.55,126.27) and (219.55,205.27) .. 
	(174.55,172.27) .. controls (129.55,139.27) and (189.55,56.27) .. (276.55,60.27) -- cycle ;
	\end{scope}

	\draw   (360.75,148.75) .. controls (360.75,143.64) and (364.89,139.5) .. (370,139.5) .. controls 
	(375.11,139.5) and (379.25,143.64) .. (379.25,148.75) .. controls (379.25,153.86) and (375.11,158) .. 
	(370,158) .. controls (364.89,158) and (360.75,153.86) .. (360.75,148.75) -- cycle ;
	\draw   (121.83,158) .. controls (121.83,120.35) and (152.35,89.83) .. (190,89.83) .. controls (227.65,89.83) 
	and (258.17,120.35) .. (258.17,158) .. controls (258.17,195.65) and (227.65,226.17) .. (190,226.17) .. 
	controls (152.35,226.17) and (121.83,195.65) .. (121.83,158) -- cycle ;
	\draw   (290.53,148.75) .. controls (290.53,109.97) and (321.97,78.53) .. (360.75,78.53) .. controls 
	(399.53,78.53) and (430.97,109.97) .. (430.97,148.75) .. controls (430.97,187.53) and (399.53,218.97) .. 
	(360.75,218.97) .. controls (321.97,218.97) and (290.53,187.53) .. (290.53,148.75) -- cycle ;
	\draw [draw=red]   (246,151) -- (311.55,150.27) ;
	\draw [draw=red]    (246,171) -- (311.55,170.27) ;
	\draw  [draw=red]   (243.55,190.27) -- (313.55,190.27) ;
	
	\draw (171,62) node [anchor=north west][inner sep=0.75pt]   [align=left] {\large{$G_{i^*}$}};
	\draw (367,142) node [anchor=north west][inner sep=0.75pt]   [align=left] {\large{$t$}};
	\draw (185,142) node [anchor=north west][inner sep=0.75pt]   [align=left] {\large{$s$}};
	\draw (273,205) node [anchor=north west][inner sep=0.75pt]   [align=left] {$X$};
	\draw (267,70) node [anchor=north west][inner sep=0.75pt]   [align=left] {\large{$e$}};
	
	\draw (100,140) node [anchor=north west][inner sep=0.75pt]   [align=left] {\large{$S$}};
	
	\draw (440,140) node [anchor=north west][inner sep=0.75pt]   [align=left] {\large{$T$}};

\end{tikzpicture}}  }
	\caption{An illustration of a good graph $G_{i^*}$ wherein the edge $e=(s,t)$ is sampled 
	and all the edges in set $X$ are not. Thus, none of  the $s$-$t$ paths, except for the edge $e$, 
	exist in $G_{i^*}$. Therefore, the spanning forest $T_{i^*}$ 
	necessarily contains the edge $e=(s,t)$. 
		\label{fig:partition_LowCon}}
\end{figure}

	We have, 
	\[
		\prob{G_i \text{ is good}} \geq 1/k \cdot (1-1/k)^{2k-1} \geq 1/8k \tag{as $k\geq 2$ so 
		$(1-1/k)^{2k-1} 
		\geq (1/2)^{3}$}.
	\]
	Given the independence of choices of $G_i$ for $i \in [r]$, we have, 
	\[
		\prob{\text{No $G_i$ is good}} \leq (1-1/8k)^{r} \leq \exp\paren{-200 k \ln n/8k} = n^{-25}.
	\]
	
	Therefore, there is a graph $G_{i^*}$ for $i^* \in [r]$ where $(s,t)$ is sampled but $X$ is 
	not (see \Cref{fig:partition_LowCon}).
	This means that the spanning forest $T_{i^*}$ has to contain the edge $(s,t)$ as 
	there is no other 
	path between $s$ and $t$ (we have effectively ``deleted'' $X$ by not sampling it).
	Thus, the edge $(s,t)$ belongs to $H$ with probability at least $1-n^{-25}$. 
	A union bound over all possible pairs $s,t \in G$ concludes the proof. 
\end{proof}


\paragraph{The SGT Streaming Algorithm}\label{sec:Algorithm}

We present our single pass SGT streaming algorithm for $k$-edge-connectivity in this section.
The parameter for SGT streaming is $\val$ i.e.\ the frequencies are bounded between $-\val$ and $\val$.
We note that the $\Ot(\cdot)$ notation will hide $\log n$ and $\log \log \val$ factors.
The algorithm outputs a certificate of $k$-edge-connectivity for the input graph at the end of the stream. 
Thus, by the definition of the certificate, to know whether or not the input graph is 
$k$-edge-connected, it suffices to test if the certificate is $k$-edge-connected, which
can be done at the end of the stream using any offline algorithm. 
Consider the following theorem:
\begin{theorem}\label{thm:one-pass-con-alg}
There is a randomized SGT streaming algorithm that given an integer $k\geq 1$ 
before the stream and a graph $G=(V,E)$ in the stream, outputs a certificate $H$ of 
$k$-edge-connectivity of $G$ with high probability using $\Ot(kn)$ bits of space.
\end{theorem}

This algorithm is just an implementation of \Cref{alg:cert} in SGT streams.
We fix the edge sets $E_i$ in \Cref{alg:cert} before the stream so the only thing we need to 
specify is how we compute the spanning forests during the stream.
We compute a spanning forest $T_i$ of $G_i$ in the stream using the dynamic streaming 
algorithm in \Cref{fact:spanning-forest} with strong $\ell_0$-samplers (\Cref{subsec:Strong-L0}).
After the stream, we output the certificate $H$.
This completes the description of the streaming algorithm.


We start by bounding the space of this algorithm.

\begin{lemma}\label{lem:one-pass-space}
	This algorithm uses $\Ot(kn)$ bits of space.
\end{lemma}
\begin{proof}
	During the stream, we run a spanning forest algorithm for each graph $G_i$ for $i \in [r]$. 
	The algorithm of~\Cref{fact:spanning-forest} with strong $\ell_0$-samplers takes at most $\Ot(n)$ bits of space.
	So the total space taken is $\Ot(\reps \cdot n)=\Ot(kn)$ bits.
\end{proof}
	
We are now ready to prove \Cref{thm:one-pass-con-alg}.
\begin{proof}[Proof of \Cref{thm:one-pass-con-alg}]
    By~\Cref{lem:one-pass-space}, this algorithm uses $\Ot(kn)$ bits of space. 
    The algorithm is also correct with high probability by \Cref{thm:k-con-cert}. 
    Moreover, none of the spanning forest algorithms of \Cref{fact:spanning-forest} fail with 
    high probability (by a union bound over the failure probabilities of 
    the $r$ spanning forest algorithms).
    
    The strong $\ell_0$-samplers work in the SGT model because they can handle large values and they do not distinguish between different non-zero values (even negative values). Also, adding the samplers for the neighborhood of vertices in a component gives a sampler for the edges going out of the component because the internal edges cancel out (similar to traditional $\ell_0$-samplers). Finally, we do not subtract edges from samplers to delete edges because to delete an edge, we need to know the exact frequency of the edge (this was the problem with the original dynamic streaming algorithm for edge connectivity \cite{AhnGM12}).
    Thus, the streaming algorithm uses $\Ot(kn)$ bits of space and, by union bound, with 
    high probability outputs a certificate 
    of $k$-edge-connectivity. 
\end{proof}

Note that the number of random bits used is more than the space bound. But we can fix this by using Nisan's pseudorandom generator (\Cref{fact:NisanPRG}). Using this increases the space by a factor of $O(\log n)$ and makes the true randomness needed fit in the space bound.
 
We also note that using the strong $\ell_0$-samplers we get a $\Ot(n)$ space algorithm for connectivity and a $\Ot(kn)$ space algorithm for $k$-vertex-connectivity in the SGT streaming model.

\section{Annotated Streaming Schemes for Dynamic Graph Streams}\label{sec:Annotated-dynamic}

\subsection{The Layering Lemma}
In this subsection, we prove some simple properties of $k$-connected graphs and prove important lemmas called layering lemmas for vertex and edge connectivity.

\subsubsection{Properties of k-connected graphs}

\begin{claim}\label{clm:adding-vertex}
	Let $G$ be $k$-vertex-connected. Consider a new vertex $v$ and attach it to $\geq k$ arbitrary 
	vertices of $G$ and call this new graph $G'$. Then $G'$ is $k$-vertex-connected.
\end{claim}
\begin{proof}
    Consider a vertex cut $X$ of size at most $k-1$. $X$ cannot disconnect any pair of vertices in $G$ since $G$ is $k$-vertex-connected. If $v$ is isolated on one side of the cut then it must have a neighbor on the other side since $X$ has at most $k-1$ vertices and $v$ has $\geq k$ neighbors.
\end{proof}

\begin{claim}\label{clm:k-paths-to-set}
	Let $G=(V,E)$ be a graph and let $T \subseteq V$ be a set of vertices. We know that there 
	are at 
	least $k$ vertex-disjoint (edge-disjoint) paths from every vertex of $T$ to a special vertex $t \not\in 
	T$. 
	If a vertex $v \not\in T \cup \set{t}$ has $k$ vertex-disjoint (edge-disjoint) paths to $T$ 
	then there are $k$ 
	vertex-disjoint (edge-disjoint) paths between $v$ and $t$.
\end{claim}
\begin{proof}
	Consider any vertex cut $X$ of size at most $k-1$.
	We will show that $v$ and $t$ are connected after deleting $X$.
	$X$ has at most $k-1$ vertices so there is a path $P$ from $v$ to $T$ that has no vertices 
	of $X$.
	Let $t' \in T$ be the vertex path $P$ ends on.
	We know that there are $k$ vertex-disjoint paths between $t'$ and $t$ (by assumption) thus after deleting $X$ there is at least one path that has no vertices of $X$.
	This means that after deleting $X$ there is a path from $v$ to $t$ via $t'$ implying that there 
	are $k$ vertex-disjoint paths between $v$ and $t$.
\end{proof}

\begin{claim}\label{clm:k-vconn-proof-simple}
    Let $G=(V,E)$ be a graph and let $X$ be a fixed minimum vertex cut of $G$.
    If there is a vertex $t\in V-X$, such that there are $k$ vertex-disjoint paths from $t$ to $v$ for all $v \in V-\set{t}$ then $G$ is 
	$k$-vertex-connected (in other words $\card{X}\geq k$).
\end{claim}
\begin{proof}
    Assume towards a contradiction that the vertex cut $X$ has size at most $k-1$. 	
    Let $T$ be a connected component containing $t$ after $X$ is deleted and let $S:=V-X-T$. 
    Consider an arbitrary vertex $s$ in $S$.
    We know that there are $k$ vertex disjoint paths between $s$ and $t$.
    So deleting $k-1$ vertices cannot disconnect $v$ and $t_i$, giving us a contradiction.
\end{proof}

\begin{corollary}\label{clm:paths-to-vertex-proof}
	Let $G=(V,E)$ be a graph and let $t_1,t_2,\ldots,t_k \in V$ be fixed vertices. If for all $i$, 
	$G$ has $k$ vertex-disjoint paths from $t_i$ to $v$ for all $v \in V-\set{t_i}$ then $G$ is 
	$k$-vertex-connected.   
\end{corollary}
\begin{proof}
    Assume towards a contradiction that there is a vertex cut $X$ of size at most $k-1$ that when deleted disconnects the graph.  
    There is a $t_i$ that is not in $X$ since there are $k$ special vertices and $\card{X}\leq k-1$.
    Applying \Cref{clm:k-vconn-proof-simple} with $X$ and $t_i$ implies that $\card{X}\geq k$, giving us a contradiction.
    Thus, $G$ is $k$-vertex-connected.
\end{proof}


\subsubsection{The Layering Lemma}

In this subsection, we prove our layering lemmas for vertex and edge connectivity that will be useful in our algorithms.
We will show complete proofs for vertex-connectivity and then only mention the differences for edge-connectivity because the proofs are very similar.
Consider the following lemmas that show short proofs for k disjoint paths from a fixed vertex to all other vertices.
\begin{lemma}\label{lem:layering}
	Let $G=(V,E)$ be a $k$-vertex-connected graph, and let $t \in V$ be an arbitrary 
	vertex. There is a $\Ot(kn)$ size proof which shows that there are $k$ 
	vertex-disjoint paths from $t$ to $v$ for all $v \in V-\set{t}$.
\end{lemma}
\begin{lemma}\label{lem:layering:econn}
	Let $G=(V,E)$ be a $k$-edge-connected graph, and let $t \in V$ be an arbitrary vertex. There is a $\Ot(k^2n)$ size proof which shows that there are $k$ edge-disjoint paths from $t$ to $v$ for all $v \in V-\set{t}$.
\end{lemma}

We will show that a short proof exists using the probabilistic method.
We first set up the structure of the proof.
Note that we will write disjoint paths without specifying edge-disjoint or vertex-disjoint to mean either of those because the layering structure for both is the same.
We want a proof that shows that the special vertex $t$ has $k$ disjoint 
paths to all vertices in $V-\set{t}$. 
The idea is to first build a set of vertices $T_0$ and show that every vertex in $T_0$ 
has $k$ disjoint paths to $t$. 
The next step is to inductively build sets $T_1,T_2,\ldots T_{\ell}$ and show that every 
vertex in 
$T_i$ has $k$ disjoint paths to the set $T_{i-1}$. Using \Cref{clm:k-paths-to-set}, this shows that every vertex in $T_i$ has $k$ disjoint paths to $t$ (since every vertex in $T_{i-1}$ has $k$ disjoint paths to $t$).
We keep doing this till we cover all the vertices using the sets $T_i$ and thus, everyone has $k$ disjoint paths to $t$.
This proves the correctness of \Cref{lem:layering} and \Cref{lem:layering:econn}.

We now show how these sets are constructed and bound the proof size.
Let $\ell := \log (n/k)$.
For $i \in [0,\ell]$, let $T_i^L, T_i^R$ be the sets where every vertex is independently sampled with probability $p_i:= \frac{k}{n} \cdot 2^i$. $T_i=T_i^L \cup T_i^R$.
Note that we sample all vertices in $T_\ell$, so we cover all the vertices.
Also, note that it is okay if a vertex is sampled in multiple $T_i$'s.
The expected size of $T_i$ is $2^{i+1} k$ for all $i \in [0,\ell]$.
We will now show that for any vertex, the proof for $k$ disjoint paths to $T_i$ is 
small in expectation.
\begin{claim}\label{clm:paths-to-Ti}
	For any vertex $v$ the proof for $k$ disjoint paths to $T_i$ takes size at 
	most $\frac{n}{2^i}$ words in expectation.
\end{claim}
\begin{proof}
	Consider the vertex $v$ and its $k$ disjoint paths $P_1, P_2,\ldots P_k$ to 
	$T_i^L$ (this exists because the graph is $k$-connected).
	For any path $P_j$, we truncate it at the first occurrence of a vertex from $T_i^R$ 
	(since we want paths to $T_i= T_i^L \cup T_i^R$).
	We now use the randomness of $T_i^R$.
	Every vertex on $P_j$ is independently sampled in $T_i^R$ with probability 
	$p_i=\frac{k}{n} \cdot 2^i$.
	Thus, in expectation, the truncated length of $P_j$ is at most $\frac{1}{p_i}=\frac{n}{k \cdot 2^i}$.
	Therefore, in expectation, the $k$ disjoint paths to $T_i$ together have size at most
	$\frac{n}{2^i}$ (by linearity of expectation). 
\end{proof}

This means that in expectation, the total proof size for $k$-vertex-connectivity is $O(k n \log (n/k))$ words.
\begin{claim}\label{clm:layer-pf-size-vconn}
	The expected total proof size for $k$-vertex-connectivity is $O(k n \log(n/k))$ words.
\end{claim}
\begin{proof}
    The proof of $k$ vertex-disjoint paths for a fixed vertex in $T_0$ takes size at most $O(n)$ words because the summed length of all paths is $O(n)$ (any vertex belongs to at most one vertex-disjoint path).
    This gives a total expected size of $O(kn)$ words for all vertices in $T_0$.
    The proof for vertices in $T_{i+1}$ takes size at most $\frac{n}{2^i}$ (by 
    \Cref{clm:paths-to-Ti}) which gives a total expected size of $2^{i+2}k \cdot 
    \frac{n}{2^i} = 4 k n$ words for all vertices in $T_{i+1}$.
    We have $\ell = \log (n/k)$ sets $T_i$ implying the claim.
\end{proof}
\begin{claim}\label{clm:layer-pf-size-econn}
	The expected total proof size for $k$-edge-connectivity is $O(k^2 n \log(n/k))$ words.
\end{claim}
\begin{proof}
    The only difference from the proof of \Cref{clm:layer-pf-size-vconn} is that the proof for vertices in $T_0$ takes size at most $O(kn)$ since each path can be of size at most $n$. This happens because edge-disjoint paths can use the same vertices. 
    This gives us an overall proof size of $O(k^2 n \log(n/k))$ words.
\end{proof}
If we can show that vertices $T_0$ have $k$ edge-disjoint paths to $t$ in smaller space then we can improve the overall proof size (since this is the main bottleneck).
For instance, if the graph has a subgraph on $O(k)$ vertices that is $k$-edge-connected then that set could be $T_0$ and the overall proof size would be $O(k n \log(n/k))$ (this would increase the verification space by an additive $O(k^2)$ words).

\begin{proof}[Proof of \Cref{lem:layering,lem:layering:econn}]
We have proved the correctness and the size bounds in expectation.
There must be some random string that achieves a proof size of at most the expected size (by definition) implying 
that there exists a 
proof of size $O(k n \log(n/k))$ words for $k$-vertex-connectivity and a proof of size $O(k^2 n \log(n/k))$ words for $k$-edge-connectivity (which an all powerful prover can find).
Thus, the prover sends this proof to the verifier.
\end{proof}

We now show that the prover can send some auxiliary information so that the verifier can verify the proof.
\begin{claim}\label{clm:layering-verification-vconn}
    The layering proof of \Cref{lem:layering} for $k$-vertex-connectivity can be verified using a $[kn + h,v]$-scheme for any $h,v$ such that $h \cdot v=n^2$.
\end{claim}
\begin{proof}
The prover sends $k$ vertex-disjoint paths for each vertex using the layering idea mentioned above (\Cref{lem:layering}).
There are three things that need to be verified.
First, all the edges that the prover sends should be a part of the input graph.
Second, the edges sent by the prover should form paths. 
Finally, the paths for a vertex should be vertex disjoint.

The prover sends a list of all edges used in the proof along with their multiplicities in $\Ot(kn)$ space. It is easy to do a subset check for the edges (ignoring multiplicities) using an $[h,v]$-scheme since there are at most $n^2$ edges (\Cref{fact:subsetdisjchk}).
The problem now is to check if the prover was truthful about the multiplicities, which can be done using the $\ell_0$-sampler trick using $\Ot(1)$ space (\Cref{fact:equaldet}). The verifier can maintain an $\ell_0$-sampler where he inserts all edges along with their multiplicities when they are provided upfront. He then deletes all edges used in the proof of $k$ vertex-disjoint paths for each vertex. Finally, he checks whether the sampler is empty. If it is not then the two sets of edges are not the same and thus, the verifier catches the lying prover. The probability of failure is at most $1-1/\poly(n)$.

The prover will send the edges of the paths one path at a time, so it is easy to verify that the edges form a path.
Finally, we need to verify that the paths for every vertex are vertex disjoint.
For this, the prover sends the vertices he is going to use in the proof upfront in increasing order. The verifier can verify this in $\Ot(1)$ space by just checking the increasing order.
The problem now is to check if the prover lied so we can run an equality check on the stream of vertices sent upfront and the vertices used in the proof. This can be done using the $\ell_0$-sampler trick in $\Ot(1)$ verification space (\Cref{fact:equaldet}).
We have to do this for all vertices so the worry is that the proof size or verification space might blow up.
Since the auxiliary information is just repeating the vertices used in the proof in sorted order, adding this auxiliary information can at most double the space used.
Also, the space on the verifier side is $\Ot(1)$ words for each vertex, but this space can be reused implying $\Ot(1)$ words of verification space for the entire disjointness check.

Therefore, the proof along with the auxiliary information takes size $\Ot(kn) + \Ot(h)$ bits and can be verified in $\Ot(v) + \Ot(1)$ bits of space, proving the $[kn+h,v]$-scheme.
\end{proof}

\begin{claim}\label{clm:layering-verification-econn}
    The layering proof of \Cref{lem:layering:econn} for edge-connectivity can be verified using a $[k^2n + h,v]$-scheme for any $h,v$ such that $h \cdot v=n^2$.
\end{claim}
The proof is almost identical with the only difference being that the proof for $k$-edge-connectivity is of size $\Ot(k^2n)$ opposed to $\Ot(kn)$ for $k$-vertex-connectivity.

\begin{corollary}\label{corr:layering-multiple}
    Let $G=(V,E)$ be a $k$-vertex-connected graph and let $t_1,t_2,\ldots,t_r \in V$ be arbitrary distinct vertices. There is a $[r \cdot (kn +h),v]$-scheme for all $h,v$ such that $h \cdot v=n^2$, which shows that for all $i \in [r]$ there are $k$ vertex-disjoint paths from $t_i$ to $v$ for all $v \in V-\set{t_i}$.
\end{corollary}
\begin{proof}
    If the prover repeats \Cref{clm:layering-verification-vconn} $r$ times and the verifier reuses his space then the proof size blows up by a factor of $r$ but we only get an additional $r$ words in the verification space. 
    Once the proof for a vertex has been verified, the memory can be emptied and reused because the verifier only has to remember that the proof has been verified for that particular vertex.
    So the verifier just remembers the vertices $t_1,t_2,\ldots,t_r$ and checks if they are distinct. But the distinctness check can be done more efficiently.
    
    The verifier keeps a counter which he increments after he has verified the proof for a vertex. After the entire proof is sent the verifier checks if the counter reached $r$. The only thing left to do is check the distinctness. In parallel, the verifier runs an $[h,v]$-scheme ($h \cdot v \geq n$) for duplicate check (\Cref{fact:dupdet}).
    Thus the additional space usage is only $\Ot(1) + \Ot(v)$ bits.
\end{proof}

We now show that \Cref{lem:layering} and \Cref{clm:layering-verification-vconn} work even when $t$ is a set of vertices (of size at least $k$) instead of a single vertex.
\begin{corollary}\label{corr:layering-sets}
	Let $G=(V,E)$ be a $k$-vertex-connected graph and let $T \subseteq V$ be an arbitrary set of 
	vertices with $\card{T} \geq k$. 
 There is a $[kn+h,v]$-scheme for any $h,v$ such that $h \cdot v=n^2$, which shows that there are $k$ 
	vertex-disjoint paths from the set $T$ to $v$ for all $v \in V-T$. 
\end{corollary}
\begin{proof}
	Create a new graph $G'$ by introducing a new vertex $t$ and connecting it to all vertices in $T$.
	We know by \Cref{clm:adding-vertex} that $G'$ is $k$-vertex-connected. 
	We can now apply \Cref{lem:layering} and \Cref{clm:layering-verification-vconn} for vertex $t$ and this gives us a proof of $k$ vertex-disjoint 
	paths from $t$ to all vertices in $V-\set{t}$.
	
	Let $v$ be a vertex in $V-T$. $v$ has $k$ vertex-disjoint paths to $t$. The last edges of all these 
	paths are from $N(t)=T$ to $t$ implying that $v$ has $k$ vertex-disjoint paths to the set $T$.
\end{proof} 

We get a similar corollary for $k$-edge-connectivity.
\begin{corollary}\label{corr:layering-sets-econn}
	Let $G=(V,E)$ be a $k$-edge-connected graph and let $T \subseteq V$ be an arbitrary set of 
	vertices with $\card{T} \geq k$. 
 There is a $[k^2n+h,v]$-scheme for any $h,v$ such that $h \cdot v=n^2$, which shows that there are $k$ 
	edge-disjoint paths from the set $T$ to $v$ for all $v \in V-T$. 
\end{corollary}


\subsection{Vertex Connectivity Schemes}
In this section, we discuss schemes for $k$-vertex-connectivity.

\verconndyn*

We start by showing a very simple $[n^2 k,n]$-scheme.

\subsubsection{Simple Algorithm}
Assume that the graph $G=(V,E)$ is $k$-vertex-connected.
The simple algorithm is to fix $k$ special vertices $t_1,t_2,\ldots t_k$ and show that all other vertices 
have $k$ vertex-disjoint paths to each $t_i$ for $i \in [k]$.
This is enough to show that the graph is $k$-vertex-connected (by 
\Cref{clm:paths-to-vertex-proof})

We know that showing $k$ vertex-disjoint paths between two 
vertices takes $O(n)$ proof size (any vertex can belong to at most one vertex-disjoint path).
We can check if the edges used in the proof belong to the graph using an $[n,n]$-scheme (\Cref{fact:subsetdisjchk}).
Thus showing that all vertices have $k$ vertex-disjoint paths to each $t_i$ has a $[k n^2,n]$-scheme.

If the graph is not $k$-vertex-connected then the prover can send the vertex cut $X$ 
along with a connected component $S$ that is formed after deleting $X$ in $\Ot(n)$ space.
The verifier then has to check whether there are any edges between $S$ and $T:=V-X-S$. So the only allowed edges are within $S$, within $T$ or incident on $X$. The verifier constructs a graph $G'$ in a streaming fashion with all such edges and then checks if the edges in input graph are a subset of $G'$.
This can be done using an $[h,v]$-scheme for subset check for any $h,v$ such that $h \cdot v=n^2$ (\Cref{fact:subsetdisjchk}).
We use a $[kn,n/k]$-scheme to conclude the simple scheme.


We now improve this algorithm using the layering lemma (\Cref{lem:layering}) and give a $[k \cdot (h+kn), v]$-scheme.

\subsubsection{Layering Algorithm}
If the graph $G=(V,E)$ is not $k$-vertex-connected, then we can use the proof of the simple algorithm which has an $[n+ h,v]$-scheme.
In the other case, the issue with the simple algorithm was that the proofs for all the vertices were of 
size $O(n)$.
The goal would be to reduce the proof sizes for some vertices.

\begin{proof}[Proof of \Cref{thm:ver-conn-dyn}]
Fix $k$ arbitrary vertices $t_1,t_2,\ldots,t_k$.
Use \Cref{corr:layering-multiple} to get a proof of $k$ vertex-disjoint paths from $t_i$ to every vertex in $V-\set{t_i}$ for all $i \in [k]$.
Thus, the proof size for all $t_i$ together is $\Ot(k\cdot (h+kn))$ 
words and can be verified using $\Ot(v)$ space for any $h,v$ such that $h\cdot v=n^2$.
This is a valid proof by \Cref{clm:paths-to-vertex-proof}  which shows that $G$ is $k$-vertex-connected. 
Thus, we get a $[k \cdot (h+kn),v]$-scheme for $k$-vertex-connectivity for any $h,v$ such that $h\cdot v=n^2$, proving \Cref{thm:ver-conn-dyn}.
\end{proof}

Notice that the product of the proof size and the verification space is $kn^2$ (when $h\geq kn$) which seems suboptimal.
There are two relaxations that can drop this extra factor of $k$.
We show them in the next two subsections.

\subsubsection{Distinguishing between \texorpdfstring{$k$}{k} and 
\texorpdfstring{$2k$}{2k} vertex-connectivity}

In this subsection, we give a $[k n + h,v]$-scheme for any $h,v$ such that $h\cdot v=n^2$, for distinguishing between vertex-connectivity 
$<k$ and $\geq 2k$. 
This is the promised-gap $k$-vertex-connectivity problem with $\eps=1$ \cite{assadi2023tight,guha2015vertex}.
\begin{lemma}\label{lem:layering-2k-vconn}
    Under dynamic graph streams on $n$-node graphs, there exists an $[kn+h,v]$-scheme for distinguishing between vertex-connectivity $<k$ and $\geq 2k$, for any $h,v$ such that $h\cdot v=n^2$.
\end{lemma}
\begin{proof}
If the graph $G=(V,E)$ is not $k$-vertex-connected we can use the $[n+h,v]$-scheme of the simple algorithm where $h\cdot v=n^2$.

If the graph $G=(V,E)$ is $2k$-vertex-connected then we fix two disjoint sets $T_L$ and $T_R$ 
containing $2k$ arbitrary vertices each.
We use \Cref{corr:layering-sets} to show $2k$ vertex-disjoint paths from every vertex in $V-T_L$ to the set 
$T_L$.
This can be done using a $[kn+h,v]$-scheme for any $h,v$ such that $h \cdot v=n^2$.
We do the same for $T_R$.
This is enough to show that $G$ is at least $k$-vertex-connected (\Cref{clm:layering-2k-vconn}) implying that $G$ is $2k$-vertex-connected (promise).
\end{proof}

\begin{claim}\label{clm:layering-2k-vconn}
	The proof shows that $G$ is at least $k$-vertex-connected.
\end{claim}
\begin{proof}
	Assume towards a contradiction that there is a vertex cut $X$ of size $<k$.
	We will show that all the vertices are connected after deleting $X$ implying that $G$ is at least 
	$k$-vertex-connected.
	
	Consider two arbitrary vertices $u,v \in T_L$.
	The proof for $T_R$ shows that each of $u,v$ have $2k$ vertex-disjoint paths to $T_R$.
	$X$ has at most $k-1$ vertices so each of $u,v$ reach at least $k+1$ vertices in $T_R$ after deleting 
	$X$.
	By the pigeonhole principle there is a vertex $t \in T_R$ which both $u$ and $v$ can reach implying 
	that there is a path between $u$ and $v$ after deleting $X$.
	Thus, any pair of vertices in $T_L$ are connected after $X$ is deleted.
	
	Consider any arbitrary vertex $s \in V-T_L$. 
	$s$ has $2k$ vertex-disjoint paths to $T_L$ so after deleting $X$ there is a vertex $u\in T_L$ that it 
	can reach.
	Thus, that after $X$ is deleted, all vertices can reach $T_L$ and $T_L$ is connected implying that $G-X$ is connected giving a 
	contradiction. 
\end{proof}


\subsubsection{A scheme in the AM model}

In this subsection, we give a $[k n+h,v]$-scheme for $k$-vertex-connectivity using public randomness and achieve a smooth tradeoff too (similar to the promised-gap $k$-vertex-connectivity problem).
This is the AM model in stream verification.
\begin{lemma}
    Under dynamic graph streams on $n$-node graphs, there exists an $[kn+h,v]$-AM-scheme for $k$-vertex-connectivity for any $h,v$ such that $h\cdot v=n^2$.
\end{lemma}
\begin{proof}
If the graph $G=(V,E)$ is not $k$-vertex-connected we can use the $[n+h,v]$-scheme of the simple algorithm where $h\cdot v=n^2$.

In the other case, we fix $X$ as the lexicographically first minimum vertex cut (we do not know $X$ but it 
is fixed).
The prover now samples $\ell = 2 \log n$ vertices $t_1,t_2,\ldots,t_\ell$ uniformly at random using the 
public randomness.
Use \Cref{corr:layering-multiple} to get a proof of $k$ vertex-disjoint paths from $t_i$ to every vertex in $V-\set{t_i}$ for all $i \in [\ell]$.
This gives a $[kn+h,v]$-scheme.

We now prove the correctness of this scheme.
If $\card{X}\geq 2k$, then we can use the proof of \Cref{lem:layering-2k-vconn}.
Thus, $\card{X}< 2k$ implying $V-X$ has $\Omega(n)$ vertices (assuming $k< 0.49n$).
This means that with high probability, there is an $i^* \in[\ell]$ such that $t_{i^*} \in V-X$.
\Cref{clm:k-vconn-proof-simple} shows that the graph is $k$-vertex-connected.
This concludes the proof, and we have shown that the graph is $k$-vertex-connected using a $[kn+h,v]$-scheme for any $h,v$ such that $h\cdot v=n^2$.
\end{proof}




\subsection{Edge Connectivity Schemes}

In this section, we discuss algorithms for $k$-edge-connectivity.

\econndyn*
\edgeconndyn*

\subsubsection{Showing the graph is not k-edge-connected}

Consider the case where the edge connectivity is $<k$.
The prover sends the $k-1$ edges whose deletion disconnects the graph and also sends one side $S$ of this cut.
The verifier first checks if the edges the prover sent were indeed part of the stream by running an $[h,v]$-scheme for subset check (\Cref{fact:subsetdisjchk}) for any $h,v$ such that $h,v=n^2$.
The verifier also has to check that there are no other edges between $S$ and $T:=V-S$.
So the only allowed edges are within $S$, within $T$ or the $k-1$ edges sent by the prover. The verifier constructs a graph $G'$ in a streaming fashion with all such edges and then checks if the edges in input graph are a subset of $G'$.
This can be done using an $[h,v]$-scheme for subset check for any $h,v$ such that $h \cdot v=n^2$ (\Cref{fact:subsetdisjchk}).
This gives an $[n+h,v]$-scheme.

\subsubsection{\texorpdfstring{$[n,n]$}{[n,n]}-scheme for k-edge-connectivity}
We now give the algorithm for when the edge connectivity is $\geq k$.
We give a scheme independent of $k$ that works for any $k$ (this is basically solving minimum-cut).
The idea is to simulate the 2-pass streaming algorithm for minimum-cut implied by \cite{rsw18}.

\begin{proof}[Proof of \Cref{thm:edge-conn-dyn}]
The verifier when looking at the stream computes a $1+\eps$ cut sparsifier of the graph for $\eps=0.01$ and stores all vertex degrees.
This can be done in space $\Ot(n/\eps^2) = \Ot(n)$ \cite{AhnGM12}.

The verifier then finds all the $1+\eps$ approximate mincuts using the cut sparsifier.
Compress all vertices that are on the same side of all of these cuts into supernodes.
This graph of supernodes preserves all small cuts.
The prover then sends the edges that go between the supernodes.
\cite{rsw18} proves that there are only $O(n)$ edges between the supernodes.
So the verifier can store all of these edges and compute the minimum cut out of all these cuts. 
The verifier compares the size of this minimum-cut to the minimum degree and outputs the smaller of the two thus giving the exact mincut.

The verifier also needs to run a subset check (\Cref{fact:subsetdisjchk}) to check whether the edges the prover sent belong to the graph. This can be done using an $[n,n]$-scheme.
Note that, the prover does not have to send all edges; just enough to make all cuts of size $\geq k$ so it is okay if the prover does not send all edges (thus a subset check is enough).
\end{proof}

\subsubsection{Layering Algorithm}
If the graph $G=(V,E)$ is not $k$-edge-connected, then we can use the proof of the simple algorithm which has an $[n+ h,v]$-scheme.

\begin{proof}[Proof of \Cref{thm:e-conn-dyn}]
Fix an arbitrary vertex $t$.
Use \Cref{clm:layering-verification-econn} to get a proof of $k$ edge-disjoint paths from $t$ to every vertex in $V-\set{t}$.
Thus, the proof size is $\Ot(h+k^2n)$ words and can be verified using $\Ot(v)$ space for any $h,v$ such that $h\cdot v=n^2$.
This is a valid proof which shows that $G$ is $k$-edge-connected. 
Say $G$ has a cut $S$ of size at most $k-1$. Let $T:=V-S$ be the other side of the cut containing $t$ (wlog). There is a vertex $s\in S$ that has $k$ edge-disjoint paths to $t$ giving a contradiction.
Thus, we get a $[h+k^2n,v]$-scheme for $k$-edge-connectivity for any $h,v$ such that $h\cdot v=n^2$.
\end{proof}

\section{Annotated Streaming Schemes for SGT Streams}\label{sec:Annotated-SGT}

The $k$-connectivity algorithms for dynamic streams do not work for SGT streams because the verification fails when the input graph can have large frequencies for the edges. Also, the frequencies can be negative. So the proof using the layering lemma remains the same but the auxiliary information and verification changes.
We prove the following algorithm for $k$-connectivity in SGT streams with parameter $\val$:
\begin{lemma}
    Under SGT streams with parameter $\val$, there exists a $[n^2\log \val + k^2 n, 1]$-scheme for $k$-vertex-connectivity and $k$-edge-connectivity.  
\end{lemma}
\begin{proof}
For $k$-vertex-connectivity fix $k$ arbitrary vertices $t_1,t_2,\ldots,t_k$.
Use \Cref{corr:layering-multiple} to get a proof of $k$ vertex-disjoint paths from $t_i$ to every vertex in $V-\set{t_i}$ for all $i \in [k]$.
For $k$-edge-connectivity fix a vertex $t$ and use \Cref{lem:layering:econn} to et a proof for $k$ edge-disjoint paths from $t$ to every vertex in $V-\set{t}$.
Thus, the proof sizes are $\Ot(k^2 n)$ 
words.

We now describe how the verification works.
There are three things that need to be verified.
First, all the edges that the prover sends should be a part of the input graph.
Second, the edges sent by the prover should form paths. 
Third, the paths for a vertex should be vertex (edge) disjoint.

We can verify the second and third case like we did for dynamic streams.
The prover will send the edges of the paths one path at a time, so it is easy to verify that the edges form a path.
To verify the disjoint paths the prover sends the vertices (edges) he is going to use in the proof in sorted order and then sends the proof.
This only increases the proof size by a constant factor.
It is easy to verify that the vertices (edges) arriving in sorted order are distinct in $\Ot(1)$ space.
The only thing left to do is check if the vertices (edges) sent upfront are the same as the ones used in the proof of disjoint paths.
We do this by running an equality check using the $\ell_0$-sampler trick using $\Ot(1)$ space (\Cref{fact:equaldet}).
When we get the vertices upfront we add them to an $\ell_0$-sampler and when we see them in the proof we subtract them from the sampler. If there is nothing left in the sampler in the end, then the two multi-sets of vertices are identical. Since the first one had distinct elements, so did the second one. The blow-up in proof size is at most a factor of $2$ because the auxiliary information just contains the vertices used in the proof.
Also, the space on the verifier side is $\Ot(1)$ words for each vertex, but this space can be reused and thus, it fits in our budget of $\Ot(1)$ words.

We now come to the hardest part which is verifying that the edges used in the proof indeed belong to the input graph.
The prover initially sends a list of all edges of the input graph along with their multiplicities sorted by sign (first sends all the positive frequency edges and then the negative frequency edges).
For technical reasons, the prover scales the multiplicities by $n^2$.
It is easy to do a multi-set equality check for these edges with multiplicities sent upfront and the edges from the input stream using the $\ell_0$-sampler trick in $\Ot(1)$ verification space (\Cref{fact:equaldet}).

In parallel, we create a sampler for the positive edges and another for the negative edges and insert the edges the prover sent upfront into the appropriate sampler according to the sign.
During the proof, we ask the prover to also mention the sign of each edge and we subtract it from the appropriate sampler. This increases the proof size by a constant factor.
After the proof, the prover sends a list of all the remaining edges i.e. edges not used in the proof along with their residual multiplicities (the residual multiplicities take the scaling by $n^2$ into account). 
We use $\ell_0$-sampler trick (\Cref{fact:equaldet}) to check the equality of the multi-sets, the set of positive edges send upfront and the positive edges used in the proof plus the residual positive edges using $\Ot(1)$ space.
We do the same for the negative edges.

The correctness is as follows. The prover cannot lie about the signs or multiplicities when the edges are sent upfront because we run an equality check with the input stream.
If the prover ever lies about a sign during the proof (or in while sending the residual edges) then the edge gets subtracted from the wrong sampler. This sampler has a negative frequency on the coordinate corresponding to that edge which can only decrease later (because we are only subtracting edges) implying that the sampler check will fail. 

If the prover lies about the multiplicities in the residual edges then the equality check will fail.
Thus, we just have two separate equality check problems. The scaling by $n^2$ ensures that an edge can be used multiple times during the proof (it can be used at most $kn \leq n^2$ times). Ignoring the residual multiplicities we are just performing a subset check.

The space used by the verifier is just $\Ot(1)$ words. The proof size, ignoring the edges sent upfront and the residual edges, is $\Ot(k^2 n)$. The edges sent upfront and the residual edges take size $O(n^2 \log (n^2 \val))=\Ot(n^2 \log \val)$.
The probability of failure is at most $1/\poly(n)$ (by union bounding over the polynomially many failure probabilities of $1/\poly(n)$).
Therefore, the proof along with the auxiliary information takes size $\Ot(n^2 \log \val + k^2 n)$ and can be verified in $\Ot(1)$ space, proving the lemma.
\end{proof}


\subsection*{Acknowledgements}
We are extremely grateful to Sepehr Assadi for many helpful conversations 
regarding the project. Prantar Ghosh would also like to thank Amit Chakrabarti and Justin Thaler for insightful discussions. Finally, we thank the anonymous reviewers of ITCS 2024 for their many detailed comments and suggestions that helped with improving the presentation of the paper.


\clearpage

\bibliographystyle{alpha}
\bibliography{new, refs}

\clearpage

\appendix

\end{document}